\documentclass[journal, final, 11pt]{IEEEtran}
\onecolumn
\usepackage{cite}
\usepackage{amsmath}
\usepackage{graphicx,amsfonts,amsthm, amssymb,bm,url,color,latexsym}
\usepackage{subcaption}
\usepackage{bbm}
\usepackage{microtype}
\usepackage{algorithm}
\usepackage{algorithmicx}
\usepackage{algpseudocode} 
\usepackage{hyperref}
\hypersetup{
    colorlinks=true,%
    citecolor=blue,%
    filecolor=blue,%
    linkcolor=blue,%
    urlcolor=blue
}
\usepackage{verbatim}
\usepackage{framed}

\newtheorem{theorem}{Theorem}[section]
\newtheorem{lemma}[theorem]{Lemma}
\newtheorem{corollary}[theorem]{Corollary}
\newtheorem{proposition}[theorem]{Proposition}

\usepackage{tikz}
\usepackage{fge}

\renewcommand{\mathbf}{\boldsymbol} 

\newcommand{\mb}{\mathbf}
\newcommand{\mc}{\mathcal}

\newcommand{\bb}{\mathbb}
\newcommand{\magnitude}[1]{ \left| #1 \right| } 
\newcommand{\set}[1]{\left\{ #1 \right\}}

\newcommand{\reals}{\bb R}

\newcommand{\eps}{\varepsilon}
\newcommand{\R}{\reals}

\newcommand{\indicator}[1]{\mathbbm 1_{#1}}

\newcommand{ \Brac }[1]{\left\lbrace #1 \right\rbrace}
\newcommand{ \brac }[1]{\left[ #1 \right]}
\newcommand{ \paren }[1]{ \left( #1 \right) }

\DeclareMathOperator{\supp}{supp}

\DeclareMathOperator{\diag}{diag}
\DeclareMathOperator{\sign}{sign}

\DeclareMathOperator{\mini}{minimize}

\DeclareMathOperator{\st}{subject\; to}

\numberwithin{equation}{section}

\newcommand{\event}{\mc E}

\newcommand{\rconcave}{r_\fgecap}
\newcommand{\Lconcave}{L_\fgecap}
\newcommand{\rconvex}{r_\fgecup}

\newcommand{\Lconvex}{L_\fgecup}

\newcommand{\wh}{\widehat}
\newcommand{\wt}{\widetilde}
\newcommand{\ol}{\overline}

\newcommand{\norm}[2]{\left\| #1 \right\|_{#2}}
\newcommand{\abs}[1]{\left| #1 \right|}
\newcommand{\innerprod}[2]{\left\langle #1,  #2 \right\rangle}
\newcommand{\prob}[1]{\bb P\left[ #1 \right]}
\newcommand{\expect}[1]{\bb E\left[ #1 \right]}

\newcommand{\js}[1]{#1}

\begin{document}
\title{Complete Dictionary Recovery over the Sphere \\ I: Overview and the Geometric Picture}
\author{Ju~Sun,~\IEEEmembership{Student Member,~IEEE,}
        Qing~Qu,~\IEEEmembership{Student Member,~IEEE,}
        and~John~Wright,~\IEEEmembership{Member,~IEEE}
\thanks{JS, QQ, and JW are all with Electrical Engineering, Columbia University, New York, NY 10027, USA. Email: \{js4038, qq2105, jw2966\}@columbia.edu. An extended abstract of the current work has been published in~\cite{sun2015complete_conf}. Proofs of some secondary results are contained in the combined technical report~\cite{sun2015complete_tr}. }
\thanks{Manuscript received xxx; revised xxx.}}

\markboth{IEEE Transaction on Information Theory,~Vol.~xx, No.~xx, xxxx~2016}%
{Sun \MakeLowercase{\textit{et al.}}: Complete Dictionary Recovery over the Sphere}



\maketitle

\begin{abstract}
We consider the problem of recovering a complete (i.e., square and invertible) matrix $\mb A_0$, from $\mb Y \in \R^{n \times p}$ with $\mb Y = \mb A_0 \mb X_0$, provided $\mb X_0$ is sufficiently sparse. This recovery problem is central to theoretical understanding of dictionary learning, which seeks a sparse representation for a collection of input signals and finds numerous applications in modern signal processing and machine learning. We give the first efficient algorithm that provably recovers $\mb A_0$ when $\mb X_0$ has $O\paren{n}$ nonzeros per column, under suitable probability model for $\mb X_0$. In contrast, prior results based on efficient algorithms \js{either only guarantee recovery when $\mb X_0$ has $O(\sqrt{n})$ zeros per column, or require multiple rounds of SDP relaxation to work when $\mb X_0$ has $O(n^{1-\delta})$ nonzeros per column (for any constant $\delta \in (0, 1)$). }

Our algorithmic pipeline centers around solving a certain nonconvex optimization problem with a spherical constraint. In this paper, we provide a geometric characterization of the objective landscape. In particular, we show that the problem is highly structured: with high probability, (1) there are no ``spurious'' local minimizers; and (2) around all saddle points the objective has a negative directional curvature. This distinctive structure makes the problem amenable to efficient optimization algorithms. In a companion paper~\cite{sun2015complete_b}, we design a second-order trust-region algorithm over the sphere that provably converges to a local minimizer from arbitrary initializations, despite the presence of saddle points.
\end{abstract}


%
\IEEEpeerreviewmaketitle


\section{Introduction}
Given $p$ signal samples from $\R^n$, i.e., $\mb Y \doteq \brac{\mb y_1, \dots, \mb y_p}$, is it possible to construct a ``dictionary'' $\mb A \doteq \brac{\mb a_1, \dots, \mb a_m}$ with $m$ much smaller than $p$, such that $\mb Y \approx \mb A \mb X$ and the coefficient matrix $\mb X$ has as few nonzeros as possible? In other words, this model \emph{dictionary learning} (DL) problem seeks a concise representation for a collection of input signals. Concise signal representations play a central role in compression, and also prove useful to many other important tasks, such as signal acquisition, denoising, and classification. 

Traditionally, concise signal representations have relied heavily on explicit analytic bases constructed in nonlinear approximation and harmonic analysis. This constructive approach has proved highly successful; the numerous theoretical advances in these fields (see, e.g., ~\cite{devore1998nonlinear, temlyakov2003nonlinear, devore2009nonlinear, candes2002new, ma2010review} for summary of relevant results) provide ever more powerful representations, ranging from the classic Fourier basis to modern multidimensional, multidirectional, multiresolution bases, including wavelets, curvelets, ridgelets, and so on. However, two challenges confront practitioners in adapting these results to new domains: which function class best describes signals at hand, and consequently which representation is most appropriate. These challenges are coupled, as function classes with known ``good'' analytic bases are rare. \footnote{As Donoho et al~\cite{donoho1998data} put it, ``...in effect, uncovering the optimal codebook structure of naturally occurring data involves more challenging empirical questions than any that have ever been solved in empirical work in the mathematical sciences.''}

Around 1996, neuroscientists Olshausen and Field discovered that sparse coding, the principle of encoding a signal with few atoms from a learned dictionary, reproduces important properties of the receptive fields of the simple cells that perform early visual processing~\cite{olshausen1996emergence, olshausen1997sparse}. The discovery has spurred a flurry of algorithmic developments and successful applications for DL in the past two decades, spanning classical image processing, visual recognition, compressive signal acquisition, and also recent deep architectures for signal classification (see, e.g., \cite{elad2010sparse, mairal2014sparse} for review of this development). 

The learning approach is particularly relevant to modern signal processing and machine learning, which deal with data of huge volume and great variety (e.g., images, audios, graphs, texts, genome sequences, time series, etc). The proliferation of problems and data seems to preclude analytically deriving optimal representations for each new class of data in a timely manner. On the other hand, as datasets grow, learning dictionaries directly from data looks increasingly attractive and promising. When armed with sufficiently many data samples of one signal class, by solving the model DL problem, one would expect to obtain a dictionary that allows sparse representation for the whole class. This hope has been borne out in a number of successful examples~\cite{elad2010sparse, mairal2014sparse} and theories~\cite{maurer2010dimensional, Vainsencher:2011, Mehta13, gribonval2013sample}. 

\subsection{Theoretical and Algorithmic Challenges}

In contrast to the above empirical successes, theoretical study of dictionary learning is still developing. For applications in which dictionary learning is to be applied in a ``hands-free'' manner, it is desirable to have efficient algorithms which are guaranteed to perform correctly, when the input data admit a sparse model. There have been several important recent results in this direction, which we will review in Section \ref{sec:lit_review}, after our sketching main results. Nevertheless, obtaining algorithms that provably succeed under broad and realistic conditions remains an important research challenge. 

To understand where the difficulties arise, we can consider a model formulation, in which we attempt to obtain the dictionary $\mb A$ and coefficients $\mb X$ which best trade-off sparsity and fidelity to the observed data:
\begin{align} \label{eq:dl_concept}
\mini_{\mb A \in \R^{n \times m}, \mb X \in \R^{m \times p}}\; \lambda \norm{\mb X}{1} + \frac{1}{2} \norm{\mb A \mb X - \mb Y}{F}^2, \; \st \;  \mb A \in \mc A.  
\end{align}
Here, $\norm{\mb X}{1} \doteq \sum_{i, j} \abs{X_{ij}}$ promotes sparsity of the coefficients, $\lambda \ge 0$ trades off the level of coefficient sparsity and quality of approximation, and $\mc A$ imposes desired structures on the dictionary.

This formulation is nonconvex: the admissible set $\mc A$ is typically nonconvex (e.g., orthogonal group, matrices with normalized columns)\footnote{For example, in nonlinear approximation and harmonic analysis, orthonormal basis or (tight-)frames are preferred; to fix the scale ambiguity discussed in the text, a common practice is to require that $\mb A$ to be column-normalized. 
}, while the most daunting nonconvexity comes from the bilinear mapping: $\paren{\mb A, \mb X} \mapsto \mb A \mb X$. Because $\paren{\mb A, \mb X}$ and $\paren{\mb A \mb \Pi \mb \Sigma, \mb \Sigma^{-1} \mb \Pi^* \mb X}$ result in the same objective value for the conceptual formulation~\eqref{eq:dl_concept}, where $\mb \Pi$ is any permutation matrix, and $\mb \Sigma$ any diagonal matrix with diagonal entries in $\{ \pm 1 \}$, and $\paren{\cdot}^*$ denotes matrix transpose. Thus, we should expect the problem to have combinatorially many global minimizers. These global minimizers are generally isolated, likely jeopardizing natural convex relaxation (see similar discussions in, e.g.,~\cite{gribonval2010dictionary} and~\cite{geng2011local}).\footnote{\js{Simple convex relaxations normally replace the objective function with a convex surrogate, and the constraint set with its convex hull. When there are multiple isolated global minimizers for the original nonconvex problem, any point in the convex hull of these global minimizers are necessarily feasible for the relaxed version, and such points tend to produce smaller or equal values than that of the original global minimizers by the relaxed objective function, due to convexity. This implies such relaxations are bound to be loose. } Semidefinite programming (SDP) lifting may be one useful general strategy to convexify bilinear inverse problems, see, e.g., \cite{ahmed2014blind, choudhary2014identifiability}. However, for problems with general nonlinear constraints, it is unclear whether the lifting always yields tight relaxation; consider, e.g.,~\cite{bandeira2013approximating, briet2014tight,choudhary2014identifiability} and the identification issue in blind deconvolution~\cite{li2015identifiability,kech2016optimal}. } This contrasts sharply with problems in sparse recovery and compressed sensing, in which simple convex relaxations are often provably effective
\cite{donoho2009observed, oymak2010new, candes2011robust, donoho2013phase, mccoy2014sharp, mu2013square, chandrasekaran2012convex, candes2013phaselift, amelunxen2014living, candes2014mathematics}. Is there any hope to obtain global solutions to the DL problem? 

\subsection{An Intriguing Numerical Experiment with Real Images} \label{sec:intro_exp}
\begin{figure}[t]
\centering  
\begin{subfigure}[t]{0.3\textwidth}
\centering
\includegraphics[width = 0.9\linewidth]{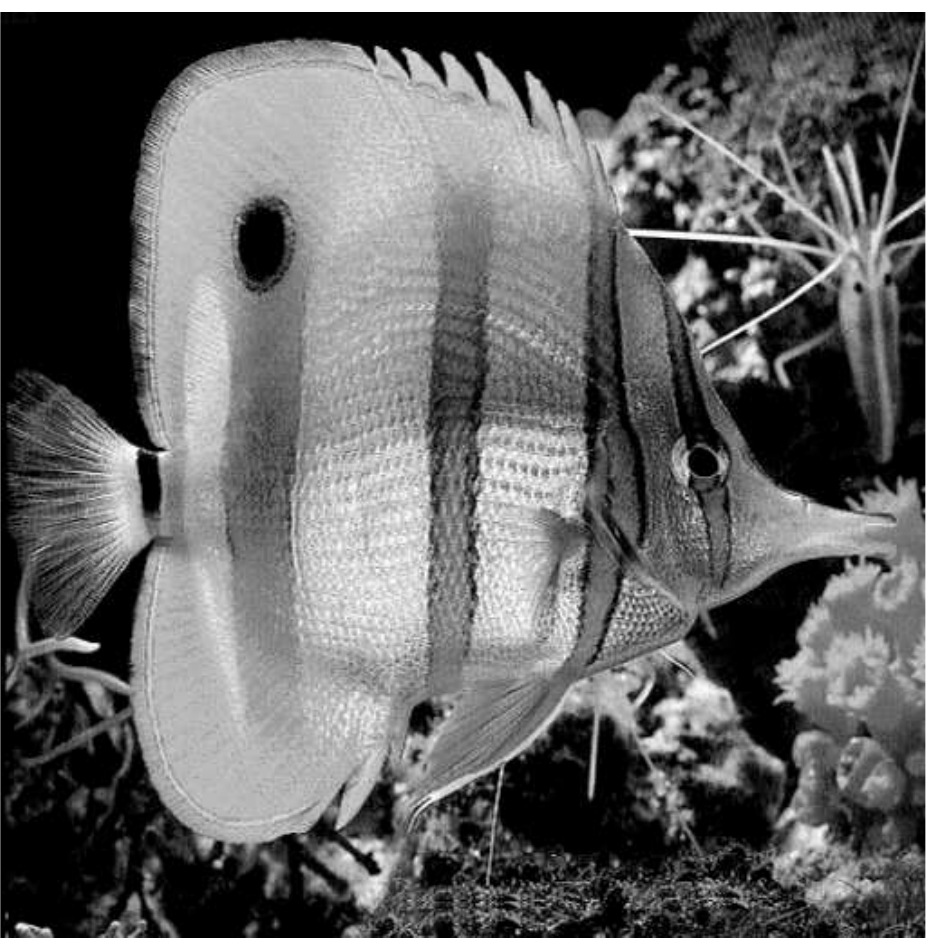} \\
\vspace{2mm}
\includegraphics[width = 0.9\linewidth]{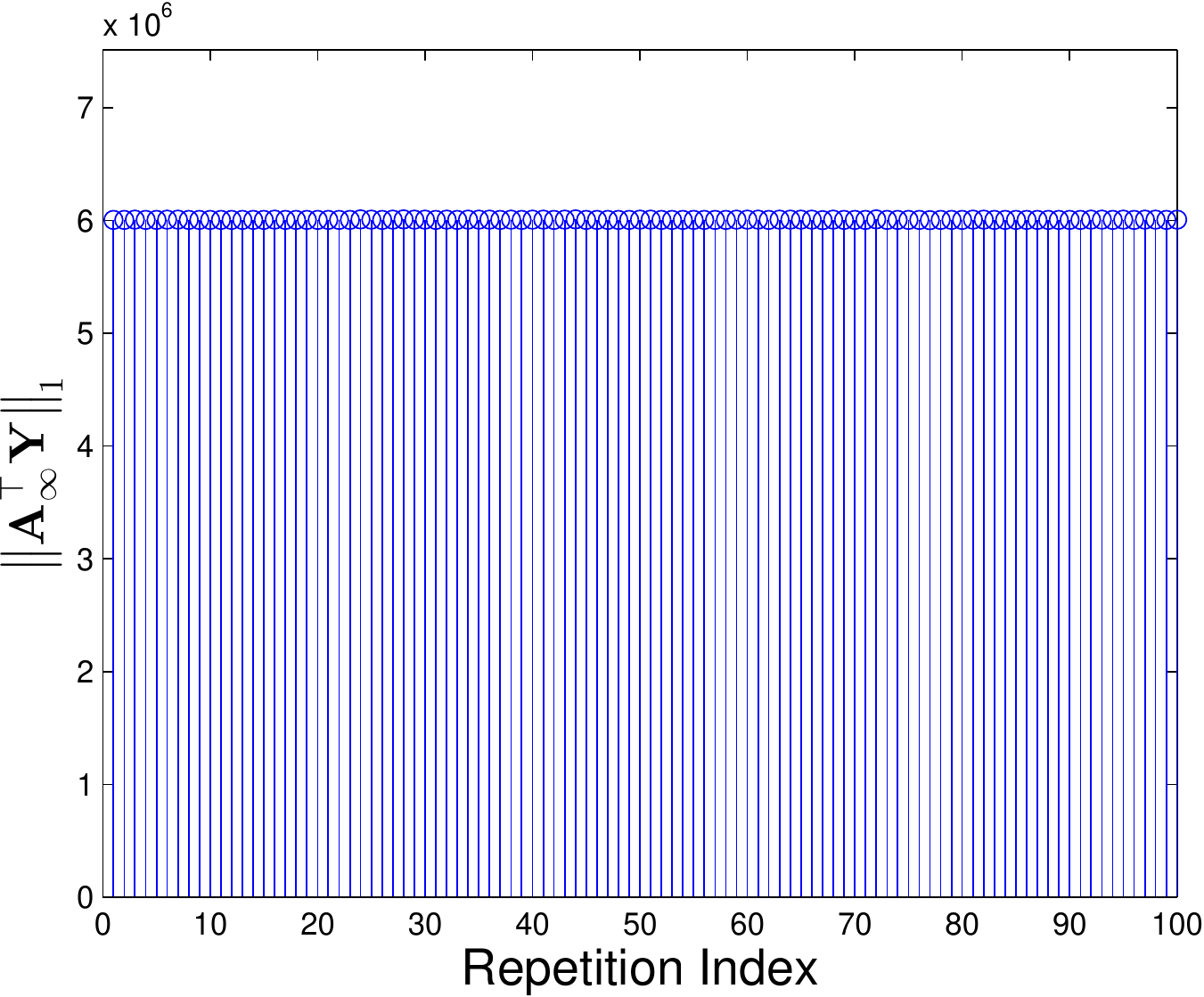} 
\end{subfigure}
\begin{subfigure}[t]{0.3\textwidth}
\centering
\includegraphics[width = 0.9\linewidth]{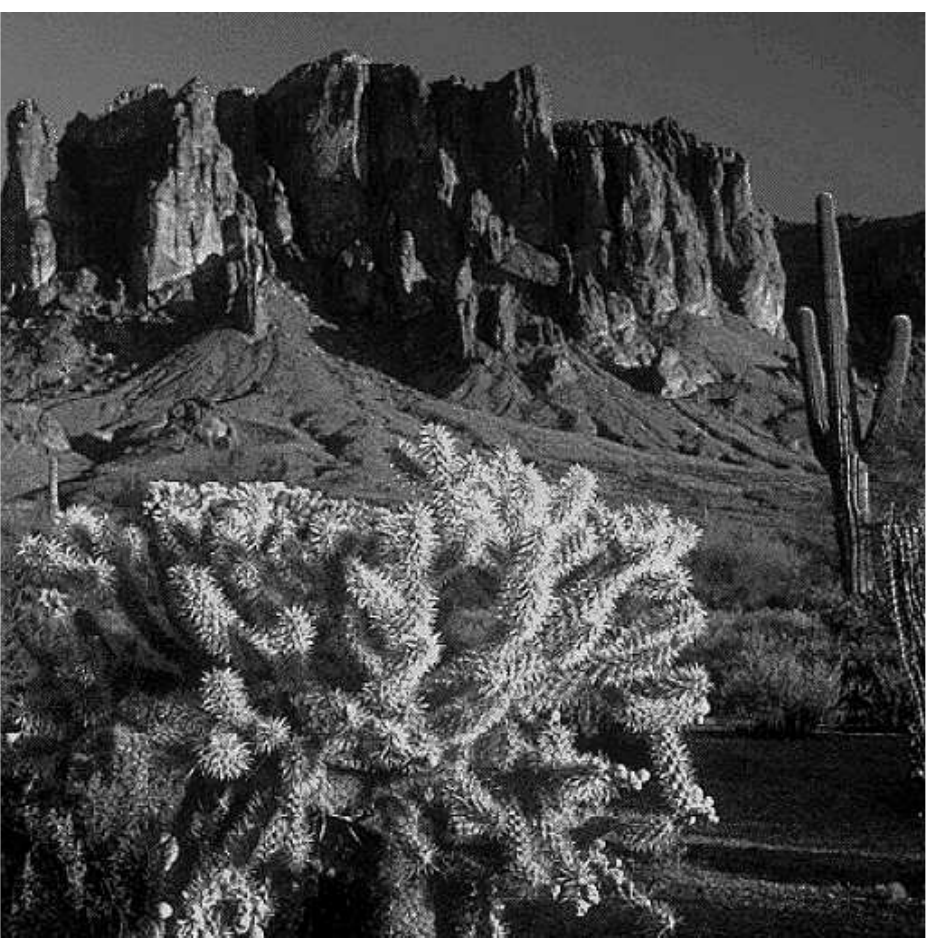} \\
\vspace{2mm}
\includegraphics[width = 0.9\linewidth]{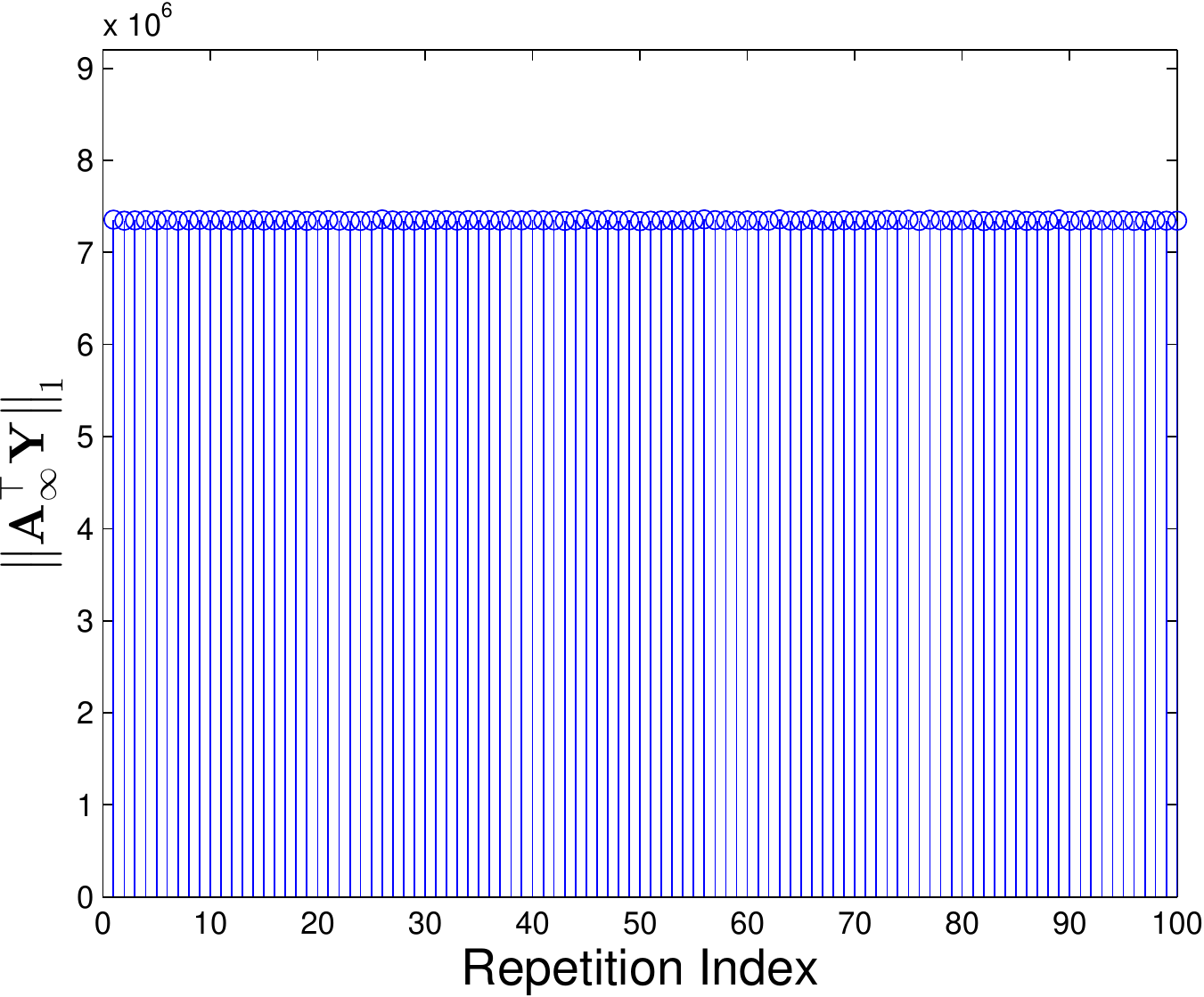}
\end{subfigure}
\begin{subfigure}[t]{0.3\textwidth}
\centering
\includegraphics[width = 0.9\linewidth]{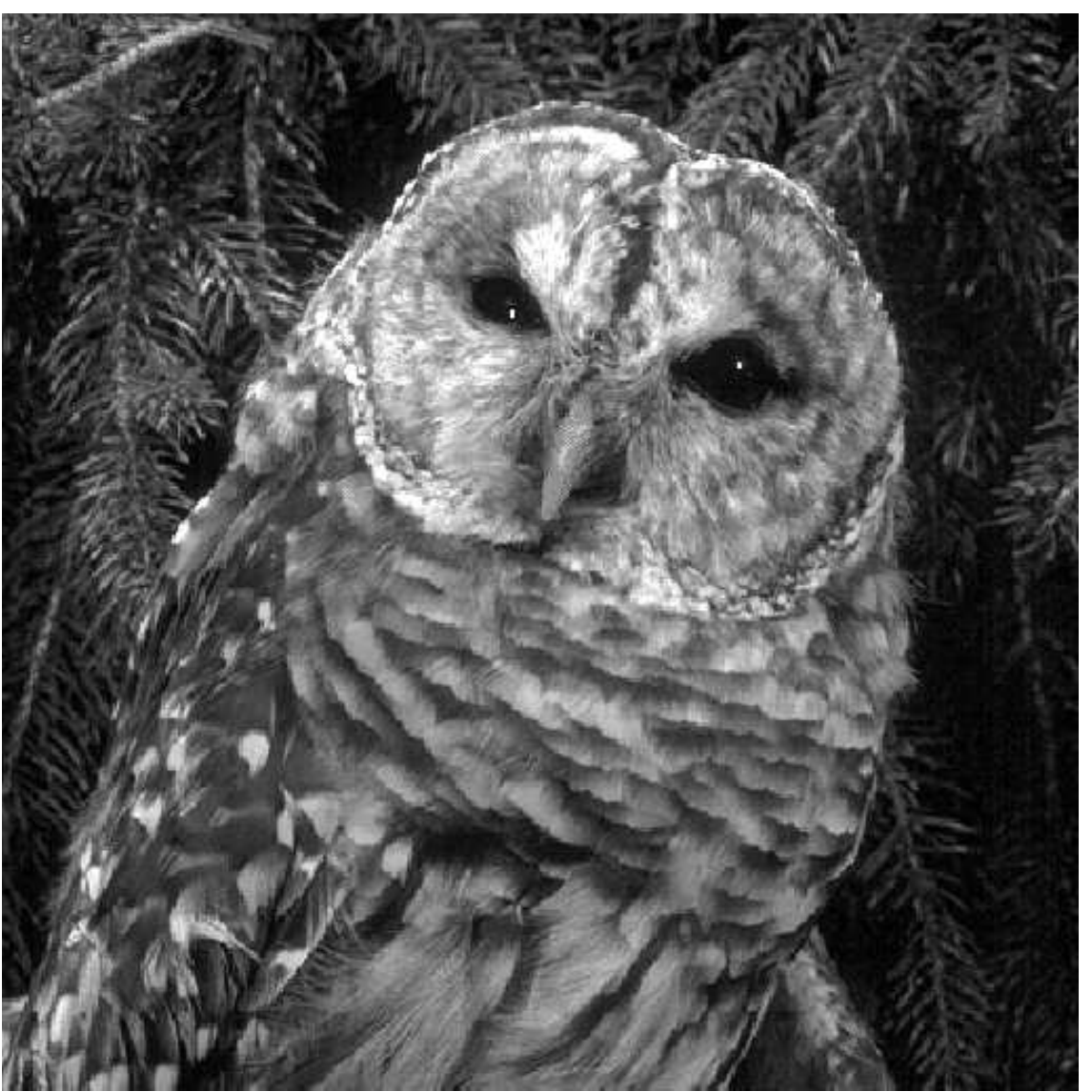} \\
\vspace{2mm}
\includegraphics[width = 0.9\linewidth]{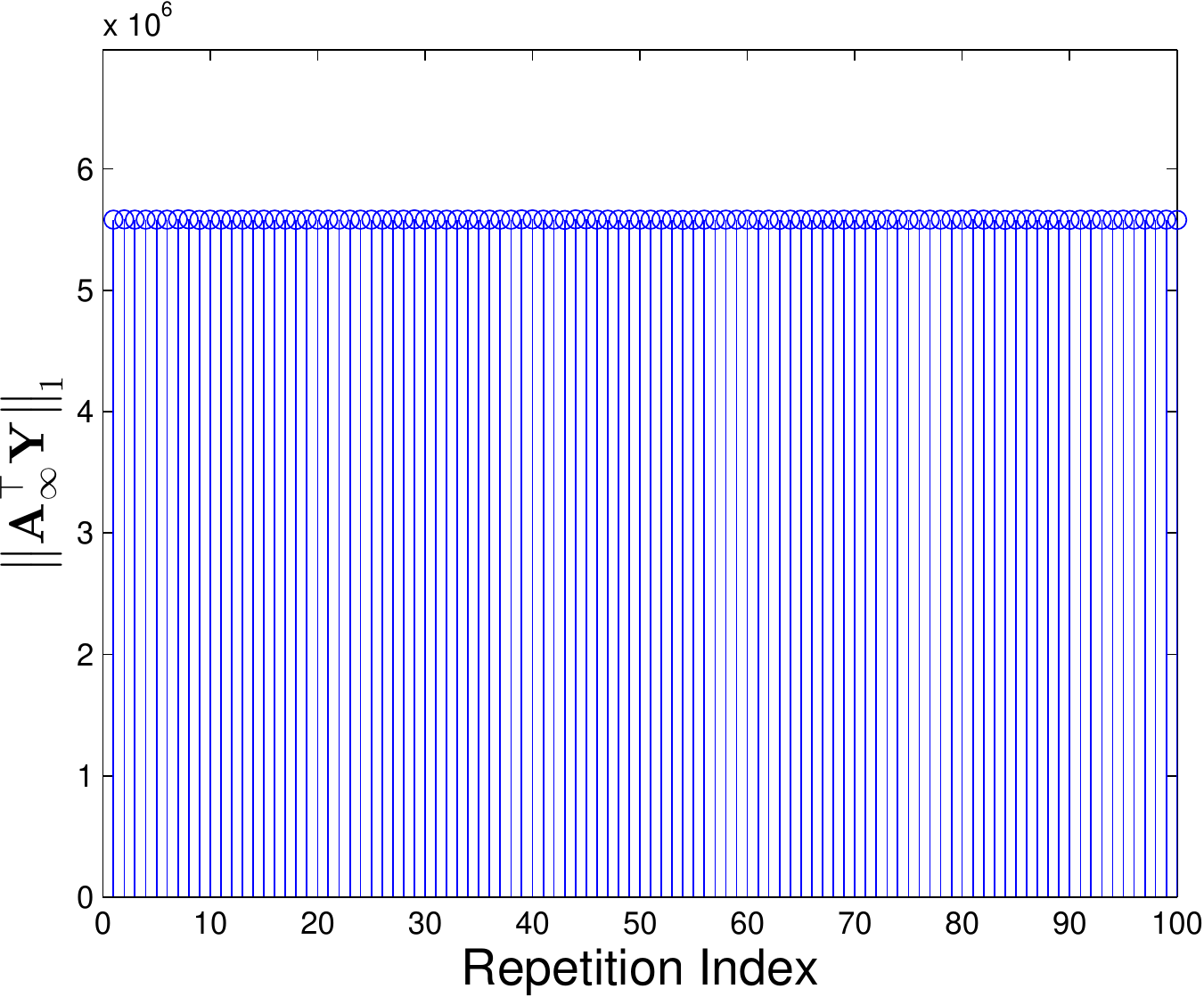}
\end{subfigure}
\caption{\textbf{Alternating direction method for~\eqref{eq:p_l1p} on uncompressed real images seems to always produce the same solution!} \textbf{Top}: Each image is $512 \times 512$ in resolution and encoded in the uncompressed \texttt{pgm} format (uncompressed images to prevent possible bias towards standard bases used for compression, such as DCT or wavelet bases). Each image is evenly divided into $8 \times 8$ non-overlapping image patches ($4096$ in total), and these patches are all vectorized and then stacked as columns of the data matrix $\mb Y$. \textbf{Bottom}: Given each $\mb Y$, we solve~\eqref{eq:p_l1p} $100$ times with independent and randomized (uniform over the orthogonal group) initialization $\mb A_0$. \js{Let $\mb A_\infty$ denote the value of $\mb A$ at convergence (we set the maximally allowable number of ADM iterations to be $10^4$ and $\lambda = 2$). } The plots show the values of $\norm{\mb A^*_{\infty} \mb Y}{1}$ across the independent repetitions. They are virtually the same and the relative differences are less than $10^{-3}$! }
\label{fig:odl_examples}
\end{figure}

We provide empirical evidence in support of a positive answer to the above question. Specifically, we learn orthogonal bases (orthobases) for real images patches. Orthobases are of interest because typical hand-designed dictionaries such as discrete cosine (DCT) and wavelet bases are orthogonal, and orthobases seem competitive in performance for applications such as image denoising, as compared to overcomplete dictionaries~\cite{bao2013fast}\footnote{See Section~\ref{sec:intro_rec} for more detailed discussions of this point. \cite{lesage2005learning} also gave motivations and algorithms for learning (union of) orthobases as dictionaries. }. 

We divide a given greyscale image into $8 \times 8$ non-overlapping patches, which are converted into $64$-dimensional vectors and stacked column-wise into a data matrix $\mb Y$. Specializing~\eqref{eq:dl_concept} to this setting, we obtain the optimization problem: 
\begin{align} \label{eq:p_l1p}
\mini_{\mb A \in \R^{n \times n}, \mb X \in \R^{n \times p}}\; \lambda \norm{\mb X}{1} + \frac{1}{2}\norm{\mb A \mb X - \mb Y}{F}^2, \; \st \; \mb A \in O_n, 
\end{align}
where $O_n$ is the set of order $n$ orthogonal matrices, i.e., order-$n$ orthogonal group. To derive a concrete algorithm for~\eqref{eq:p_l1p}, one can deploy the alternating direction method (ADM)\footnote{This method is also called alternating minimization or (block) coordinate descent method. see, e.g., ~\cite{bertsekas1989parallel, tseng2001convergence} for classic results and~\cite{attouch2010proximal,bolte2014proximal} for several interesting recent developments. }, i.e., alternately minimizing the objective function with respect to (w.r.t.) one variable while fixing the other. The iteration sequence actually takes very simple form: for $k = 1, 2, 3, \dots$, 
\begin{align*}
\mb X_k = \mc S_{\lambda}\brac{\mb A^*_{k-1} \mb Y}, \qquad \mb A_k = \mb U \mb V^* \;\text{for} \; \mb U \mb D \mb V^* = \mathtt{SVD}\paren{\mb Y \mb X_k^*}
\end{align*}
where $\mc S_{\lambda}\brac{\cdot}$ denotes the well-known soft-thresholding operator acting elementwise on matrices, i.e.,  $\mc S_{\lambda}\brac{x} \doteq \sign\paren{x} \max\paren{\abs{x} - \lambda, 0}$ for any scalar $x$. 

Fig.~\ref{fig:odl_examples} shows what we obtained using the simple ADM algorithm, with \emph{independent and randomized initializations}: 
\begin{quote}
\centering
\emph{The algorithm seems to always produce the same optimal value, regardless of the initialization.}
\end{quote} 
This observation \js{is consistent with the possibility that the heuristic ADM algorithm may} \emph{always converge to a global minimizer}! \footnote{Technically, the convergence to global solutions is surprising because even convergence of ADM to critical points is not guaranteed in general, see, e.g., \cite{attouch2010proximal,bolte2014proximal} and references therein.} Equally surprising is that the phenomenon has been observed on real images\footnote{Actually the same phenomenon is also observed for simulated data when the coefficient matrix obeys the Bernoulli-Gaussian model, which is defined later. The result on real images supports that previously claimed empirical successes over two decades may be non-incidental. }. One may imagine only random data typically have ``favorable'' structures; in fact, almost all existing theories for DL pertain only to random data~\cite{spielman2012exact,agarwal2013learning,arora2013new,agarwal2013exact,arora2014more,arora2015simple}. 

\subsection{Dictionary Recovery and Our Results} \label{sec:intro_rec}
In this paper, we take a step towards explaining the surprising effectiveness of nonconvex optimization heuristics for DL. We focus on the \emph{dictionary recovery} (DR) setting: given a data matrix $\mb Y$ generated as $\mb Y = \mb A_0 \mb X_0$, where $\mb A_0 \in \mc A \subseteq \R^{n \times m}$ and $\mb X_0 \in \R^{m \times p}$ is ``reasonably sparse'', try to recover $\mb A_0$ and $\mb X_0$. Here recovery means to return any pair $\paren{\mb A_0 \mb \Pi \mb \Sigma, \mb \Sigma^{-1} \mb \Pi^* \mb X_0}$, where $\mb \Pi$ is a permutation matrix and $\mb \Sigma$ is a nonsingular diagonal matrix, i.e., recovering up to sign, scale, and permutation.

To define a reasonably simple and structured problem, we make the following assumptions: 
\begin{itemize}
\item The target dictionary $\mb A_0$ is complete, i.e., square and invertible ($m = n$). In particular, this class includes orthogonal dictionaries. Admittedly overcomplete dictionaries tend to be more powerful for modeling and to allow sparser representations. Nevertheless, most classic hand-designed dictionaries in common use are orthogonal. Orthobases are competitive in performance for certain tasks such as image denoising~\cite{bao2013fast}, and admit faster algorithms for learning and encoding. \footnote{Empirically, there is no systematic evidence supporting that overcomplete dictionaries are strictly necessary for good performance in all published applications (though~\cite{olshausen1997sparse} argues for the necessity from a neuroscience perspective). Some of the ideas and tools developed here for complete dictionaries may also apply to certain classes of structured overcomplete dictionaries, such as tight frames. See Section~\ref{sec:discuss} for relevant discussion. }  
\item The coefficient matrix $\mb X_0$ follows the Bernoulli-Gaussian (BG) model with rate $\theta$: $\brac{X_0}_{ij} = \Omega_{ij} V_{ij}$, with $\Omega_{ij} \sim \mathrm{Ber}\paren{\theta}$ and $V_{ij} \sim \mc N\paren{0, 1}$, where all the different random variables are jointly independent. We write compactly $\mb X_0 \sim_{i.i.d.} \mathrm{BG}\paren{\theta}$. \js{This BG model, or the Bernoulli-Subgaussian model as used in~\cite{spielman2012exact}, is a reasonable first model for generic sparse coefficients: the Bernoulli process enables explicit control on the (hard) sparsity level, and the (sub)-Gaussian process seems plausible for modeling variations in magnitudes. Real signals may admit encoding coefficients with additional or different characteristics. We will focus on generic sparse encoding coefficients as a first step towards theoretical understanding. }
\end{itemize}
In this paper, we provide a nonconvex formulation for the DR problem, and characterize the geometric structure of the formulation that allows development of efficient algorithms for optimization. In the companion paper~\cite{sun2015complete_b}, we derive an efficient algorithm taking advantage of the structure, and describe a complete algorithmic pipeline for efficient recovery. Together, we prove the following result: 
\begin{theorem}[Informal statement of our results, a detailed version included in the companion paper~\cite{sun2015complete_b}]
For any $\theta \in \paren{0, 1/3}$, given $\mb Y = \mb A_0 \mb X_0$ with $\mb A_0$ a complete dictionary and $\mb X_0 \sim_{i.i.d.} \mathrm{BG}\paren{\theta}$, there is a polynomial-time algorithm that recovers (up to sign, scale, and permutation) $\mb A_0$ and $\mb X_0$ with high probability (at least $1-O(p^{-6})$) whenever $p \ge p_{\star}\paren{n, 1/\theta, \kappa\paren{\mb A_0}, 1/\mu}$ for a fixed polynomial $p_\star\paren{\cdot}$, where $\kappa\paren{\mb A_0}$ is the condition number of $\mb A_0$ and $\mu$ is a parameter that can be set as $cn^{-5/4}$ for a constant $c > 0$. 
\end{theorem}
Obviously, even if $\mb X_0$ is known, one needs $p \ge n$ to make the identification problem well posed. Under our particular probabilistic model, a simple coupon collection argument implies that one needs $p \ge \Omega\paren{\tfrac{1}{\theta}\log n}$ to ensure all atoms in $\mb A_0$ are observed with high probability (w.h.p.). Ensuring that an efficient algorithm exists may demand more. Our result implies when $p$ is polynomial in $n$, $1/\theta$ and $\kappa(\mb A_0)$, recovery with an efficient algorithm is possible. 

The parameter $\theta$ controls the sparsity level of $\mb X_0$. Intuitively, the recovery problem is easy for small $\theta$ and becomes harder for large $\theta$.\footnote{Indeed, when $\theta$ is small enough such that columns of $\mb X_0$ are predominately $1$-sparse, one directly observes scaled versions of the atoms (i.e., columns of $\mb X_0$); when $\mb X_0$ is fully dense corresponding to $\theta = 1$, recovery is never possible as one can easily find another complete $\mb A_0'$ and fully dense $\mb X_0'$ such that $\mb Y = \mb A_0' \mb X_0'$ with $\mb A_0'$ not equivalent to $\mb A_0$. 
} It is perhaps surprising that an efficient algorithm can succeed up to constant $\theta$, i.e., linear sparsity in $\mb X_0$. Compared to the case when $\mb A_0$ is known, there is only at most a constant gap in the sparsity level one can deal with. 

For DL, our result gives the first efficient algorithm that provably recovers complete $\mb A_0$ and $\mb X_0$ when $\mb X_0$ has $O(n)$ nonzeros per column under appropriate probability model. Section~\ref{sec:lit_review} provides detailed comparison of our result with other recent recovery results for complete and overcomplete dictionaries. 


\subsection{Main Ingredients and Innovations}
In this section we describe three main ingredients that we use to obtain the stated result. 

\subsubsection{\textbf{A Nonconvex Formulation}}
Since $\mb Y = \mb A_0 \mb X_0$ and $\mb A_0$ is complete, $\mathrm{row}\paren{\mb Y} = \mathrm{row}\paren{\mb X_0}$ ($\mathrm{row}\paren{\cdot}$ denotes the row space of a matrix) and hence rows of $\mb X_0$ are sparse vectors in the known (linear) subspace $\mathrm{row}\paren{\mb Y}$. We can use this fact to first recover the rows of $\mb X_0$, and subsequently recover $\mb A_0$ by solving a system of linear equations. In fact, for $\mb X_0 \sim_{i.i.d.} \mathrm{BG}\paren{\theta}$, rows of $\mb X_0$ are the $n$ \emph{sparsest} vectors (directions) in $\mathrm{row}\paren{\mb Y}$ w.h.p. whenever $p \ge \Omega\paren{n\log n}$~\cite{spielman2012exact}. \js{Thus, recovering rows of $\mb X_0$ is equivalent to finding the sparsest vectors/directions (due to the scale ambiguity) in $\mathrm{row}(\mb Y)$. Since any vector in $\mathrm{row}(\mb Y)$ can be written as $\mb q^* \mb Y$ for a certain $\mb q$, one might try to solve}
\begin{align} \label{eq:l0_sparse_vec}
\mini\; \norm{\mb q^* \mb Y}{0}\; \; \st\;\; \js{\mb q^* \mb Y \neq \mb 0}
\end{align}
\js{to find the sparsest vector in $\mathrm{row}(\mb Y)$. Once the sparsest one is found, one then appropriately reduces the subspace $\mathrm{row}(\mb Y)$ by one dimension, and solves an analogous version of~\eqref{eq:l0_sparse_vec} to find the second sparsest vector. The process is continued recursively until all sparse vectors are obtained.  
The above idea of reducing the original recovery problem into finding sparsest vectors in a known subspace first appeared in~\cite{spielman2012exact}.} 

The objective is discontinuous, and the domain is an open set. In particular, the homogeneous constraint is unconventional and tricky to deal with. Since the recovery is up to scale, one can remove the homogeneity by fixing the scale of $\mb q$. Known relaxations~\cite{spielman2012exact, demanet2014scaling} fix the scale by setting $\norm{\mb q^* \mb Y}{\infty} = 1$ \js{and use $\norm{\cdot}{1}$ as a surrogate to $\norm{\cdot}{0}$}, where $\norm{\cdot}{\infty}$ is the elementwise $\ell^{\infty}$ norm, \js{leading to the optimization problem\begin{align} \label{eq:LP_form}
\mini\; \norm{\mb q^* \mb Y}{1} \; \; \st \; \; \norm{\mb q^* \mb Y}{\infty} = 1. 
\end{align} 
}
\js{The constraint means at least one coordinate of $\mb q^* \mb Y$ has unit magnitude\footnote{The sign ambiguity is tolerable here. }. } Thus, ~\eqref{eq:LP_form} reduces to a sequence of convex (linear) programs. \cite{spielman2012exact} has shown that (see also~\cite{demanet2014scaling}) solving~\eqref{eq:LP_form} recovers $\paren{\mb A_0, \mb X_0}$ for very sparse $\mb X_0$, but the idea provably breaks down when $\theta$ is slightly above $O(1/\sqrt{n})$, or equivalently when each column of $\mb X_0$ has more than $O\paren{\sqrt{n}}$ nonzeros. 

Inspired by our previous image experiment, we work with a \emph{nonconvex} alternative\footnote{A similar formulation has been proposed in~\cite{zibulevsky2001blind} in the context of blind source separation; see also~\cite{qu2014finding}. }:
\begin{align} \label{eq:main_l2}
\mini\;f(\mb q; \widehat{\mb Y}) \doteq \frac{1}{p} \sum_{k=1}^p h_{\mu}\paren{\mb q^* \widehat{\mb y}_k}, \; \st \; \norm{\mb q}{} = 1, 
\end{align}
where $\widehat{\mb Y} \in \R^{n\times p}$ is a proxy for $\mb Y$ (i.e., after appropriate processing), $k$ indexes columns of $\widehat{\mb Y}$, and $\norm{\cdot}{}$ is the usual $\ell^2$ norm for vectors. Here $h_{\mu}\paren{\cdot}$ is chosen to be a convex smooth approximation to $\abs{\cdot}$, namely,  
\begin{align} \label{eq:logexp}
h_{\mu}\paren{z} = \mu \log\paren{\frac{\exp\paren{z/\mu} + \exp\paren{-z/\mu}}{2}} = \mu \log \cosh(z/\mu), 
\end{align} 
which is infinitely differentiable and $\mu$ controls the smoothing level.\footnote{In fact, there is nothing special about this choice and we believe that any valid smooth (twice continuously differentiable) approximation to $\abs{\cdot}$ would work and yield qualitatively similar results. We also have some preliminary results showing the latter geometric picture remains the same for certain nonsmooth functions, such as a modified version of the Huber function, though the analysis involves handling a different set of technical subtleties. The algorithm also needs additional modifications.}
\begin{figure}[!htbp]
\centering
\includegraphics[width = 0.35\textwidth]{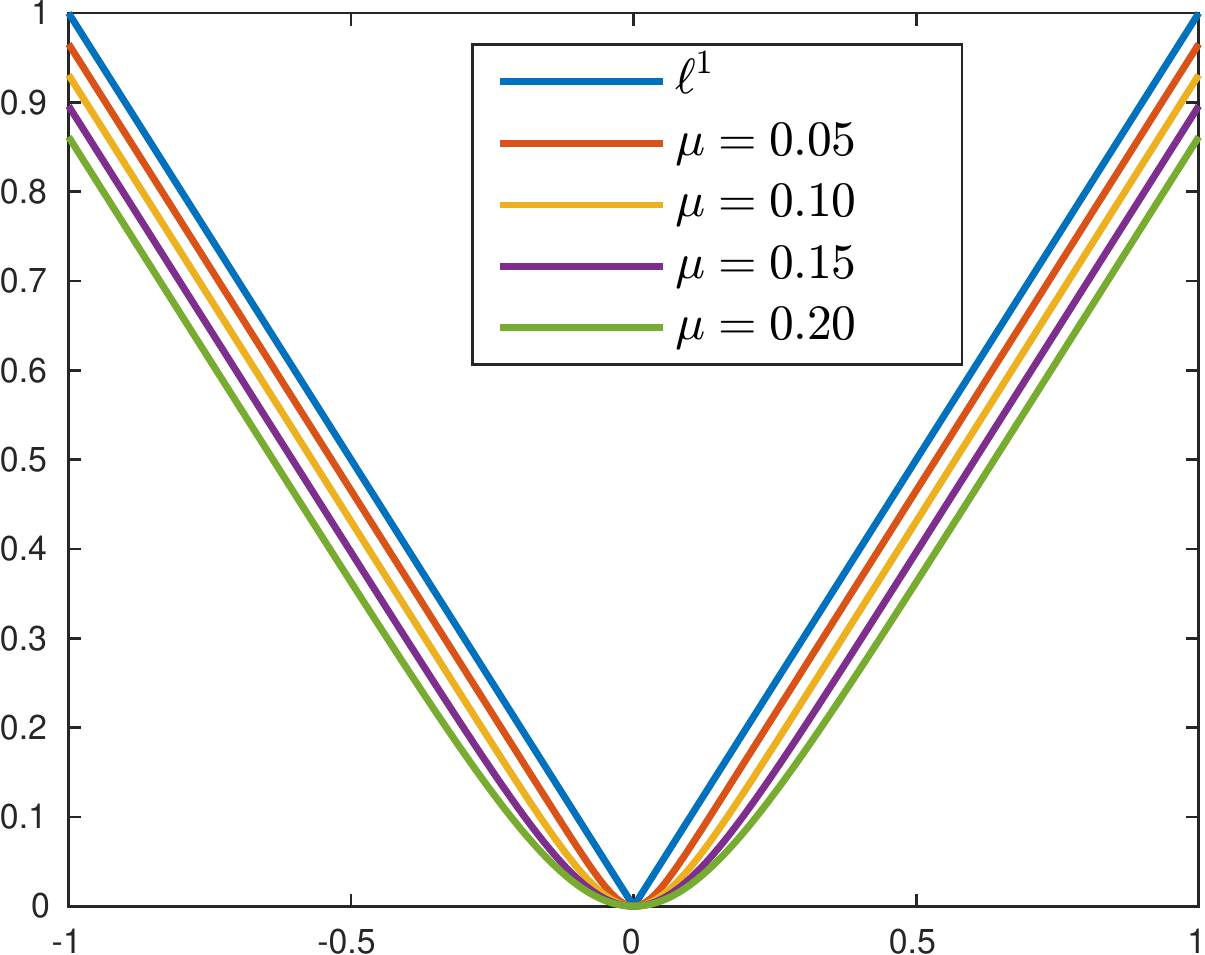}
\caption{The smooth $\ell^1$ surrogate defined in~\eqref{eq:logexp} vs. the $\ell^1$ function, for varying values of $\mu$. The surrogate approximates the $\ell^1$ function more closely when $\mu$ gets smaller.  }
\label{fig:hmu_plot}
\end{figure}
 \js{An illustration of the $h_{\mu}(\cdot)$ function vs. the $\ell^1$ function is provided in Fig.~\ref{fig:hmu_plot}. } The spherical constraint is nonconvex. Hence, a-priori, it is unclear whether \eqref{eq:main_l2} admits efficient algorithms that attain global optima. Surprisingly, simple descent algorithms for \eqref{eq:main_l2} exhibit very striking behavior: on many practical numerical examples\footnote{... not restricted to the model we assume here for $\mb A_0$ and $\mb X_0$. }, they appear to produce global solutions. Our next section will uncover interesting geometrical structures underlying the phenomenon. 

\subsubsection{\textbf{A Glimpse into High-dimensional Function Landscape}} \label{sec:overview_geometry}
\begin{figure}[t]
\centering
\includegraphics[width=0.3\textwidth]{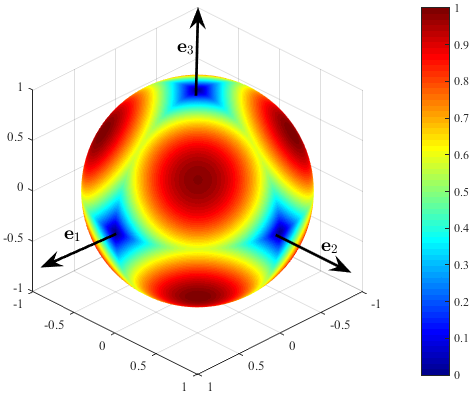}
\hspace{0.1in}
\includegraphics[width=0.3\textwidth]{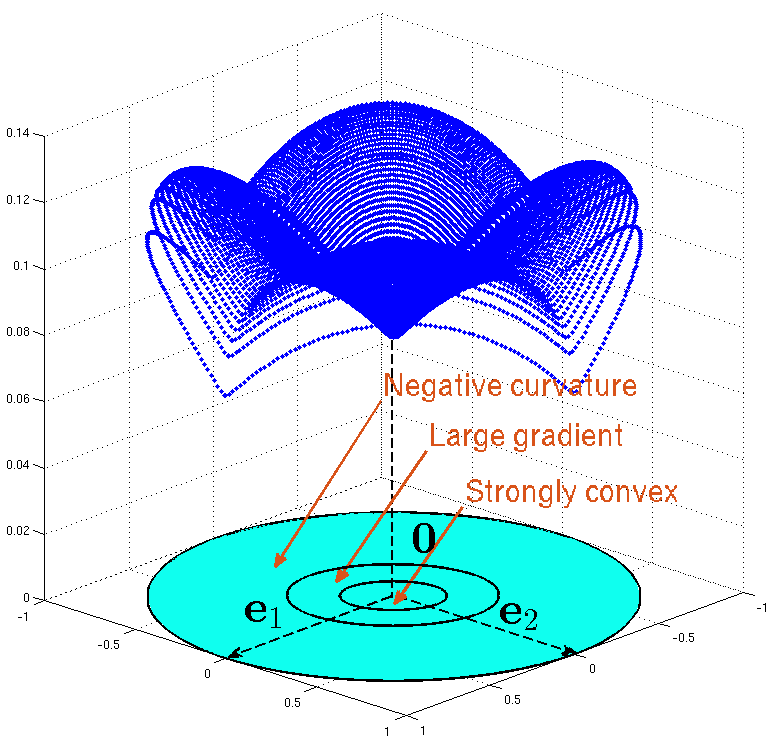} 
\caption{\textbf{Why is dictionary learning over $\bb S^{n-1}$ tractable?} \js{Assume the target dictionary $\mb A_0 = \mb I$. \textbf{Left:} Large sample objective function $\bb E_{\mb X_0}\brac{f\paren{\mb q}}$. The only local minimizers are the signed basis vectors $\pm \mb e_i$. \textbf{Right:} A visualization of the function as a height above the equatorial section $\mb e_3^\perp$, i.e., $\mathrm{span}\{\mb e_1, \mb e_2\} \cap \bb B^3$. The derived function is obtained by assigning values of points on the upper hemisphere to their corresponding projections on the equatorial section $\mb e_3^\perp$. The minimizers for the derived function are $\mb 0, \pm \mb e_1, \pm \mb e_2$. Around $\mb 0$ in $\mb e_3^\perp$, the function exhibits a small region of strong convexity, a region of large gradient, and finally a region in which the direction away from $\mb 0$ is a direction of negative curvature.}}
\label{fig:large-sample-sphere}
\end{figure}
For the moment, \js{suppose $\mb A_0 = \mb I$ and take $\widehat{\mb Y} = \mb Y = \mb A_0 \mb X_0 = \mb X_0$ in~\eqref{eq:main_l2}. Fig.~\ref{fig:large-sample-sphere} (left) plots $\bb E_{\mb X_0}\brac{f\paren{\mb q; \mb X_0}}$ over $\mb q \in \bb S^2$ ($n=3$). Remarkably, $\bb E_{\mb X_0}\brac{f\paren{\mb q; \mb X_0}}$ has no spurious local minimizers. In fact, every local minimizer $\wh{\mb q}$ is one of the signed standard basis vectors, i.e., $\pm \mb e_i$'s where $i \in \{1, 2, 3\}$. Hence, $\wh{\mb q}^* \mb Y$ reproduces a certain row of $\mb X_0$, and all minimizers reproduce all rows of $\mb X_0$. }

\js{Let $\mb e_3^\perp$ be the equatorial section that is orthogonal to $\mb e_3$, i.e., $\mb e_3^\perp \doteq \mathrm{span}(\mb e_1, \mb e_2) \cap \bb B^3$. } To better illustrate the above point, we project the upper hemisphere above $\mb e_3^\perp$ onto $\mb e_3^\perp$. The projection is bijective and we equivalently define a reparameterization $g: \mb e_3^\perp \mapsto \R$ of $f$. Fig.~\ref{fig:large-sample-sphere} (right) plots the graph of $g$. Obviously the only local minimizers are $\mb 0, \pm \mb e_1, \pm \mb e_2$, and they are also global minimizers. Moreover, the apparent nonconvex landscape has interesting structures around $\mb 0$: when moving away from $\mb 0$, one sees successively a strongly convex region, a strong gradient region, and a region where at each point one can always find a direction of negative curvature. This geometry implies that at any nonoptimal point, there is always at least one direction of descent. Thus, any algorithm that can take advantage of the descent directions will likely converge to a global minimizer, irrespective of initialization. 

Two challenges stand out when implementing this idea. For geometry, one has to show similar structure exists for general complete $\mb A_0$, in high dimensions ($n \ge 3$), when the number of observations $p$ is finite (vs.\ the expectation in the experiment). For algorithms, we need to be able to take advantage of this structure without knowing $\mb A_0$ ahead of time. In Section~\ref{sec:overview_alg}, we describe a Riemannian trust region method which addresses the latter challenge. 

\paragraph{Geometry for orthogonal $\mb A_0$.} In this case, we take $\widehat{\mb Y} = \mb Y = \mb A_0 \mb X_0$. Since $f\paren{\mb q; \mb A_0 \mb X_0} \allowbreak = f\paren{\mb A_0^* \mb q; \mb X_0}$, the landscape of $f\paren{\mb q; \mb A_0 \mb X_0}$ is simply a rotated version of that of $f\paren{\mb q; \mb X_0}$, i.e., when $\mb A_0 = \mb I$. Hence we will focus on the case when $\mb A_0 = \mb I$. Among the $2n$ symmetric sections of $\bb S^{n-1}$ centered around the signed basis vectors $\pm \mb e_1, \dots, \pm \mb e_n$, we work with the symmetric section around $\mb e_n$ as an exemplar. \js{An illustration of the symmetric sections and the exemplar we choose to work with on $\bb S^2$ is provided in Fig.~\ref{fig:symsec_plot}.} The result will carry over to all sections with the same argument; together this provides a complete characterization of the function $f\paren{\mb q; \mb X_0}$ over $\bb S^{n-1}$.    

\begin{figure}[!htbp]
\centering
\includegraphics[width = 0.6\textwidth]{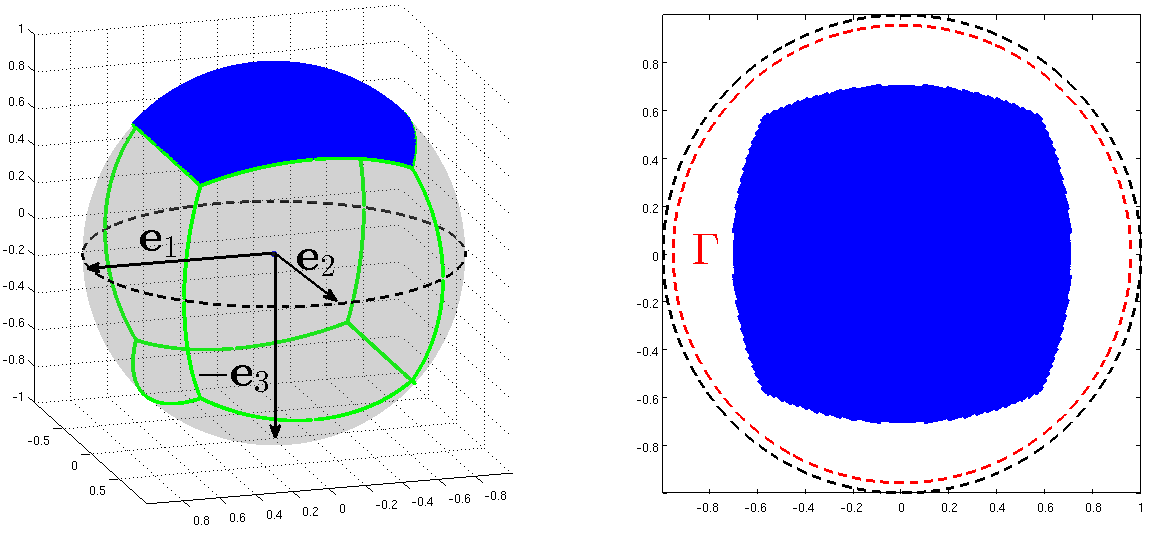}
\caption{\js{Illustration of the six symmetric sections on $\bb S^2$ and the exemplar we work with. \textbf{Left:} The six symmetric sections on $\bb S^2$, as divided by the green curves. The signed basis vectors, $\pm \mb e_i$'s, are centers of these sections. We choose to work with the exemplar that is centered around $\mb e_3$ that is shaded in blue. \textbf{Right:} Projection of the upper hemisphere onto the equatorial section $\mb e_3^\perp$. The blue region is projection of the exemplar under study. The larger region enclosed by the red circle is the $\Gamma$ set on which we characterize the reparametrized function $g$. }}
\label{fig:symsec_plot}
\end{figure}

\js{To study the function on this exemplar region,} we again invoke the projection trick described above, this time onto the equatorial section $\mb e_n^\perp$. This can be formally captured by the reparameterization mapping: 
\begin{align}
\mb q\paren{\mb w} = \paren{\mb w, \sqrt{1-\norm{\mb w}{}^2}}, \; \mb w \in \bb B^{n-1}, 
\end{align}
where $\mb w$ is the new variable and $\bb B^{n-1}$ is the unit ball in $\R^{n-1}$. We first study the composition $g\paren{\mb w; \mb X_0} \doteq f\paren{\mb q\paren{\mb w}; \mb X_0}$ over the set
\begin{align}
\Gamma \doteq \Brac{\mb w: \norm{\mb w}{} < \sqrt{\tfrac{4n-1}{4n}}} \subsetneq \bb B^{n-1} .
\end{align}
It can be verified the exemplar we chose to work with is strictly contained in this set\footnote{Indeed, if $\innerprod{\mb q}{\mb e_n} \ge \abs{\innerprod{\mb q}{\mb e_i}}$ for all $i \neq n$, $1 - \norm{\mb w}{}^2 = q_n^2 \ge 1/n$, implying $\norm{\mb w}{}^2 \le \tfrac{n-1}{n} < \tfrac{4n-1}{4n}$. The reason we have defined an open set instead of a closed (compact) one is to avoid potential trivial local minimizers located on the boundary. \js{We study behavior of $g$ over this slightly larger set $\Gamma$, instead of just the projection of the chosen symmetric section, to conveniently deal with the boundary effect: if we choose to work with just projection of the chosen symmetric section, there would be considerable technical subtleties at the boundaries when we call the union argument to cover the whole sphere. } }. \js{This is illustrated for the case $n=3$ in Fig.~\ref{fig:symsec_plot} (right). }

Our analysis characterizes the properties of $g\paren{\mb w; \mb X_0}$ by studying three quantities 
\begin{align*}
\nabla^2 g\paren{\mb w; \mb X_0}, \quad \frac{\mb w^* \nabla g\paren{\mb w; \mb X_0}}{\norm{\mb w}{}}, \quad \frac{\mb w^* \nabla^2 g\paren{\mb w; \mb X_0} \mb w}{\norm{\mb w}{}^2}
\end{align*}
respectively over three consecutive regions moving away from the origin, corresponding to the three regions in Fig.~\ref{fig:large-sample-sphere} (right). In particular, through typical expectation-concentration style arguments, we show that there exists a positive constant $c$ such that 
\begin{align} \label{eq:intro_geo_ineq}
\nabla^2 g\paren{\mb w; \mb X_0} \succeq \frac{1}{\mu} c\theta \mb I,  \quad \frac{\mb w^* \nabla g\paren{\mb w; \mb X_0}}{\norm{\mb w}{}} \ge c \theta, \quad \frac{\mb w^* \nabla^2 g\paren{\mb w; \mb X_0} \mb w}{\norm{\mb w}{}^2} \le -c\theta
\end{align}
over the respective regions w.h.p., confirming our low-dimensional observations described above. In particular, the favorable structure we observed for $n = 3$ persists in high dimensions, w.h.p., even when $p$ is large \emph{yet finite}, for the case $\mb A_0$ is orthogonal. Moreover, the local minimizer of $g\paren{\mb w; \mb X_0}$ over $\Gamma$ is very close to $\mb 0$, within a distance of $O\paren{\mu}$\footnote{\js{When $p \to \infty$, the local minimizer is exactly $\mb 0$; deviation from $\mb 0$ that we described is due to finite-sample perturbation. The deviation distance depends both the $h_{\mu}(\cdot)$ and $p$; see Theorem~\ref{thm:geometry_orth} for example.}  }. 

\paragraph{Geometry for complete $\mb A_0$.} For general complete dictionaries $\mb A_0$, we hope that the function $f$  retains the nice geometric structure discussed above. We can ensure this by ``preconditioning'' $\mb Y$ such that the output looks as if being generated from a certain orthogonal matrix, possibly plus a small perturbation. We can then argue that the perturbation does not significantly affect qualitative properties of the objective landscape. Write
\begin{align} \label{eq:precond_eq}
\overline{\mb Y} = \paren{\tfrac{1}{p\theta} \mb Y \mb Y^*}^{-1/2} \mb Y. 
\end{align}
Note that for $\mb X_0 \sim_{i.i.d.} \mathrm{BG}\paren{\theta}$, $\expect{\mb X_0 \mb X_0^*}/\paren{p\theta} = \mb I$. Thus, one expects $\tfrac{1}{p\theta} \mb Y \mb Y^* = \tfrac{1}{p\theta} \mb A_0 \mb X_0 \mb X_0^* \mb A_0^*$ to behave roughly like $\mb A_0 \mb A_0^*$ and hence $\overline{\mb Y}$ to behave like 
\begin{align}
\paren{\mb A_0 \mb A_0^*}^{-1/2} \mb A_0 \mb X_0 
& = \paren{\mb U \mb \Sigma \mb V^* \mb V \mb \Sigma \mb U^*}^{-1/2} \mb U \mb \Sigma \mb V^* \mb X_0 \nonumber\\
& = \mb U \mb \Sigma^{-1} \mb U^* \mb U \mb \Sigma \mb V^* \mb X_0 
\nonumber \\
& = \mb U \mb V^* \mb X_0
\end{align}
where $\mathtt{SVD}(\mb A_0) = \mb U \mb \Sigma \mb V^*$. It is easy to see $\mb U \mb V^*$ is an orthogonal matrix. Hence the preconditioning scheme we have introduced is technically sound.  

Our analysis shows that $\overline{\mb Y}$ can be written as 
\begin{align}
\overline{\mb Y} = \mb U \mb V^* \mb X_0 + \mb \Xi \mb X_0, 
\end{align}
where $\mb \Xi$ is a matrix with a small magnitude. Simple perturbation argument shows that the constant $c$ in~\eqref{eq:intro_geo_ineq} is at most shrunk to $c/2$ for all $\mb w$ when $p$ is sufficiently large. Thus, the qualitative aspects of the geometry have not been changed by the perturbation. 

\subsubsection{\textbf{A Second-order Algorithm on Manifold: Riemannian Trust-Region Method}} \label{sec:overview_alg}
We do not know $\mb A_0$ ahead of time, so our algorithm needs to take advantage of the structure described above without knowledge of $\mb A_0$. Intuitively, this seems possible as the descent direction in the $\mb w$ space appears to also be a local descent direction for $f$ over the sphere. Another issue is that although the optimization problem has no spurious local minimizers, it does have many saddle points with indefinite Hessian, which we call \emph{ridable saddles} \footnote{See~\cite{sun2015nonconvex} and~\cite{ge2015escaping}. } (Fig.~\ref{fig:large-sample-sphere}). We can use second-order information to guarantee to escape from such saddle points. In the companion paper~\cite{sun2015complete_b}, we derive an algorithm based on the Riemannian trust region method (TRM)~\cite{absil2007trust, absil2009} for this purpose. There are other algorithmic possibilities; see, e.g., ~\cite{goldfarb1980curvilinear, ge2015escaping}. 

We provide here only the basic intuition why a local minimizer can be retrieved by the second-order trust-region method. Consider an unconstrained optimization problem 
\begin{align*}
\min_{\mb x \in \R^n} \phi\paren{\mb x}. 
\end{align*}
Typical (second-order) TRM proceeds by successively forming a second-order approximation to $\phi$ at the current iterate, 
\begin{align} \label{eqn:trm_quad_approx}
\widehat{\phi}(\mb \delta; \mb x^{(r-1)}) \doteq \phi(\mb x^{(r-1)}) + \nabla^* \phi(\mb x^{(r-1)}) \mb \delta   + \tfrac{1}{2} \mb \delta^* \mb Q(\mb x^{(r-1)}) \mb \delta, 
\end{align}
where $\mb Q(\mb x^{(r-1)})$ is a proxy for the Hessian matrix $\nabla^2 \phi(\mb x^{(r-1)})$, which encodes the second-order geometry. The next movement direction is determined by seeking a minimum of $\widehat{\phi}(\mb \delta; \mb x^{(r-1)})$ over a small region, normally a norm ball $\|\mb \delta\|_p \le \Delta$, called the trust region, inducing the well-studied trust-region subproblem that can efficiently solved: 
\begin{align}
\mb \delta^{(r)} \doteq \mathop{\arg\min}_{\mb \delta \in \R^n, \norm{\mb \delta}{p} \le \Delta} \widehat{\phi}(\mb \delta; \mb x^{(r-1)}), 
\end{align}
where $\Delta$ is called the trust-region radius that controls how far the movement can be made. If we take $\mb Q(\mb x^{(r-1)}) = \nabla^2 \phi(\mb x^{(r-1)})$ for all $r$, then whenever the gradient is nonvanishing or the Hessian is indefinite, we expect to decrease the objective function by a concrete amount provided $\|\mb \delta\|$ is sufficiently small. Since the domain is compact, the iterate sequence ultimately moves into the strongly convex region, where the trust-region algorithm behaves like a typical Newton algorithm. All these are generalized to our objective over the sphere and made rigorous in the companion paper~\cite{sun2015complete_b}.

\subsection{Prior Arts and Connections} \label{sec:lit_review}
It is far too ambitious to include here a comprehensive review of the exciting developments of DL algorithms and applications after the pioneer work~\cite{olshausen1996emergence}. We refer the reader to Chapter 12 - 15 of the book~\cite{elad2010sparse} and the survey paper~\cite{mairal2014sparse} for summaries of relevant developments in image analysis and visual recognition. In the following, we focus on reviewing recent developments on the theoretical side of dictionary learning, and draw connections to problems and techniques that are relevant to the current work. 

\paragraph{Theoretical Dictionary Learning} 
The theoretical study of DL in the recovery setting started only very recently. \cite{aharon2006uniqueness} was the first to provide an algorithmic procedure to correctly extract the generating dictionary. The algorithm requires exponentially many samples and has exponential running time; see also~\cite{hillar2011can}. Subsequent work~\cite{gribonval2010dictionary, geng2011local, schnass2014local, schnass2014identifiability,schnass2015convergence} studied when the target dictionary is a local optimizer of natural recovery criteria. These meticulous analyses show that polynomially many samples are sufficient to ensure local correctness under natural assumptions. However, these results do not imply that one can design efficient algorithms to obtain the desired local optimizer and hence the dictionary. 

\cite{spielman2012exact} initiated the on-going research effort to provide efficient algorithms that globally solve DR. They showed that one can recover a complete dictionary $\mb A_0$ from $\mb Y = \mb A_0 \mb X_0$ by solving a certain sequence of linear programs, when $\mb X_0$ is a sparse random matrix (\js{under the Bernoulli-Subgaussian model}) with $O(\sqrt{n})$ nonzeros per column \js{(and the method provably breaks down when $\mb X_0$ contains slightly more than $\Omega(\sqrt{n})$ nonzeros per column)}. \cite{agarwal2013learning, agarwal2013exact} and~\cite{arora2013new, arora2015simple} gave efficient algorithms that provably recover overcomplete ($m \ge n$), incoherent dictionaries, based on a combination of \{clustering or spectral initialization\} and local refinement. These algorithms again succeed when $\mb X_0$ has $\widetilde{O}(\sqrt{n})$ \footnote{The $\widetilde{O}$ suppresses some logarithm factors.} nonzeros per column. Recent work \cite{barak2014dictionary} provided the first polynomial-time algorithm that provably recovers most ``nice'' overcomplete dictionaries when $\mb X_0$ has $O(n^{1-\delta})$ nonzeros per column for any constant $\delta \in (0, 1)$. However, the proposed algorithm runs in super-polynomial (quasipolynomial) time when the sparsity level goes up to $O(n)$. Similarly, \cite{arora2014more} also proposed a super-polynomial  time algorithm that guarantees recovery with (almost) $O\paren{n}$ nonzeros per column. \js{Detailed models for those methods dealing with overcomplete dictionaries are differ from one another; nevertheless, they all assume each column of $\mb X_0$ has bounded sparsity levels, and the nonzero coefficients have certain sub-Gaussian magnitudes\footnote{Thus, one may anticipate that the performances of those methods do not change much qualitatively, if the BG model for the coefficients had been assumed. }.} By comparison, we give the first \emph{polynomial-time} algorithm that provably recovers complete dictionary $\mb A_0$ when $\mb X_0$ has $O\paren{n}$ nonzeros per column, \js{under the BG model.} 

Aside from efficient recovery, other theoretical work on DL includes results on identifiability~\cite{aharon2006uniqueness, hillar2011can, wu2015local}, generalization bounds~\cite{maurer2010dimensional, Vainsencher:2011, Mehta13, gribonval2013sample}, and noise stability~\cite{GribonvalJB14}. 

\paragraph{Finding Sparse Vectors in a Linear Subspace} 
We have followed~\cite{spielman2012exact} and cast the core problem as finding the sparsest vectors in a given linear subspace, which is also of independent interest. Under a planted sparse model\footnote{... where one sparse vector embedded in an otherwise random subspace.}, \cite{demanet2014scaling} showed that solving a sequence of linear programs similar to~\cite{spielman2012exact} can recover sparse vectors with sparsity up to $O\paren{p/\sqrt{n}}$, sublinear in the vector dimension. \cite{qu2014finding} improved the recovery limit to $O\paren{p}$ by solving a nonconvex sphere-constrained problem similar to~\eqref{eq:main_l2}\footnote{The only difference is that they chose to work with the Huber function as a proxy of the $\norm{\cdot}{1}$ function. } via an ADM algorithm. The idea of seeking rows of $\mb X_0$ sequentially by solving the above core problem sees precursors in~\cite{zibulevsky2001blind} for blind source separation, and \cite{gottlieb2010matrix} for matrix sparsification. \cite{zibulevsky2001blind} also proposed a nonconvex optimization similar to~\eqref{eq:main_l2} here and that employed in~\cite{qu2014finding}. 

\paragraph{Nonconvex Optimization Problems} For other nonconvex optimization problems of recovery of structured signals\footnote{This is a body of recent work studying nonconvex recovery up to statistical precision, including, e.g., \cite{loh2011high,loh2013regularized,wang2014nonconvex,balakrishnan2014statistical,wang2014high,loh2014support,loh2015statistical,sun2015provable}. }, 
including low-rank matrix completion/recovery~\cite{keshavan2010matrix, jain2013low, hardt2014understanding, hardt2014fast, netrapalli2014non, jain2014fast, sun2014guaranteed, zheng2015convergent, tu2015low, chen2015fast}, phase retreival~\cite{netrapalli2013phase, candes2015phase, chen2015solving, white2015local}, tensor recovery~\cite{jain2014provable, anandkumar2014guaranteed, anandkumar2014analyzing, anandkumar2015tensor}, mixed regression~\cite{yi2013alternating, lee2013near}, structured element pursuit~\cite{qu2014finding}, and recovery of simultaneously structured signals~\cite{lee2013near}, numerical linear algebra and optimization~\cite{jain2015computing, bhojanapalli2015dropping}, the initialization plus local refinement strategy adopted in theoretical DL~\cite{agarwal2013learning, agarwal2013exact, arora2013new, arora2015simple, arora2014more} is also crucial: nearness to the target solution enables exploiting the local property of the optimizing objective to ensure that the local refinement succeeds.\footnote{The powerful framework~\cite{attouch2010proximal,bolte2014proximal} to establish local convergence of ADM algorithms to critical points applies to DL/DR also, see, e.g., \cite{bao2014l0, bao2014convergent, bao2014convergence}. However, these results do not guarantee to produce global optima. } By comparison, we provide a complete characterization of the global geometry, which admits efficient algorithms without any special initialization.


\paragraph{Independent Component Analysis (ICA) and Other Matrix Factorization Problems} 
\begin{figure}[!htbp]
\centering  
\begin{subfigure}[t]{0.3\textwidth}
\centering
\includegraphics[width = 1\linewidth]{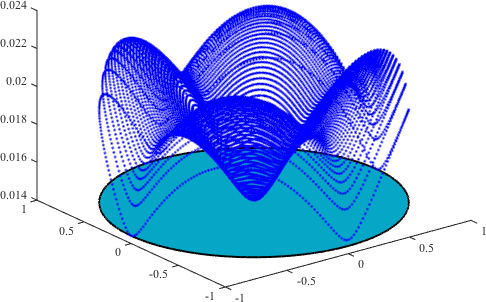} 
\caption{\footnotesize Correlated Gaussian, $\theta = 0.1$}
\label{subfig:gaussian-corr-0.1}
\end{subfigure}
\begin{subfigure}[t]{0.3\textwidth}
\centering
\includegraphics[width = 1\linewidth]{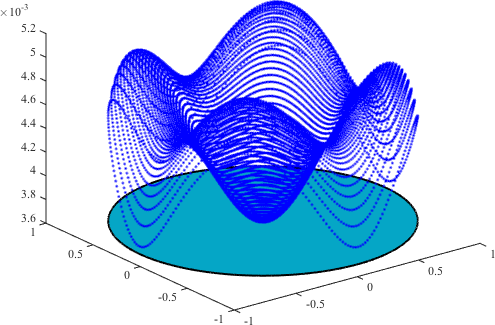} 
\caption{\footnotesize Correlated Uniform, $\theta = 0.1$}
\label{subfig:uniform-corr-0.1}
\end{subfigure} 
\begin{subfigure}[t]{0.3\textwidth}
\centering
\includegraphics[width = 1\linewidth]{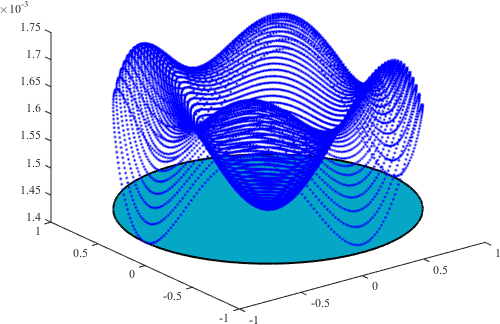}
\caption{\footnotesize Independent Uniform, $\theta = 0.1$}
\label{subfig:uniform-ind-0.1}
\end{subfigure} \\
\begin{subfigure}[t]{0.3\textwidth}
\centering
\includegraphics[width = 1\linewidth]{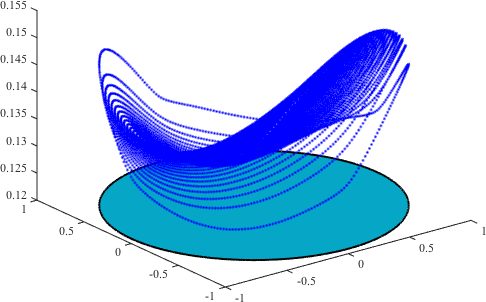} 
\caption{\footnotesize Correlated Gaussian, $\theta = 0.9$}
\label{subfig:gaussian-corr-0.9}
\end{subfigure}
\begin{subfigure}[t]{0.3\textwidth}
\centering
\includegraphics[width = 1\linewidth]{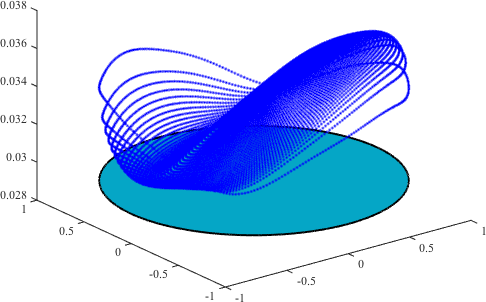}
\caption{\footnotesize Correlated Uniform, $\theta = 0.9$}
\label{subfig:uniform-corr-0.9}
\end{subfigure}
\begin{subfigure}[t]{0.3\textwidth}
\centering
\includegraphics[width = 1\linewidth]{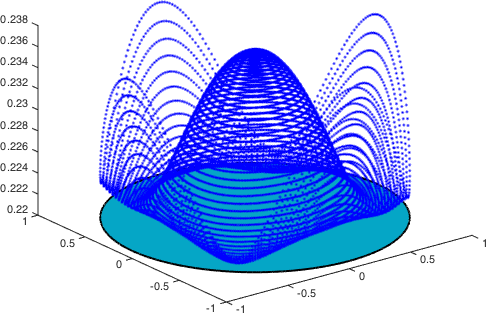}
\caption{\footnotesize Independent Uniform, $\theta = 1$}
\label{subfig:uniform-ind-0.9}
\end{subfigure}
\caption{\textbf{Asymptotic function landscapes when rows of $\mb X_0$ are not independent.} W.l.o.g., we again assume $\mb A_0 = \mb I$. In (a) and (d), $\mb X_0 = \mb \Omega \odot \mb V$, with $\mb \Omega \sim_{i.i.d.} \mathrm{Ber}(\theta)$ and columns of $\mb X_0$ i.i.d. Gaussian vectors obeying $\mb v_i \sim \mc N(\mb 0, \mb \Sigma^2)$ for symmetric $\mb \Sigma$ with $1$'s on the diagonal and i.i.d. off-diagonal entries distributed as $\mc N(0, \sqrt{2}/20)$. Similarly, in (b) and (e), $\mb X_0 = \mb \Omega \odot \mb W$, with $\mb \Omega \sim_{i.i.d.} \mathrm{Ber}(\theta)$ and columns of $\mb X_0$ i.i.d. vectors generated as $\mb w_i = \mb \Sigma \mb u^i$ with $\mb u_i \sim_{i.i.d.} \mathrm{Uniform}[-0.5,0.5]$. For comparison, in (c) and (f), $\mb X_0 = \mb \Omega \odot \mb W$ with $\mb \Omega \sim_{i.i.d.} \mathrm{Ber}(\theta)$ and $\mb W \sim_{i.i.d.} \mathrm{Uniform}[-0.5,0.5]$. Here $\odot$ denote the elementwise product, and the objective function is still based on the $\log\cosh$ function as in~\eqref{eq:main_l2}. }
\label{fig:DL-dependent}
\end{figure}
DL can also be considered in the general framework of matrix factorization problems,  which encompass the classic principal component analysis (PCA), ICA, and clustering, and more recent problems such as nonnegative matrix factorization (NMF), multi-layer neural nets (deep learning architectures). Most of these problems are NP-hard. Identifying tractable cases of practical interest and providing provable efficient algorithms are subject of on-going research endeavors; see, e.g., recent progresses on NMF~\cite{arora2012computing}, and learning deep neural nets~\cite{arora2013provable,sedghi2014provable,neyshabur2013sparse,livni2014computational}. 

ICA factors a data matrix $\mb Y$ as $\mb Y = \mb A \mb X$ such that $\mb A$ is square and rows of $\mb X$ achieve maximal statistical independence~\cite{hyvarinen2000independent,hyvarinen2001}. In theoretical study of the recovery problem, it is often assumed that rows of $\mb X_0$ are (weakly) independent (see, e.g., ~\cite{comon1994independent, frieze1996learning, arora2012provable}). Our i.i.d.\ probability model on $\mb X_0$ implies rows of $\mb X_0$ are independent, aligning our problem perfectly with the ICA problem. More interestingly, the $\log\cosh$ objective we analyze here was proposed as a general-purpose \emph{contrast function} in ICA that has not been thoroughly analyzed~\cite{hyvarinen99fast}.  Algorithms and analysis with another popular contrast function, the fourth-order cumulants, however, indeed overlap with ours considerably~\cite{frieze1996learning, arora2012provable}\footnote{Nevertheless, the objective functions are apparently different. Moreover, we have provided a complete geometric characterization of the objective, in contrast to~\cite{frieze1996learning, arora2012provable}. We believe the geometric characterization could not only provide insight to the algorithm, but also help improve the algorithm in terms of stability and also finding all components. }. While this interesting connection potentially helps port our analysis to ICA, it is a fundamental question to ask what is playing a more vital role for DR, sparsity or independence. 

Fig.~\ref{fig:DL-dependent} helps shed some light in this direction, where we again plot the asymptotic objective landscape with the natural reparameterization as in Section~\ref{sec:overview_geometry}. From the left and central panels, it is evident that even without independence, $\mb X_0$ with sparse columns induces the familiar geometric structures we saw in Fig.~\ref{fig:large-sample-sphere}; such structures are broken when the sparsity level becomes large. We believe all our later analyses can be generalized to the correlated cases we experimented with. On the other hand, from the right panel\footnote{We have not showed the results on the BG model here, as it seems the structure persists even when $\theta$ approaches $1$. We suspect the ``phase transition'' of the landscape occurs at different points for different distributions and Gaussian is the outlying case where the transition occurs at $1$. }, it seems that with independence, the function landscape undergoes a transition, as sparsity level grows: target solution goes from minimizers of the objective to the maximizers of the objective. Without adequate knowledge of the true sparsity, it is unclear whether one would like to minimize or maximize the objective.\footnote{For solving the ICA problem, this suggests the $\log \cosh$ contrast function, that works well empirically~\cite{hyvarinen99fast}, may not work for all distributions (rotation-invariant Gaussian excluded of course), at least when one does not process the data (say perform certain whitening or scaling).} This suggests that sparsity, instead of independence, makes our current algorithm for DR work. 

\paragraph{Nonconvex Problems with Similar Geometric Structure} Besides ICA discussed above, it turns out that a handful of other practical problems arising in signal processing and machine learning induce the ``no spurious minimizers, all saddles are second-order'' structure under natural setting, including the eigenvalue problem, generalized phase retrieval~\cite{sun2016geometric}, orthogonal tensor decomposition~\cite{ge2015escaping}, low-rank matrix recovery/completion~\cite{bhojanapalli2016global,ge2016matrix}, noisy phase synchronization and community detection~\cite{boumal2016non,boumal2016nonconvex,bandeira2016low}, linear neural nets learning~\cite{baldi1989neural,kawaguchi2016deep,soudry2016no}. \cite{sun2015nonconvex} gave a review of these problems, and discussed how the methodology developed in this and the companion paper~\cite{sun2015complete_b} can be generalized to solve those problems. 

\subsection{Notations, and Reproducible Research}
We use bold capital and small letters such as $\mb X$ and $\mb x$ to denote matrices and vectors, respectively. Small letters are reserved for scalars. Several specific mathematical objects we will frequently work with: $O_k$ for the orthogonal group of order $k$, $\bb S^{n-1}$ for the unit sphere in $\R^n$, $\bb B^n$ for the unit ball in $\R^n$, and $[m] \doteq \set{1, \dots, m}$ for positive integers $m$. We use $\paren{\cdot}^*$ for matrix transposition, causing no confusion as we will work entirely on the real field. We use superscript to index rows of a matrix, such as $\mb x^i$ for the $i$-th row of the matrix $\mb X$,  and subscript to index columns, such as $\mb x_j$. All vectors are defaulted to column vectors. So the $i$-th row of $\mb X$ as a row vector will be written as $\paren{\mb x^i}^*$. For norms, $\norm{\cdot}{}$ is the usual $\ell^2$ norm for a vector and the operator norm (i.e., $\ell^2 \to \ell^2$) for a matrix; all other norms will be indexed by subscript, for example the Frobenius norm $\norm{\cdot}{F}$ for matrices and the element-wise max-norm $\norm{\cdot}{\infty}$. We use $\mb x \sim \mc L$ to mean that the random variable $\mb x$ is distributed according to the law $\mc L$. Let $\mc N$ denote the Gaussian law. Then $\mb x \sim \mc N\paren{\mb 0, \mb I}$ means that $\mb x$ is a standard Gaussian vector. Similarly, we use $\mb x \sim_{i.i.d.} \mc L$ to mean elements of $\mb x$ are independently and identically distributed according to the law $\mc L$. So the fact $\mb x \sim \mc N\paren{\mb 0, \mb I}$ is equivalent to that $\mb x \sim_{i.i.d.} \mc N\paren{0, 1}$. One particular distribution of interest for this paper is the Bernoulli-Gaussian with rate $\theta$: $Z \sim B \cdot G$, with $G \sim \mc N\paren{0, 1}$ and $B \sim \mathrm{Ber}\paren{\theta}$. We also write this compactly as $Z \sim \mathrm{BG}\paren{\theta}$. We reserve indexed $C$ and $c$ for absolute constants when stating and proving technical results. The scopes of such constants are local unless otherwise noted. We use standard notations for most other cases, with exceptions clarified locally. 

The codes to reproduce all the figures and experimental results are available online: 
\begin{quote}
\centering
\url{https://github.com/sunju/dl_focm} . 
\end{quote}

\section{The High-dimensional Function Landscape} \label{sec:geometry}
To characterize the function landscape of $f\paren{\mb q; \mb X_0}$ over $\bb S^{n-1}$, we mostly work with the function 
\begin{align}\label{eqn:function-g}
g\paren{\mb w} \doteq f\paren{\mb q\paren{\mb w}; \mb X_0} = \frac{1}{p} \sum_{k=1}^p h_{\mu}\paren{\mb q\paren{\mb w}^* \paren{\mb x_0}_k}, 
\end{align}
induced by the reparametrization
\begin{align}
\mb q\paren{\mb w} = \paren{\mb w, \sqrt{1-\norm{\mb w}{}^2}},  \quad \mb w \in \bb B^{n-1}. 
\end{align}
In particular, we focus our attention to the smaller set 
\begin{align}
\Gamma = \set{\mb w: \norm{\mb w}{} < \sqrt{\frac{4n-1}{4n}}} \subsetneq \bb B^{n-1},  
\end{align}
because $\mb q\paren{\Gamma}$ contains all points $\mb q \in \bb S^{n-1}$ with $n \in \mathop{\arg\max}_{i \in \pm [n]} \mb q^* \mb e_i$ and we can similarly characterize other parts of $f$ on $\bb S^{n-1}$ using projection onto other equatorial sections. Note that over $\Gamma$, $q_n = \sqrt{1-\norm{\mb w}{}^2} \ge 1/(2\sqrt{n})$. 

\subsection{Main Geometric Theorems} 
\begin{theorem}[High-dimensional landscape - orthogonal dictionary]\label{thm:geometry_orth}
Suppose $\mb A_0 = \mb I$ and hence $\mb Y = \mb A_0 \mb X_0 = \mb X_0$. There exist positive constants $c_\star$ and $C$, such that for any $\theta \in (0,1/2)$ and $\mu < c_a\min\Brac{\theta n^{-1}, n^{-5/4}}$, whenever 
\begin{align}
p \ge \frac{C}{\mu^2 \theta^2} n^3 \log \frac{n}{\mu \theta},
\end{align}
the following hold simultaneously with probability at least $1 - c_b p^{-6}$: 
\begin{align}
\nabla^2 g(\mb w; \mb X_0) &\succeq \frac{c_\star \theta}{\mu} \mb I \quad &\forall \, \mb w \quad \text{s.t.}& \quad \norm{\mb w}{} \le \frac{\mu}{4\sqrt{2}},  \label{eqn:hess-zero-uni-orth} \\
\frac{\mb w^* \nabla g(\mb w; \mb X_0)}{\norm{\mb w}{}} &\ge c_\star \theta \quad &\forall \, \mb w \quad \text{s.t.}& \quad \frac{\mu}{4\sqrt{2}} \le \norm{\mb w}{} \le \frac{1}{20 \sqrt{5}} \label{eqn:grad-uni-orth} \\
\frac{\mb w^* \nabla^2 g(\mb w; \mb X_0) \mb w}{\norm{\mb w}{}^2} &\le - c_\star \theta \quad &\forall \, \mb w \quad \text{s.t.}& \quad \frac{1}{20 \sqrt{5}} \le \norm{\mb w}{} \le \sqrt{\frac{4n-1}{4n}},   \label{eqn:curvature-uni-orth}
\end{align}
{\em and} the function $g(\mb w; \mb X_0)$ has exactly one local minimizer $\mb w_\star$ over the open set $\Gamma \doteq \set{ \mb w: \norm{\mb w}{} < \sqrt{ \tfrac{4n-1}{4n} } }$, which satisfies
\begin{equation}
\norm{\mb w_\star - \mb 0 }{} \;\le\; \min\Brac{\frac{c_c\mu}{\theta} \sqrt{\frac{n \log p}{p}}, \frac{\mu}{16}}.    
\end{equation}
Here $c_a$ through $c_c$ are all positive constants. 
\end{theorem}
Here $\mb q\paren{\mb 0} = \mb e_n$, which exactly recovers the last row of $\mb X_0$, $(\mb x_0)^n$. Though the unique local minimizer $\mb w_\star$ may not be $\mb 0$, it is very near to $\mb 0$. Hence the resulting $\mb q\paren{\mb w_\star}$ produces a close approximation to $(\mb x_0)^n$. Note that $\mb q\paren{\Gamma}$ (strictly) contains all points $\mb q \in \bb S^{n-1}$ such that $n = \mathop{\arg \max}_{i \in \pm [n]} \mb q^* \mb e_i$. We can characterize the graph of the function $f\paren{\mb q; \mb X_0}$ in the vicinity of other signed basis vector $\pm \mb e_i$ simply by changing the equatorial section $\mb e_n^\perp$ to $\mb e_i^\perp$. Doing this $2n$ times (and multiplying the failure probability in Theorem~\ref{thm:geometry_orth} by $2n$), we obtain a characterization of $f\paren{\mb q; \mb X_0}$ over the entirety of $\bb S^{n-1}$.\footnote{In fact, it is possible to pull the very detailed geometry captured in~\eqref{eqn:hess-zero-uni-orth} through~\eqref{eqn:curvature-uni-orth} back to the sphere (i.e., the $\mb q$ space) also; analysis of the Riemannian trust-region algorithm later does part of these. We will stick to this simple global version here. }
 The result is captured by the next corollary. 

\begin{corollary} \label{cor:geometry_orth}
Suppose $\mb A_0 = \mb I$ and hence $\mb Y = \mb A_0 \mb X_0 = \mb X_0$. There exist positive constant $C$, such that for any $\theta \in (0,1/2)$ and $\mu < c_a\min\Brac{\theta n^{-1}, n^{-5/4}}$, whenever $
p \ge \frac{C}{\mu^2 \theta^2} n^3 \log \frac{n}{\mu \theta} $, with probability at least $1 - c_b p^{-5}$, the function $f\paren{\mb q; \mb X_0}$ has exactly $2n$ local minimizers over the sphere $\bb S^{n-1}$. In particular, there is a bijective map between these minimizers and signed basis vectors $\set{\pm \mb e_i}_i$, such that the corresponding local minimizer $\mb q_\star$ and $\mb b \in \set{\pm \mb e_i}_i$ satisfy 
\begin{align}
\norm{\mb q_\star - \mb b}{} \le \sqrt{2}\min\Brac{\frac{c_c\mu}{\theta} \sqrt{\frac{n \log p}{p}}, \frac{\mu}{16}}. 
\end{align} 
Here $c_a$ to $c_c$ are positive constants.
\end{corollary}
\begin{proof}
By Theorem~\ref{thm:geometry_orth}, over $\mb q\paren{\Gamma}$, $\mb q\paren{\mb w_\star}$ is the unique local minimizer. Suppose not. Then there exist $\mb q' \in \mb q\paren{\Gamma}$ with $\mb q' \neq \mb q\paren{\mb w_\star}$ and $\eps > 0$, such that $f\paren{\mb q'; \mb X_0} \le f\paren{\mb q; \mb X_0}$ for all $\mb q \in \mb q\paren{\Gamma}$ satisfying $\norm{\mb q' - \mb q}{} < \eps$. Since the mapping $\mb w \mapsto \mb q\paren{\mb w}$ is $2\sqrt{n}$-Lipschitz (Lemma~\ref{lem:lip-h-mu}), $g\paren{\mb w\paren{\mb q'}; \mb X_0} \le g\paren{\mb w\paren{\mb q}; \mb X_0}$ for all $\mb w \in \Gamma$ satisfying $\norm{\mb w\paren{\mb q'} - \mb w\paren{\mb q}}{} < \eps/\paren{2\sqrt{n}}$, implying $\mb w\paren{\mb q'}$ is a local minimizer different from $\mb w_\star$, a contradiction. Let $\norm{\mb w_\star - \mb 0}{} = \eta$. Straightforward calculation shows 
\begin{align*}
\norm{\mb q\paren{\mb w_\star} - \mb e_n}{}^2 = (1-\sqrt{1-\eta^2})^2 + \eta^2 = 2 - 2\sqrt{1-\eta^2} \le 2\eta^2. 
\end{align*}
Repeating the argument $2n$ times in the vicinity of other signed basis vectors $\pm \mb e_i$ gives $2n$ local minimizers of $f$. Indeed, the $2n$ symmetric sections cover the sphere with certain overlaps. 
\js{We claim that none of the $2n$ local minimizers lies in the overlapped regions. This is due to the nearness of these local minimizers to standard basis vectors. To see this, w.l.o.g., suppose $\mb q$, which is the local minimizer next to $\mb e_n$, is in the overlapped region determined by $e_n$ and $e_i$ for some $i \ne n$. This implies that 
\begin{align*}
\norm{[w_1(\mb q), \dots, w_{i-1}(\mb q), w_{i+1}(\mb q), \dots, q_n]}{}^2 < \frac{4n-1}{4n} 
\end{align*}  
by the definition of our symmetric sections. On the other hand, we know 
\begin{align*}
\norm{[w_1(\mb q), \dots, w_{i-1}(\mb q), w_{i+1}(\mb q), \dots, q_n]}{}^2 \ge q_n^2 = 1 - \eta^2. 
\end{align*}
Thus, so long as $1 - \eta^2 \ge \tfrac{4n-1}{4n}$, or $\eta \le 1/(2\sqrt{n})$, a contradiction arises. Since $\eta \in O(\mu)$ and $\mu \le O(n^{-5/4})$ by our assumption, our claim is confirmed. 
}
There are no extra local minimizers, as any extra local minimizer must be contained in at least one of the $2n$ symmetric sections, making two different local minimizers in one section, contradicting the uniqueness result we obtained above. 
\end{proof}

Though the $2n$ isolated local minimizers may have different objective values, they are equally good in the sense each of them helps produce a close approximation to a certain row of $\mb X_0$. As discussed in Section~\ref{sec:overview_geometry}, for cases $\mb A_0$ is an orthobasis other than $\mb I$, the landscape of $f\paren{\mb q; \mb Y}$ is simply a rotated version of the one we characterized above. 

\begin{theorem}[High-dimensional landscape - complete dictionary]\label{thm:geometry_comp}
Suppose $\mb A_0$ is complete with its condition number $\kappa\paren{\mb A_0}$. There exist positive constants $c_\star$ (particularly, the same constant as in Theorem~\ref{thm:geometry_orth}) and $C$, such that for any $\theta \in (0,1/2)$ and $\mu < c_a\min\Brac{\theta n^{-1}, n^{-5/4}}$, when 
\begin{align}
p \ge \frac{C}{c_\star^2 \theta^2} \max\set{\frac{n^4}{\mu^4}, \frac{n^5}{\mu^2}} \kappa^8\paren{\mb A_0} \log^4\paren{\frac{\kappa\paren{\mb A_0} n}{\mu \theta}}
\end{align}
and $\overline{\mb Y} \doteq \sqrt{p\theta}\paren{\mb Y \mb Y^*}^{-1/2} \mb Y$, $\mb U \mb \Sigma \mb V^* = \mathtt{SVD}\paren{\mb A_0}$, the following hold simultaneously with probability at least $1 - c_bp^{-6}$: 
\begin{align}
\nabla^2 g(\mb w; \mb V \mb U^* \overline{\mb Y}) &\succeq \frac{c_\star \theta}{2\mu} \mb I \quad &\forall \, \mb w \quad \text{s.t.}& \quad \norm{\mb w}{} \le \frac{\mu}{4\sqrt{2}},  \label{eqn:hess-zero-uni-comp} \\
\frac{\mb w^* \nabla g(\mb w; \mb V \mb U^* \overline{\mb Y})}{\norm{\mb w}{}} &\ge \frac{1}{2}c_\star \theta \quad &\forall \, \mb w \quad \text{s.t.}& \quad \frac{\mu}{4\sqrt{2}} \le \norm{\mb w}{} \le \frac{1}{20 \sqrt{5}} \label{eqn:grad-uni-comp} \\
\frac{\mb w^* \nabla^2 g(\mb w; \mb V \mb U^* \overline{\mb Y}) \mb w}{\norm{\mb w}{}^2} &\le -\frac{1}{2} c_\star \theta \quad &\forall \, \mb w \quad \text{s.t.}& \quad \frac{1}{20 \sqrt{5}} \le \norm{\mb w}{} \le \sqrt{\frac{4n-1}{4n}},   \label{eqn:curvature-uni-comp}
\end{align}
{\em and} the function $g(\mb w; \mb V \mb U^* \overline{\mb Y})$ has exactly one local minimizer $\mb w_\star$ over the open set $\Gamma \doteq \set{ \mb w: \norm{\mb w}{} < \sqrt{ \tfrac{4n-1}{4n} } }$, which satisfies
\begin{equation}
\norm{\mb w_\star - \mb 0 }{} \;\le\; \mu /7.    
\end{equation}
Here $c_a, a_b$ are both positive constants.
\end{theorem}
\begin{corollary}  \label{cor:geometry_comp}
Suppose $\mb A_0$ is complete with its condition number $\kappa\paren{\mb A_0}$. There exist positive constants $c_\star$ (particularly, the same constant as in Theorem~\ref{thm:geometry_orth}) and $C$, such that for any $\theta \in (0,1/2)$ and $\mu < c_a\min\Brac{\theta n^{-1}, n^{-5/4}}$, when $p \ge \frac{C}{c_\star^2 \theta^2} \max\set{\frac{n^4}{\mu^4}, \frac{n^5}{\mu^2}} \kappa^8\paren{\mb A_0} \allowbreak \log^4\paren{\frac{\kappa\paren{\mb A_0} n}{\mu\theta}}$ and $\overline{\mb Y} \doteq \sqrt{p\theta}\paren{\mb Y \mb Y^*}^{-1/2} \mb Y$, $\mb U \mb \Sigma \mb V^* = \mathtt{SVD}\paren{\mb A_0}$, with probability at least $1 - c_b p^{-5}$, the function $f\paren{\mb q; \mb V \mb U^* \overline{\mb Y}}$ has exactly $2n$ local minimizers over the sphere $\bb S^{n-1}$. In particular, there is a bijective map between these minimizers and signed basis vectors $\set{\pm \mb e_i}_i$, such that the corresponding local minimizer $\mb q_\star$ and $\mb b \in \set{\pm \mb e_i}_i$ satisfy 
\begin{align}
\norm{\mb q_\star - \mb b}{} \le \sqrt{2}\mu/7. 
\end{align} 
Here $c_a, c_b$ are both positive constants.
\end{corollary}
We omit the proof to Corollary~\ref{cor:geometry_comp} as it is almost identical to that of corollary~\ref{cor:geometry_orth}. From the above theorems, it is clear that for any saddle point in the $\mb w$ space, the Hessian has at least one negative eigenvalue with an associated eigenvector $\mb w/\|\mb w\|$. Now the question is whether all saddle points of $f$ on $\bb S^{n-1}$ have analogous properties, since as alluded to in Section~\ref{sec:overview_alg}, we need to perform actual optimization in the $\mb q$ space. This is indeed true, but we will only argue informally in the companion paper~\cite{sun2015complete_b}. The arguments need to be put in the language of Riemannian geometry, and we can switch back and forth between $\mb q$ and $\mb w$ spaces in our algorithm analysis without stating this fact.

\subsection{Useful Technical Lemmas and Proof Ideas for Orthogonal Dictionaries} \label{sec:geo_results_orth}
Proving Theorem~\ref{thm:geometry_orth} is conceptually straightforward: one shows that the expectation of each quantity of interest has the claimed property, and then proves that each quantity concentrates uniformly about its expectation. The detailed calculations are nontrivial. 

Note that  
\begin{align*}
\bb E_{\mb X_0}\brac{g\paren{\mb q; \mb X_0}} = \bb E_{\mb x \sim_{i.i.d.} \mathrm{BG}\paren{\theta}}\brac{h_{\mu}\paren{\mb q\paren{\mb w}^* \mb x}}. 
\end{align*}
The next three propositions show that in the expected function landscape, we see successively strongly convex region, large gradient region, and negative directional curvature region when moving away from zero, as depicted in Fig.~\ref{fig:large-sample-sphere} and sketched in Section~\ref{sec:overview_geometry}. 

\begin{proposition} \label{prop:geometry_asymp_curvature}
For any $\theta \in (0, 1/2)$, if $\mu \le c\min\Brac{\theta n^{-1}, n^{-5/4}}$, it holds for all $\mb w$ with $1/\paren{20\sqrt{5}} \le \norm{\mb w}{} \le \sqrt{(4n-1)/(4n)}$ that  
\begin{align*}
\frac{\mb w^* \nabla^2_{\mb w} \expect{h_{\mu}\left(\mb q^*\paren{\mb w} \mb x\right)} \mb w}{\norm{\mb w}{}^2}  \le -\frac{\theta }{2\sqrt{2\pi}}. 
\end{align*}
Here $c > 0$ is a constant. 
\end{proposition}
\begin{proof}
See Page~\pageref{sec:proof_geo_asym_curvature} under Section~\ref{sec:proof_geo_asym_curvature}. 
\end{proof}

\begin{proposition} \label{prop:geometry_asymp_gradient}
For any $\theta \in (0, 1/2)$, if $\mu \le 9/50$, it holds for all $\mb w$ with $\mu/(4\sqrt{2}) \le \norm{\mb w}{} \le 1/(20\sqrt{5})$ that 
\begin{align*}
\frac{\mb w^* \nabla_{\mb w} \expect{h_\mu (\mb q^*\paren{\mb w} \mb x)}}{\norm{\mb w}{}} \ge  \frac{\theta}{20\sqrt{2\pi}}. 
\end{align*}
\end{proposition}
\begin{proof}
See Page~\pageref{sec:proof_geo_asym_gradient} under Section~\ref{sec:proof_geo_asym_gradient}. 
\end{proof}

\begin{proposition} \label{prop:geometry_asymp_strong_convexity}
For any $\theta \in (0, 1/2)$, if $\mu \le 1/(20\sqrt{n})$, it holds for all $\mb w$ with $\norm{\mb w}{} \le \mu/(4\sqrt{2})$ that 
\begin{align*}
\nabla^2_{\mb w} \bb E[h_{\mu}\paren{\mb q^*\paren{\mb w}\mb x}]  \succeq \frac{\theta}{5\sqrt{2\pi}\mu} \mb I. 
\end{align*}
\end{proposition}
\begin{proof}
See Page~\pageref{sec:proof_geo_asym_strcvx} under Section~\ref{sec:proof_geo_asym_strcvx}. 
\end{proof}

To prove that the above hold qualitatively for finite $p$, i.e., the function $g\paren{\mb w; \mb X_0}$, we will need first prove that for a fixed $\mb w$ each of the quantity of interest concentrates about their expectation w.h.p., and the function is nice enough (Lipschitz) such that we can extend the results to all $\mb w$ via a discretization argument. The next three propositions provide the desired pointwise concentration results.  

\begin{proposition} \label{prop:concentration-gradient}
For every $\mb w \in \Gamma$, it holds that for any $t > 0$, 
\begin{align*}
\bb P\brac{\abs{\frac{\mb w^*\nabla g(\mb w; \mb X_0)
}{\norm{\mb w}{}}-\expect{\frac{\mb w^*\nabla g(\mb w; \mb X_0)
}{\norm{\mb w}{}}}}\geq t} \leq 2\exp\paren{-\frac{pt^2}{8n+4\sqrt{n}t}}.
\end{align*}
\end{proposition}
\begin{proof}
See Page~\pageref{proof:pt_cn_gradient} under Section~\ref{proof:cn_point}. 
\end{proof}

\begin{proposition}\label{prop:concentration-hessian-negative}
Suppose $0 <\mu \leq 1/\sqrt{n}$. For every $\mb w \in \Gamma$, it holds that for any $t > 0$, 
\begin{align*}
\bb P \brac{\abs{\frac{\mb w^*\nabla^2 g(\mb w; \mb X_0)\mb w}{\norm{\mb w}{}^2} - \bb E\brac{\frac{\mb w^*\nabla^2 g(\mb w; \mb X_0)\mb w}{\norm{\mb w}{}^2}}}\ge t}\leq 4\exp\paren{- \frac{p\mu^2t^2}{512n^2+32n\mu t}}.
\end{align*}
\end{proposition}
\begin{proof}
See Page~\pageref{proof:pt_cn_curvature} under Section~\ref{proof:cn_point}. 
\end{proof}

\begin{proposition}\label{prop:concentration-hessian-zero}
Suppose $0 <  \mu \le 1/\sqrt{n}$. For every $\mb w \in \Gamma \cap \set{\mb w: \norm{\mb w}{} \le 1/4}$, it holds that for any $t > 0$, 
\begin{align*}
\bb P\brac{\norm{\nabla^2 g(\mb w; \mb X_0) - \bb E\brac{\nabla^2 g(\mb w; \mb X_0)}}{} \geq t} \;\leq\; 4n\exp\paren{-\frac{p\mu^2t^2}{512n^2+32\mu n t}}. 
\end{align*} 
\end{proposition}
\begin{proof}
See Page~\pageref{proof:pt_cn_strcvx} under Section~\ref{proof:cn_point}. 
\end{proof}

The next three propositions provide the desired Lipschitz results. 

\begin{proposition}[Hessian Lipschitz]\label{prop:lip-hessian-negative}
Fix any $\rconcave \in \paren{0, 1}$. Over the set $\Gamma \cap \set{\mb w: \norm{\mb w}{} \ge \rconcave}$, $\mb w^* \nabla^2 g(\mb w; \mb X_0) \mb w/\norm{\mb w}{}^2$ is $\Lconcave$-Lipschitz with 
\begin{align*}
\Lconcave \le \frac{16n^3}{\mu^2} \norm{\mb X_0}{\infty}^3 + \frac{8n^{3/2}}{\mu \rconcave} \norm{\mb X_0}{\infty}^2 + \frac{48 n^{5/2} }{\mu} \norm{\mb X_0}{\infty}^2 + 96 n^{5/2} \norm{\mb X_0}{\infty}.
\end{align*}
\end{proposition} 
\begin{proof}
See Page~\pageref{proof:cn_lips_curvature} under Section~\ref{proof:cn_lips}. 
\end{proof}

\begin{proposition}[Gradient Lipschitz]\label{prop:lip-gradient}
Fix any $r_g \in \paren{0, 1}$. Over the set $\Gamma \cap \set{\mb w: \norm{\mb w}{} \ge r_g}$, $\mb w^* \nabla g( \mb w; \mb X_0 )/\norm{\mb w}{}$ is $L_g$-Lipschitz with 
\begin{align*}
L_g \le \frac{2 \sqrt{n} \norm{\mb X_0}{\infty}}{r_g} + 8 n^{3/2} \norm{\mb X_0}{\infty} + \frac{4 n^2}{\mu} \norm{\mb X_0}{\infty}^2.
\end{align*}
\end{proposition} 
\begin{proof}
See Page~\pageref{proof:cn_lips_gradient} under Section~\ref{proof:cn_lips}. 
\end{proof}

\begin{proposition}[Lipschitz for Hessian around zero]\label{prop:lip-hessian-zero}
Fix any $\rconvex \in (0, 1/2)$. Over the set $\Gamma \cap \set{\mb w: \norm{\mb w}{} \le \rconvex}$, $\nabla^2 g( \mb w; \mb X_0)$ is $\Lconvex$-Lipschitz with  
\begin{align*}
\Lconvex \;\le\;\frac{4n^2}{\mu^2} \norm{\mb X_0}{\infty}^3+\frac{4n}{\mu}\norm{\mb X_0}{\infty}^2 +  \frac{8\sqrt{2}\sqrt{n}}{\mu}\norm{\mb X_0}{\infty}^2 + 8\norm{\mb X_0}{\infty}. 
\end{align*}
\end{proposition}
\begin{proof}
See Page~\pageref{proof:cn_lips_strcvx} under Section~\ref{proof:cn_lips}. 
\end{proof}
Integrating the above pieces, Section~\ref{sec:proof_geometry_orth} provides a complete proof of Theorem~\ref{thm:geometry_orth}. 

\subsection{Extending to Complete Dictionaries} \label{sec:geo_results_comp}
As hinted in Section~\ref{sec:overview_geometry}, instead of proving things from scratch, we build on the results we have obtained for orthogonal dictionaries. In particular, we will work with the preconditioned data matrix 
\begin{align} \label{eq:precon_def}
\overline{\mb Y} \doteq \sqrt{p \theta} (\mb Y \mb Y^*)^{-1/2} \mb Y
\end{align}
and show that the function landscape $f\paren{\mb q; \overline{\mb Y}}$ looks qualitatively like that of orthogonal dictionaries (up to a global rotation), provided that $p$ is large enough.  

The next lemma shows $\overline{\mb Y}$ can be treated as being generated from an orthobasis with the same BG coefficients, plus small noise. 
\begin{lemma} \label{lem:pert_key_mag}
For any $\theta \in \paren{0, 1/2}$, suppose $\mb A_0$ is complete with condition number $\kappa\paren{\mb A_0}$ and $\mb X_0 \sim_{i.i.d.} \mathrm{BG}\paren{\theta}$. Provided $p \ge C\kappa^4\paren{\mb A_0} \theta n^2 \log (n \theta \kappa\paren{\mb A_0})$, one can write $\overline{\mb Y}$ as defined in~\eqref{eq:precon_def} as 
\begin{align*}
\overline{\mb Y} = \mb U \mb V^* \mb X_0 + \mb \Xi \mb X_0, 
\end{align*}
for a certain $\mb \Xi$ obeying $\norm{\mb \Xi}{} \le 20\kappa^4\paren{\mb A} \sqrt{\frac{\theta n \log p}{p}}$, with probability at least $1-p^{-8}$. Here $\mb U \mb \Sigma \mb V^* = \mathtt{SVD}\paren{\mb A_0}$, and $C > 0$ is a constant. 
\end{lemma}
\begin{proof}
See Page~\pageref{proof:comp_pert_bound} under Section~\ref{sec:proof_geometry_comp}. 
\end{proof}

Notice that $\mb U \mb V^*$ above is orthogonal, and that landscape of $f(\mb q; \overline{\mb Y})$ is simply a rotated version of that of $f(\mb q; \mb V \mb U^* \overline{\mb Y})$, or using the notation in the above lemma, that of $f(\mb q; \mb X_0 + \mb V \mb U^* \mb \Xi \mb X_0) = f(\mb q; \mb X_0 + \widetilde{\mb \Xi} \mb X_0)$ with $\widetilde{\mb \Xi} \doteq \mb V \mb U^* \mb \Xi$. So similar to the orthogonal case, it is enough to consider this ``canonical'' case, and its ``canonical'' reparametrization: 
\begin{align*}
g\paren{\mb w; \mb X_0 + \widetilde{\mb \Xi} \mb X_0} = \frac{1}{p}\sum_{k=1}^p h_{\mu}\paren{\mb q^*\paren{\mb w} \paren{\mb x_0}_k + \mb q^*\paren{\mb w} \widetilde{\mb \Xi} \paren{\mb x_0}_k}. 
\end{align*}
The following lemma provides quantitative comparison between the gradient and Hessian of $g\paren{\mb w; \mb X_0 + \widetilde{\mb \Xi} \mb X_0}$ and that of $g\paren{\mb w; \mb X_0}$. 
\begin{lemma} \label{lem:pert_key_grad_hess}
For all $\mb w \in \Gamma$, 
\begin{align*}
\norm{\nabla_{\mb w} g(\mb w; \mb X_0 + \widetilde{\mb \Xi} \mb X_0) - \nabla_{\mb w{}} g\paren{\mb w; \mb X_0} }{} & \le C_a\frac{n}{\mu} \log\paren{np} \|\widetilde{\mb \Xi}\|, \\
\norm{\nabla_{\mb w}^2 g(\mb w; \mb X_0 + \widetilde{\mb \Xi} \mb X_0) - \nabla_{\mb w}^2 g\paren{\mb w; \mb X_0}}{} & \le C_b \max\set{\frac{n^{3/2}}{\mu^2}, \frac{n^2}{\mu}} \log^{3/2}\paren{np} \|\widetilde{\mb \Xi}\|
\end{align*}
with probability at least $1 - \theta\paren{np}^{-7} - \exp\paren{-0.3 \theta np}$. Here $C_a, C_b$ are positive constants. 
\end{lemma}
\begin{proof}
See Page~\pageref{proof:comp_pert_bound2} under Section~\ref{sec:proof_geometry_comp}. 
\end{proof}
Combining the above two lemmas, it is easy to see when $p$ is large enough, $\|\widetilde{\mb \Xi}\| = \norm{\mb \Xi}{}$ is then small enough (Lemma~\ref{lem:pert_key_mag}), and hence changes to the gradient and Hessian caused by the perturbation are small. This gives the results presented in Theorem~\ref{thm:geometry_comp}; see Section~\ref{sec:proof_geometry_comp} for a detailed proof. In particular, for the $p$ chosen in Theorem~\ref{thm:geometry_comp}, it holds that 
\begin{align} \label{eq:pert_upper_bound}
\|\widetilde{\mb \Xi}\| \le c c_\star \theta \paren{\max\set{\frac{n^{3/2}}{\mu^2}, \frac{n^2}{\mu}} \log^{3/2}\paren{np}}^{-1}
\end{align}
for a certain constant $c$ which can be made arbitrarily small by making the constant $C$ in $p$ large.

\section{Discussion} \label{sec:discuss}
The dependency of $p$ on $n$ and other parameters could be suboptimal due to several factors: (1) The $\ell^1$ proxy. Derivatives of the $\log\cosh$ function we adopted entail the $\tanh$ function, which is not amenable to effective approximation and affects the sample complexity; (2) Space of geometric characterization. It seems working directly on the sphere (i.e., in the $\mb q$ space) could simplify and possibly improve certain parts of the analysis; (3) Dealing with the complete case. Treating the complete case directly, rather than using (pessimistic) bounds to treat it as a perturbation of the orthogonal case, is very likely to improve the sample complexity. Particularly, general linear transforms may change the space significantly, such that preconditioning and comparing to the orthogonal transforms may not be the most efficient way to proceed. 

It is possible to extend the current analysis to other dictionary settings. Our geometric structures (and algorithms) allow plug-and-play noise analysis. Nevertheless, we believe a more stable way of dealing with noise is to directly extract the whole dictionary, i.e., to consider geometry and optimization (and perturbation) over the orthogonal group. This will require additional nontrivial technical work, but likely feasible thanks to the relatively complete knowledge of the orthogonal group~\cite{edelman1998geometry, absil2009}. A substantial leap forward would be to extend the methodology to recovery of \emph{structured} overcomplete dictionaries, such as tight frames. Though there is no natural elimination of one variable, one can consider the marginalization of the objective function w.r.t. the coefficients and work with implicit functions. \footnote{This recent work~\cite{arora2015simple} on overcomplete DR has used a similar idea. The marginalization taken there is near to the global optimum of one variable, where the function is well-behaved. Studying the global properties of the marginalization may introduce additional challenges.} For the coefficient model, as we alluded to in Section~\ref{sec:lit_review}, our analysis and results likely can be carried through to coefficients with statistical dependence and physical constraints. 

The connection to ICA we discussed in Section~\ref{sec:lit_review} suggests our geometric characterization and algorithms can be modified for the ICA problem. This likely will provide new theoretical insights and computational schemes to ICA. In the surge of theoretical understanding of nonconvex heuristics~\cite{keshavan2010matrix, jain2013low, hardt2014understanding, hardt2014fast, netrapalli2014non, jain2014fast, netrapalli2013phase, candes2015phase, jain2014provable, anandkumar2014guaranteed, yi2013alternating, lee2013near, qu2014finding, lee2013near, agarwal2013learning, agarwal2013exact, arora2013new, arora2015simple, arora2014more}, the initialization plus local refinement strategy mostly differs from practice, whereby random initializations seem to work well, and the analytic techniques developed in that line are mostly fragmented and highly specialized. The analytic and algorithmic framework we developed here holds promise to providing a coherent account of these problems, see~\cite{sun2015nonconvex}. In particular, we have intentionally separated the geometric characterization and algorithm development, hoping to making both parts modular. It is interesting to see how far we can streamline the geometric characterization. Moreover, the separation allows development of more provable and practical algorithms, say in the direction of~\cite{ge2015escaping}. 

\section{Proofs of Technical Results} \label{sec:proof_geometry}
In this section, we provide complete proofs for technical results stated in Section~\ref{sec:geometry}. Before that, let us introduce some convenient notations and common results. Since we deal with BG random variables and random vectors, it is often convenient to write such vector explicitly as $\mb x = \brac{\Omega_1 v_1, \dots, \Omega_n v_n} = \bm \Omega \odot \mb v$, where $\Omega_1, \dots, \Omega_n$ are i.i.d.\ Bernoulli and $v_1, \dots, v_n$ are i.i.d.\ standard normal. For a particular realization of such random vector, we will denote the support as $\mc I \subset [n]$. Due to the particular coordinate map in use, we will often refer to subset $\mc J \doteq \mc I \setminus\Brac{n}$ and the random vectors $\overline{\mb x} \doteq \brac{\Omega_1 v_1, \dots, \Omega_{n-1} v_{n-1}}$ and $\overline{\mb v} \doteq \brac{v_1, \dots, v_{n-1}}$ in $\R^{n-1}$. \js{Naturally, $x_n$ and $q_n(\mb w)$ denote the last coordinates in $\mb x$ and $\mb q$, respectively. Hence, by our notation, $\mb q^*(\mb w) \mb x = \mb w^* \ol{\mb x} + q_n(\mb w) x_n$. By Lemma \ref{lem:derivatives_basic_surrogate} and chain rules, the following are immediate:} 
\begin{align}
\nabla_{\mb w} h_{\mu}\paren{\mb q^*\paren{\mb w} \mb x} & = \tanh\paren{\frac{\mb q^*\paren{\mb w} \mb x}{\mu}} \paren{\overline{\mb x} - \frac{x_n}{q_n\paren{\mb w}} \mb w},\label{eqn:lse-gradient} \\
\nabla^2_{\mb w} h_{\mu}\paren{\mb q^*\paren{\mb w} \mb x} & = \frac{1}{\mu} \brac{1-\tanh^2\paren{\frac{\mb q^*\paren{\mb w} \mb x}{\mu}}} \paren{\overline{\mb x} - \frac{x_n}{q_n\paren{\mb w}} \mb w} \paren{\overline{\mb x} - \frac{x_n}{q_n\paren{\mb w}} \mb w}^* \nonumber \\
& \qquad - x_n \tanh\paren{\frac{\mb q^*\paren{\mb w} \mb x}{\mu}} \paren{\frac{1}{q_n\paren{\mb w}} \mb I + \frac{1}{q_n^3\paren{\mb w}} \mb w \mb w^*}. \label{eqn:lse-hessian}
\end{align}

\subsection{Proofs for Section~\ref{sec:geo_results_orth}}
\subsubsection{Proof of Proposition~\ref{prop:geometry_asymp_curvature}} \label{sec:proof_geo_asym_curvature}
The proof involves some delicate analysis, particularly polynomial approximation of the function $f\paren{t} = 1/\paren{1+t}^2$ over $t \in \left[0, 1\right]$. This is naturally induced by the $1-\tanh^2\paren{\cdot}$ function. The next lemma characterizes one polynomial approximation of $f\paren{t}$. 

\begin{lemma} \label{lem:neg_curvature_norm_bound} 
Consider $f(t) = 1/(1+t)^2$ for $t \in \brac{0,1}$. For every $T > 1$, there is a sequence $b_0, b_1, \dots$, with $\norm{\mb b}{\ell^1} = T < \infty$, such that the polynomial $p(t) = \sum_{k=0}^\infty b_k t^k$ satisfies 
\begin{align*}
\norm{f-p}{L^1[0,1]} \;\le\; \frac{1}{2\sqrt{T}}, \quad \norm{f-p}{L^\infty[0,1]} \;\le\; \frac{1}{\sqrt{T}}. 
\end{align*}
In particular, one can choose $b_k = (-1)^k(k+1)\beta^k$ with $\beta = 1-1/\sqrt{T} < 1$ such that 
\begin{align*}
p\paren{t} = \frac{1}{\paren{1+\beta t}^2} = \sum_{k=0}^\infty (-1)^k(k+1)\beta^k t^k. 
\end{align*}
Moreover, such sequence satisfies $0<\sum_{k=0}^\infty \frac{b_k}{(1+k)^3}<\sum_{k=0}^\infty \frac{\abs{b_k}}{(1+k)^3}<2$. 
\end{lemma}
\begin{proof}
See page~\pageref{proof:poly_approx_tanh} under Section~\ref{sec:app_aux_results}. 
\end{proof}

\begin{lemma}\label{lem:neg_curvature_tanh_square_1} 
Let $X\sim \mc N\paren{0,\sigma_X^2}$ and $Y\sim \mc N\paren{0,\sigma_Y^2}$ be independent. We have
\js{
\begin{multline*}
\bb E\brac{\paren{1-\tanh^2\paren{\frac{X+Y}{\mu}} } X^2 \indicator{X+Y>0} } \le
 \frac{1}{\sqrt{2\pi}} \frac{\mu \sigma_X^2 \sigma_Y^2}{\paren{\sigma_X^2 + \sigma_Y^2}^{3/2}} + \frac{3}{4\sqrt{2\pi}} \frac{\sigma_X^2 \mu^3}{\paren{\sigma_X^2 + \sigma_Y^2}^{5/2}} \paren{3\mu^2 + 4 \sigma_X^2}. 
\end{multline*}
}
\end{lemma}

\begin{proof}
For $X + Y \ge 0$, let $Z = \exp\paren{-2(X+Y)/\mu}  \in [0,1]$, then 
\begin{align*}
1-\tanh^2\paren{\frac{X + Y}{\mu}} = \frac{4Z}{\paren{1+Z}^2}.
\end{align*} 
First fix any $T > 1$. By Lemma~\ref{lem:neg_curvature_norm_bound}, we choose the polynomial $p_{\beta}\paren{Z} = \frac{1}{\paren{1+\beta Z}^2}$ with $\beta = 1 - 1/\sqrt{T}$ to upper bound $f\paren{Z} = \frac{1}{\paren{1+Z}^2}$. So we have 
\begin{align*}
\bb E\brac{\paren{1-\tanh^2\paren{\frac{X+Y}{\mu}} } X^2 \indicator{X+Y>0} } 
& = 4\bb E\brac{Zf\paren{Z} X^2 \indicator{X+Y>0} } \nonumber \\
& \le 4\bb E\brac{Zp_{\beta}\paren{Z} X^2 \indicator{X+Y>0} } \nonumber \\
& = 4\sum_{k=0}^\infty \Brac{b_k\bb E\brac{Z^{k+1}X^2\indicator{X+Y>0} }  }, 
\end{align*}
where $b_k = (-1)^k(k+1)\beta^k$, and the exchange of infinite summation and expectation above is justified due to 
\begin{align*}
\sum_{k=0}^\infty \abs{b_k}\bb E\brac{Z^{k+1}X^2\indicator{X+Y>0} } \le \sum_{k=0}^\infty \abs{b_k}\bb E\brac{X^2\indicator{X+Y>0} } \le \sigma_X^2  \sum_{k=0}^\infty \abs{b_k} < \infty
\end{align*} 
and the dominated convergence theorem (see, e.g., theorem 2.24 and 2.25 of~\cite{folland1999real}). By Lemma~\ref{lem:aux_asymp_proof_a}, we have 
\begin{align*}
& \sum_{k=0}^\infty \Brac{b_k\bb E\brac{Z^{k+1}X^2\indicator{X+Y>0} }  } \\
=\; & \sum_{k=0}^\infty \paren{-\beta}^k \paren{k+1} \left[\paren{\sigma_X^2+\frac{4\paren{k+1}^2}{\mu^2}\sigma_X^4}\exp\paren{\frac{2\paren{k+1}^2}{\mu^2}\paren{\sigma_X^2 + \sigma_Y^2}}\Phi^c\paren{\frac{2\paren{k+1}}{\mu}\sqrt{\sigma_X^2 + \sigma_Y^2}} \right. \\
& \qquad \left. - \frac{2\paren{k+1}}{\mu} \frac{\sigma_X^4}{\sqrt{2\pi}\sqrt{\sigma_X^2 + \sigma_Y^2}}\right] \\
\le \; & \frac{1}{\sqrt{2\pi}}\sum_{k=0}^\infty \paren{-\beta}^k \paren{k+1}\brac{\frac{\sigma_X^2 \mu}{2\paren{k+1}\sqrt{\sigma_X^2 + \sigma_Y^2}} - \frac{\sigma_X^2 \mu^3}{8\paren{k+1}^3 \paren{\sigma_X^2 + \sigma_Y^2}^{3/2}} - \frac{\mu \sigma_X^4}{2\paren{k+1}\paren{\sigma_X^2 + \sigma_Y^2}^{3/2}}} \nonumber \\
& \qquad + \frac{3}{\sqrt{2\pi}}\sum_{k=0}^\infty \beta^k \paren{k+1} \paren{\sigma_X^2+\frac{4\paren{k+1}^2}{\mu^2}\sigma_X^4} \frac{\mu^5}{32\paren{k+1}^5 \paren{\sigma_X^2 + \sigma_Y^2}^{5/2}}, 
\end{align*}
where we have applied Type I upper and lower bounds for $\Phi^c\paren{\cdot}$ to even $k$ and odd $k$ respectively and rearranged the terms to obtain the last line. Using the following estimates (see Lemma~\ref{lem:neg_curvature_norm_bound})
\begin{align*}
\sum_{k=0}^\infty \paren{-\beta}^k = \frac{1}{1+\beta}, \quad \sum_{k=0}^\infty \frac{b_k}{\paren{k+1}^3} \ge 0, \quad  \sum_{k=0}^\infty \frac{\abs{b_k}}{\paren{k+1}^5} \le \sum_{k=0}^\infty \frac{\abs{b_k}}{\paren{k+1}^3} \le 2, 
\end{align*}
we obtain 
\begin{align*}
\sum_{k=0}^\infty \Brac{b_k\bb E\brac{Z^{k+1}X^2\indicator{X+Y>0} } } 
 \le 
 \frac{1}{2\sqrt{2\pi}} \frac{\mu \sigma_X^2 \sigma_Y^2}{\paren{\sigma_X^2 + \sigma_Y^2}^{3/2}} \frac{1}{1+\beta} + \frac{3}{16\sqrt{2\pi}} \frac{\sigma_X^2 \mu^3}{\paren{\sigma_X^2 + \sigma_Y^2}^{5/2}} \paren{3\mu^2 + 4 \sigma_X^2}. 
\end{align*}
\js{Since the above holds for any $T > 1$, we obtain the claimed result by letting $T \to \infty$ such that $\beta \to 1$. }
\end{proof}

\js{
\begin{lemma}\label{lem:neg_curvature_tanh_square_3} 
Let $X\sim \mc N\paren{0,\sigma_X^2}$ and $Y\sim \mc N\paren{0,\sigma_Y^2}$ be independent. We have
\begin{align*}
\bb E\brac{\paren{1-\tanh^2\paren{\frac{X+Y}{\mu}} } XY \indicator{X+Y>0} } \ge -\frac{1}{\sqrt{2\pi}} \frac{\mu \sigma_X^2 \sigma_Y^2}{\paren{\sigma_X^2 + \sigma_Y^2}^{3/2}} - \frac{3}{\sqrt{2\pi}} \frac{\sigma_X^2 \sigma_Y^2 \mu^3}{\paren{\sigma_X^2 + \sigma_Y^2}^{5/2}}.
\end{align*}
\end{lemma}

\begin{proof}
For $X + Y > 0$, let $z = \exp\paren{-2(X + Y)/\mu } \in [0,1]$, then 
\begin{align*}
1-\tanh^2\paren{\frac{X+Y}{\mu}} = \frac{4Z}{\paren{1+Z}^2}.\
\end{align*} 
First fix any $T > 1$. By Lemma~\ref{lem:neg_curvature_norm_bound}, we choose the polynomial $p_{\beta}\paren{Z} = \frac{1}{\paren{1+\beta Z}^2}$ with $\beta = 1 - 1/\sqrt{T}$ to upper bound $f\paren{Z} = \frac{1}{\paren{1+Z}^2}$. So we have 
\begin{align*}
& \bb E\brac{\paren{1-\tanh^2\paren{\frac{X+Y}{\mu}} } XY \indicator{X+Y>0} } \\
=\; & 4\bb E\brac{Zf\paren{Z} XY \indicator{X+Y>0} } \nonumber \\
=\; & 4\bb E\brac{Zp_{\beta}\paren{Z} XY \indicator{X+Y>0} } - 4\bb E\brac{Z\paren{p_{\beta}\paren{Z} -f(Z)} XY \indicator{X+Y>0} }.  
\end{align*}
Now the first term can be rewritten as 
\begin{align*}
4\bb E\brac{Zp_{\beta}\paren{Z} XY \indicator{X+Y>0} } = 4\sum_{k=0}^\infty \Brac{b_k\bb E\brac{Z^{k+1}X^2\indicator{X+Y>0} }  },
\end{align*}
where $b_k = (-1)^k(k+1)\beta^k$, and exchange of infinite summation and expectation is justified, due to 
\begin{align*}
\sum_{k=0}^\infty \abs{b_k}\bb E\brac{Z^{k+1}XY\indicator{X+Y>0} } \le \sum_{k=0}^\infty \abs{b_k}\bb E\brac{\abs{XY}\indicator{X+Y>0} } \le \max\set{\sigma_X^2, \sigma_Y^2 } \sum_{k=0}^\infty \abs{b_k} < \infty
\end{align*} 
and the dominated convergence theorem (see, e.g., theorem 2.24 and 2.25 of~\cite{folland1999real}). By Lemma~\ref{lem:aux_asymp_proof_a}, we have 
\begin{align*}
& \sum_{k=0}^\infty \Brac{b_k\bb E\brac{Z^{k+1}XY\indicator{X+Y>0} }  } \\
=\; & \sum_{k=0}^\infty \paren{-\beta}^k \paren{k+1} \left[\frac{4\paren{k+1}^2}{\mu^2}\sigma_X^2 \sigma_Y^2 \exp\paren{\frac{2\paren{k+1}^2}{\mu^2}\paren{\sigma_X^2 + \sigma_Y^2}}\Phi^c\paren{\frac{2\paren{k+1}}{\mu}\sqrt{\sigma_X^2 + \sigma_Y^2}} \right. \\
& \qquad \left. - \frac{2\paren{k+1}}{\mu} \frac{\sigma_X^2 \sigma_Y^2}{\sqrt{2\pi}\sqrt{\sigma_X^2 + \sigma_Y^2}}\right] \\
\ge \; & -\frac{1}{\sqrt{2\pi}}\sum_{k=0}^\infty \paren{-\beta}^k \paren{k+1} \frac{\mu \sigma_X^2 \sigma_Y^2}{2(k+1)\paren{\sigma_X^2 + \sigma_Y^2}^{3/2}} -\frac{3}{8\sqrt{2\pi}}\sum_{k=0}^\infty \beta^k \paren{k+1} \frac{\mu^3 \sigma_X^2 \sigma_Y^2}{\paren{k+1}^3 \paren{\sigma_X^2 + \sigma_Y^2}^{5/2}}, 
\end{align*}
where we have applied Type I lower and upper bounds for $\Phi^c\paren{\cdot}$ to even $k$ and odd $k$ respectively and rearranged the terms to obtain the last line. Using the following estimates (see Lemma~\ref{lem:neg_curvature_norm_bound})
\begin{align*}
\sum_{k=0}^\infty \paren{-\beta}^k = \frac{1}{1+\beta}, \quad \sum_{k=0}^\infty \frac{\abs{b_k}}{\paren{k+1}^3} \le 2, 
\end{align*}
we obtain 
\begin{align*}
\sum_{k=0}^\infty \Brac{b_k\bb E\brac{Z^{k+1}XY\indicator{X+Y>0} } } 
 \ge 
 -\frac{1}{2\sqrt{2\pi}} \frac{\mu \sigma_X^2 \sigma_Y^2}{\paren{\sigma_X^2 + \sigma_Y^2}^{3/2}} \frac{1}{1+\beta} - \frac{3}{4\sqrt{2\pi}} \frac{\sigma_X^2 \sigma_Y^2 \mu^3}{\paren{\sigma_X^2 + \sigma_Y^2}^{5/2}}.  
\end{align*}
For the second term, by Lemma \ref{lem:neg_curvature_norm_bound}, we have
\begin{align*}
\bb E\brac{\paren{p_{\beta}(Z)-f(Z)} Z XY\indicator{X+Y>0}} 
& \le \norm{p-f}{L^\infty[0,1]} \bb E\brac{Z\abs{XY}\indicator{X+Y>0}} \\
& \le \norm{p-f}{L^\infty[0,1]} \expect{\abs{X}} \expect{\abs{Y}} \quad (\text{$Z \le 1$, $X, Y$ independent})\\
& \le \frac{2}{\pi\sqrt{T}} \sigma_X \sigma_Y. 
\end{align*}
Thus, 
\begin{align*}
\bb E\brac{\paren{1-\tanh^2\paren{\frac{X+Y}{\mu}} } XY \indicator{X+Y>0} } \ge -\frac{2}{\sqrt{2\pi}} \frac{\mu \sigma_X^2 \sigma_Y^2}{\paren{\sigma_X^2 + \sigma_Y^2}^{3/2}} \frac{1}{1+\beta} - \frac{3}{\sqrt{2\pi}} \frac{\sigma_X^2 \sigma_Y^2 \mu^3}{\paren{\sigma_X^2 + \sigma_Y^2}^{5/2}} - \frac{8}{\pi\sqrt{T}} \sigma_X \sigma_Y. 
\end{align*}
Since the above bound holds for any $T > 1$, we obtain the claimed result by letting $T \to \infty$ such that $\beta \to 1$ and $1/\sqrt{T} \to 0$. 
\end{proof}
}

\begin{lemma} \label{lem:neg_curvature_tanh_square_2}
Let $X\sim \mc N\paren{0,\sigma_X^2}$ and $Y\sim \mc N\paren{0,\sigma_Y^2}$ be independent. We have 
\js{
\begin{align*}
\bb E\left[ \tanh\left( \frac{X + Y }{\mu} \right) X \right] \ge 
\sqrt{\frac{2}{\pi}}\frac{\sigma_X^2}{\sqrt{\sigma_X^2 + \sigma_Y^2}} - \frac{2\sigma_X^2\mu^2}{\sqrt{2\pi} \paren{\sigma_X^2 + \sigma_Y^2}^{3/2}} - \frac{3\sigma_X^2\mu^4}{2\sqrt{2\pi}\paren{\sigma_X^2 + \sigma_Y^2}^{5/2}}. 
\end{align*}}
\end{lemma}

\begin{proof}
By Lemma \ref{lem:aux_asymp_proof_a}, we know 
\begin{align*}
\bb E \brac{\tanh\paren{\frac{X+Y}{\mu}} X  } \;=\; \frac{\sigma_X^2}{\mu } \bb E\brac{1- \tanh^2\paren{\frac{X+Y}{\mu}} }
\end{align*}
Similar to the proof of the above lemma, for $X + Y > 0$, let $Z \doteq \exp\paren{-2\frac{X + Y}{\mu}}$ and $f\paren{Z} \doteq \frac{1}{\paren{1+Z}^2}$. First fix any $T > 1$. We will use $4zp_{\beta}\paren{Z} = \frac{4Z}{\paren{1+\beta Z}^2}$ to approximate the $1-\tanh^2\paren{\frac{X + Y}{\mu}} = 4Zf\paren{Z}$ function from above, where again $\beta = 1 - 1/\sqrt{T}$. So we obtain 
\begin{align*}
\bb E\brac{1- \tanh^2\paren{\frac{X+Y}{\mu}} } 
& = 8\expect{f\paren{Z} Z \indicator{X + Y > 0}} \nonumber \\
& = 8\expect{p_{\beta}\paren{Z} Z \indicator{X + Y > 0}} - 8\expect{\paren{p_{\beta}\paren{Z} - f\paren{Z} }Z \indicator{X + Y > 0}}. 
\end{align*}
Now for the first term, we have 
\begin{align*}
\expect{p_{\beta}\paren{Z} Z \indicator{X + Y > 0}} = \sum_{k=0}^\infty b_k \bb E\brac{ Z^{k+1} \indicator{X+Y>0} }, 
\end{align*}
justified as $\sum_{k=0}^\infty \abs{b_k} \bb E\brac{ Z^{k+1} \indicator{X+Y>0} } \le \sum_{k=0}^\infty \abs{b_k} < \infty$ making the dominated convergence theorem (see, e.g., theorem 2.24 and 2.25 of~\cite{folland1999real}) applicable. To proceed, from Lemma~\ref{lem:aux_asymp_proof_a}, we obtain
\begin{align*}
& \sum_{k=0}^\infty  b_k \bb E\brac{ Z^{k+1} \indicator{X+Y>0} } \nonumber \\
=\; & \sum_{k=0}^\infty \paren{-\beta}^k \paren{k+1} \exp\paren{\frac{2}{\mu^2}\paren{k+1}^2 \paren{\sigma_X^2 + \sigma_Y^2}} \Phi^c\paren{\frac{2}{\mu}\paren{k+1} \sqrt{\sigma_X^2 + \sigma_Y^2}} \\
\ge\; & \frac{1}{\sqrt{2\pi}} \sum_{k = 0}^\infty \paren{-\beta}^k \paren{k+1}\paren{\frac{\mu}{2\paren{k+1}\sqrt{\sigma_X^2 + \sigma_Y^2}} - \frac{\mu^3}{8\paren{k+1}^3\paren{\sigma_X^2 + \sigma_Y^2}^{3/2}}} \\
& \qquad - \frac{3}{\sqrt{2\pi}} \sum_{k=0}^\infty \beta^k \paren{k+1} \frac{\mu^5}{32\paren{k+1}^5\paren{\sigma_X^2 + \sigma_Y^2}^{5/2}}, 
\end{align*}
where we have applied Type I upper and lower bounds for $\Phi^c\paren{\cdot}$ to odd $k$ and even $k$ respectively and rearranged the terms to obtain the last line. Using the following estimates (see Lemma~\ref{lem:neg_curvature_norm_bound})
\begin{align*}
\sum_{k=0}^\infty \paren{-\beta}^k = \frac{1}{1+\beta}, \quad  0 \le \sum_{k=0}^\infty \frac{b_k}{\paren{k+1}^3} \le \sum_{k=0}^\infty \frac{\abs{b_k}}{\paren{k+1}^5} \le \sum_{k=0}^\infty \frac{\abs{b_k}}{\paren{k+1}^3} \le 2, 
\end{align*}
we obtain 
\begin{align*}
\sum_{k=0}^\infty  b_k \bb E\brac{ Z^{k+1} \indicator{X+Y>0} } \ge 
\frac{\mu}{2\sqrt{2\pi} \sqrt{\sigma_X^2 + \sigma_Y^2}} \frac{1}{1+\beta} - \frac{\mu^3}{4\sqrt{2\pi}\paren{\sigma_X^2 + \sigma_Y^2}^{3/2}} - \frac{3\mu^5}{16\sqrt{2\pi}\paren{\sigma_X^2 + \sigma_Y^2}^{5/2}}. 
\end{align*}

For the second term, by Lemma \ref{lem:aux_asymp_proof_a} and Lemma \ref{lem:neg_curvature_norm_bound}, we have
\begin{align*}
\bb E\brac{\paren{p_{\beta}(Z)-f(Z)} Z\indicator{X+Y>0}} \le \norm{p-f}{L^\infty[0,1]} \bb E\brac{Z\indicator{X+Y>0}}\le \frac{\mu}{2\sqrt{2\pi T} \sqrt{\sigma_X^2+\sigma_Y^2}}, 
\end{align*}
where we have also used Type I upper bound for $\Phi^c\paren{\cdot}$. Combining the above estimates, we get
\begin{multline*}
\bb E\brac{\tanh\paren{\frac{X+Y}{\mu}} X} \ge 
\frac{4\sigma_X^2}{\sqrt{2\pi}\sqrt{\sigma_X^2 + \sigma_Y^2}} \paren{\frac{1}{1+\beta} - \frac{1}{\sqrt{T}}} - \frac{2\sigma_X^2\mu^2}{\sqrt{2\pi} \paren{\sigma_X^2 + \sigma_Y^2}^{3/2}} - \frac{3\sigma_X^2\mu^4}{2\sqrt{2\pi}\paren{\sigma_X^2 + \sigma_Y^2}^{5/2}}. 
\end{multline*}
\js{Since the above holds for any $T > 1$, we obtain the claimed result by letting $T \to \infty$, such that $\beta \to 1$ and $1/\sqrt{T} \to 0$. }
\end{proof}

\begin{proof}{(\textbf{of Proposition~\ref{prop:geometry_asymp_curvature}})}
For any $i \in [n-1]$, we have 
\begin{align*}
\int_0^1 \int_{\mb x} \abs{\frac{\partial}{\partial w_i} h_{\mu}\paren{\mb q^*\paren{\mb w} \mb x}} \mu\paren{d\mb x}\; d w_i \le \int_0^1 \int_{\mb x} \paren{\abs{x_i} + \abs{x_n} \frac{1}{q_n\paren{\mb w}}} \mu\paren{d\mb x}\; d w_i < \infty. 
\end{align*}
Hence by Lemma~\ref{lemma:exchange_diff_int} we obtain $\frac{\partial }{\partial w_i}\expect{h_{\mu}\paren{\mb q^*\paren{\mb w} \mb x}} = \expect{\frac{\partial }{\partial w_i} h_{\mu}\paren{\mb q^*\paren{\mb w} \mb x}}$. Moreover for any $j \in [n-1]$, 
\begin{multline*}
\int_0^1 \int_{\mb x} \abs{\frac{\partial^2}{\partial w_j \partial w_i} h_{\mu}\paren{\mb q^*\paren{\mb w} \mb x}} \mu\paren{d\mb x}\; d w_j \le \\
\int_0^1 \int_{\mb x} \brac{\frac{1}{\mu}\paren{\abs{x_i} + \frac{\abs{x_n}}{q_n\paren{\mb w}}}\paren{\abs{x_j} + \frac{\abs{x_n}}{q_n\paren{\mb w}}} + \abs{x_n}\paren{\frac{1}{q_n\paren{\mb w}} + \frac{1}{q_n^3\paren{\mb w}}}} \mu\paren{d\mb x}\; d w_i < \infty. 
\end{multline*}
Invoking Lemma~\ref{lemma:exchange_diff_int} again we obtain 
\begin{align*}
\frac{\partial^2}{\partial w_j \partial w_i} \expect{h_{\mu}\paren{\mb q^*\paren{\mb w} \mb x}} = \frac{\partial}{\partial w_j}\expect{\frac{\partial }{\partial w_i} h_{\mu}\paren{\mb q^*\paren{\mb w} \mb x}} = \expect{\frac{\partial^2}{\partial w_j \partial w_i} h_{\mu}\paren{\mb q^*\paren{\mb w} \mb x}}. 
\end{align*}
The above holds for any pair of $i, j \in [n-1]$, so it follows that 
\begin{align*}
\nabla^2_{\mb w}\expect{h_{\mu}\paren{\mb q^*\paren{\mb w} \mb x}} = \expect{\nabla^2_{\mb w} h_{\mu}\paren{\mb q^*\paren{\mb w} \mb x}}. 
\end{align*}
Hence it is easy to see that 
\begin{align*}
& \mb w^* \nabla^2_{\mb w} \expect{h_{\mu}\left(\mb q^*\paren{\mb w} \mb x\right)} \mb w  \nonumber \\
=\; &  \underbrace{\frac{1}{\mu}\expect{\left(1-\tanh^2\left(\frac{\mb q^*\paren{\mb w} \mb x}{\mu}\right)\right) \left(\mb w^* \overline{\mb x} - \frac{x_n}{q_n\paren{\mb w}}\norm{\mb w}{}^2\right)^2}}_{(\mc A)} - \underbrace{\expect{\tanh\left(\frac{\mb q^*\paren{\mb w} \mb x}{\mu}\right) \frac{x_n}{q_n^3\paren{\mb w}}\norm{\mb w}{}^2}}_{(\mc B)}. 
\end{align*}
\js{
\begin{enumerate}
\item \textbf{An upper bound for $(\mc A)$.}  
When $x_n$ is not in support set of $\mb x$, the term reduces to 
\begin{align*}
& \frac{2}{\mu}\expect{\left(1-\tanh^2\left(\frac{\mb w^*\overline{\mb x}}{\mu}\right)\right) \paren{\mb w^* \overline{\mb x}}^2 \indicator{\mb w^* \overline{\mb x} > 0}} \\
\le\; & \frac{8}{\mu}\expect{\exp\paren{-2\frac{\mb w^* \overline{\mb x}}{\mu}} \paren{\mb w^* \overline{\mb x}}^2 \indicator{\mb w^* \overline{\mb x} > 0}} \quad (\text{Lemma~\ref{lem:derivatives_basic_surrogate}})\\
\le\; & 8 \exp(-2)\mu, 
\end{align*}
where to obtain the last line we used that $t \mapsto \exp(-2t/\mu)t^2$ for $t > 0$ is maximized at $\mu$. 

When $x_n$ is in the support set, we expand the square term inside the expectation and obtain 
\begin{align*}
(\mc A)_{x_n \ne 0} & = \frac{2}{\mu} \bb E_{\mc J}\bb E_{\mb v}\brac{\left(1-\tanh^2\left(\frac{\mb w^*_{\mc J}\overline{\mb v} + q_n\paren{\mb w} v_n}{\mu}\right)\right) \left(\mb w^*_{\mc J} \overline{\mb v}\right)^2 \indicator{\mb w^*_{\mc J}\overline{\mb v} + q_n\paren{\mb w} v_n > 0}} \\
& \qquad + \frac{2}{\mu} \frac{\norm{\mb w}{}^4}{q_n^4\paren{\mb w}} \bb E_{\mc J} \bb E_{\mb v}\brac{\left(1-\tanh^2\left(\frac{\mb w^*_{\mc J}\overline{\mb v} + q_n\paren{\mb w} v_n}{\mu}\right)\right) \left(q_n\paren{\mb w} v_n\right)^2 \indicator{\mb w^*_{\mc J}\overline{\mb v} + q_n\paren{\mb w} v_n > 0}} \\
& \qquad -  \frac{4}{\mu} \frac{\norm{\mb w}{}^2}{q_n^2\paren{\mb w}} \bb E_{\mc J} \bb E_{\mb v}\brac{\left(1-\tanh^2\left(\frac{\mb w^*_{\mc J}\overline{\mb v} + q_n\paren{\mb w} v_n}{\mu}\right)\right) \left(\mb w^*_{\mc J} \overline{\mb v}\right) \left(q_n\paren{\mb w} v_n\right) \indicator{\mb w^*_{\mc J}\overline{\mb v} + q_n\paren{\mb w} v_n > 0}}\\
& =  \frac{2}{\mu} \bb E_{\mc J}\bb E_{X, Y}\brac{\left(1-\tanh^2\left(\frac{X + Y}{\mu}\right)\right) Y^2 \indicator{X + Y > 0}} \\
& \qquad + \frac{2}{\mu} \frac{\norm{\mb w}{}^4}{q_n^4\paren{\mb w}} \bb E_{\mc J} \bb E_{X, Y}\brac{\left(1-\tanh^2\left(\frac{X + Y}{\mu}\right)\right) X^2 \indicator{X + Y > 0}} \\
& \qquad - \frac{4}{\mu} \frac{\norm{\mb w}{}^2}{q_n^2\paren{\mb w}} \bb E_{\mc J} \bb E_{X, Y}\brac{\left(1-\tanh^2\left(\frac{X + Y}{\mu}\right)\right) XY \indicator{X + Y > 0}}, 
\end{align*}
where conditioned on each support set $\mc J$, we let $X \doteq q_n\paren{\mb w} v_n \sim \mc N\paren{0, q_n^2\paren{\mb w}}$ and $Y \doteq \mb w_{\mc J}^* \overline{\mb v} \sim \mc N\paren{0, \norm{\mb w_{\mc J}}{}^2}$. An upper bound for the above is obtained by calling the estimates in Lemma~\ref{lem:neg_curvature_tanh_square_1} and Lemma~\ref{lem:neg_curvature_tanh_square_3}: 
\begin{align*}
(\mc A)_{x_n \ne 0} 
& \le \frac{2}{\mu} \bb E_{\mc J}\brac{\frac{1}{\sqrt{2\pi}} \frac{\mu \norm{\mb w_{\mc J}}{}^2 q_n^2\paren{\mb w}}{\norm{\mb q_{\mc I}}{}^3} + \frac{3}{4\sqrt{2\pi}} \frac{\norm{\mb w_{\mc J}}{}^2 \mu^3}{\norm{\mb q_{\mc I}}{}^5} \paren{3\mu^2 + 4 \norm{\mb w_{\mc J}}{}^2}} \\
& \qquad + \frac{2}{\mu} \frac{\norm{\mb w}{}^4}{q_n^4\paren{\mb w}} \bb E_{\mc J} \brac{\frac{1}{\sqrt{2\pi}} \frac{\mu \norm{\mb w_{\mc J}}{}^2 q_n^2\paren{\mb w}}{\norm{\mb q_{\mc I}}{}^3} + \frac{3}{4\sqrt{2\pi}} \frac{q_n^2\paren{\mb w} \mu^3}{\norm{\mb q_{\mc I}}{}^5} \paren{3\mu^2 + 4 q_n^2\paren{\mb w}} } \\
& \qquad + \frac{4}{\mu} \frac{\norm{\mb w}{}^2}{q_n^2\paren{\mb w}} \bb E_{\mc J} \brac{\frac{1}{\sqrt{2\pi}} \frac{\mu \norm{\mb w_{\mc J}}{}^2 q_n^2\paren{\mb w}}{\norm{\mb q_{\mc I}}{}^3} + \frac{3}{\sqrt{2\pi}} \frac{\norm{\mb w_{\mc J}}{}^2 q_n^2(\mb w)\mu^3}{\norm{\mb q_{\mc I}}{}^5}} \\
& \le \sqrt{\frac{2}{\pi}} \bb E_{\mc J} \brac{\frac{\norm{\mb w_{\mc J}}{}^2 q_n^4(\mb w) + \norm{\mb w_{\mc J}}{}^2 \norm{\mb w}{}^4 + 2\norm{\mb w_{\mc J}}{}^2 \norm{\mb w}{}^2 q_n^2(\mb w)}{q_n^2(\mb w) \norm{\mb q_{\mc I}}{}^3}} \\
& \qquad + \frac{12\mu^2}{\sqrt{2\pi} q_n^5(\mb w)} +  \frac{9\mu^4}{2\sqrt{2\pi} q_n^5(\mb w)} \\
& \le \frac{1}{q_n^2(\mb w)}\sqrt{\frac{2}{\pi}} \bb E_{\mc J} \brac{ \frac{\norm{\mb w_{\mc J}}{}^2}{\norm{\mb q_{\mc I}}{}^3}} + \frac{12\mu^2}{\sqrt{2\pi} q_n^5(\mb w)} +  \frac{9\mu^4}{2\sqrt{2\pi} q_n^5(\mb w)}, 
\end{align*}
where we have used $\mu < q_n\paren{\mb w} \le \norm{\mb q_{\mc I}}{}$ and $\norm{\mb w_{\mc J}}{} \le \norm{\mb q_{\mc I}}{}$ and $\norm{\mb w}{} \le 1$ and $\theta \in \paren{0, 1/2}$ to simplify the intermediate quantities to obtain the last line. 

Thus, we obtain that 
\begin{align}
(\mc A) 
& \le \paren{1-\theta} \cdot 8\exp(-2)\mu + \frac{\theta}{q_n^2(\mb w)}\sqrt{\frac{2}{\pi}} \bb E_{\mc J} \brac{ \frac{\norm{\mb w_{\mc J}}{}^2}{\norm{\mb q_{\mc I}}{}^3}} + \frac{12\theta\mu^2}{\sqrt{2\pi} q_n^5(\mb w)} +  \frac{9\theta \mu^4}{2\sqrt{2\pi} q_n^5(\mb w)} \nonumber \\
& \le \frac{\theta}{q_n^2(\mb w)}\sqrt{\frac{2}{\pi}} \bb E_{\mc J} \brac{ \frac{\norm{\mb w_{\mc J}}{}^2}{\norm{\mb q_{\mc I}}{}^3}} + 2\mu + \frac{12\theta \mu^2}{\sqrt{2\pi}}\paren{\frac{1}{q_n^3(\mb w)} + \frac{1}{q_n^5(\mb w)}}. 
\end{align}

\item \textbf{A lower bound for $(\mc B)$.} Similarly, we obtain 
\begin{align*}
& \expect{\tanh\left(\frac{\mb q^*\paren{\mb w} \mb x}{\mu}\right) \frac{x_n}{q_n^3\paren{\mb w}}\norm{\mb w}{}^2} \nonumber \\
=\; & \frac{\norm{\mb w}{}^2\theta}{q_n^4\paren{\mb w}} \bb E_{\mc J} \bb E_{\mb v}\brac{\tanh\paren{\frac{\mb w^*_{\mc J} \overline{\mb v} + q_n\paren{\mb w} v_n}{\mu}} v_n q_n\paren{\mb w}} \nonumber \\
\ge\; & \frac{\norm{\mb w}{}^2\theta}{q_n^4\paren{\mb w}} \bb E_{\mc J} \brac{\sqrt{\frac{2}{\pi}}\frac{q_n^2\paren{\mb w}}{\norm{\mb q_{\mc I}}{}} - \sqrt{\frac{2}{\pi}}\frac{q_n^2\paren{\mb w}\mu^2}{\norm{\mb q_{\mc I}}{}^3} - \frac{3q_n^2\paren{\mb w}\mu^4}{2\sqrt{2\pi}\norm{\mb q_{\mc I}}{}^5}} \quad (\text{Lemma~\ref{lem:neg_curvature_tanh_square_2}})\\
\ge\; & \sqrt{\frac{2}{\pi}} \frac{\theta}{q_n^2\paren{\mb w}} \bb E_{\mc J}\brac{\frac{\norm{\mb w}{}^2}{\norm{\mb q_{\mc I}}{}}} - \frac{4\theta \mu^2}{\sqrt{2\pi}q_n^5(\mb w)}. 
\end{align*}
\end{enumerate}

Collecting the above estimates, we obtain 
\begin{align}
& \mb w^* \nabla^2_{\mb w}\expect{h_{\mu}\paren{\mb q^*\paren{\mb w} \mb x}} \mb w \nonumber \\
\le \; & \sqrt{\frac{2}{\pi}} \frac{\theta}{q_n^2\paren{\mb w}} \bb E_{\mc J}\brac{\frac{\norm{\mb w_{\mc J}}{}^2}{\norm{\mb q_{\mc I}}{}^3} - \frac{\norm{\mb w}{}^2\paren{\norm{\mb w_{\mc J}}{}^2 + q_n^2\paren{\mb w}}}{\norm{\mb q_{\mc I}}{}^3}}+ 2\mu + \frac{4\theta \mu^2}{\sqrt{2\pi}}\paren{\frac{3}{q_n^3(\mb w)} + \frac{4}{q_n^5(\mb w)}} \nonumber \\
=\; &  -\sqrt{\frac{2}{\pi}}\theta \expect{\frac{\norm{\mb w_{\mc J^c}}{}^2}{\norm{\mb q_{\mc I}}{}^3}} + 2\mu + \frac{4\theta \mu^2}{\sqrt{2\pi}}\paren{\frac{3}{q_n^3(\mb w)} + \frac{4}{q_n^5(\mb w)}} \nonumber \\
\le \; & -\sqrt{\frac{2}{\pi}}\theta \paren{1-\theta} \norm{\mb w}{}^2 \expect{\frac{1}{\norm{\mb q_{\mc I}}{}^3}} + 2\mu + \frac{4\theta \mu^2}{\sqrt{2\pi}}\paren{\frac{3}{q_n^3(\mb w)} + \frac{4}{q_n^5(\mb w)}}, \label{eq:neg_curvature_final_form}
\end{align}
where to obtain the last line we have invoked the association inequality in Lemma~\ref{lemma:harris_ineq}, as both $\norm{\mb w_{\mc J^c}}{}^2$ and $1/\norm{\mb q_{\mc I}}{}^3$ both coordinatewise nonincreasing w.r.t. the index set. Substituting the upper bound for $\mu$ into~\eqref{eq:neg_curvature_final_form} and noting $q_n\paren{\mb w} \ge 1/(2\sqrt{n})$ (implied by the assumption $\norm{\mb w}{} \le \sqrt{(4n-1)/(4n)}$), we obtain the claimed result. 
}
\end{proof}

\subsubsection{Proof of Proposition~\ref{prop:geometry_asymp_gradient}} \label{sec:proof_geo_asym_gradient}
\begin{proof}
By similar consideration as proof of the above proposition, the following is justified:  
\begin{align*}
\nabla_{\mb w}\expect{h_{\mu}\paren{\mb q^*\paren{\mb w} \mb x}} = \expect{\nabla_{\mb w} h_{\mu}\paren{\mb q^*\paren{\mb w} \mb x}}. 
\end{align*}
Now consider 
\begin{align}
\mb w^* \nabla \expect{h_\mu (\mb q^*\paren{\mb w} \mb x)} 
& = \nabla \expect{\mb w^* h_\mu (\mb q^*\paren{\mb w} \mb x)} \nonumber \\
& = \underbrace{\expect{\tanh\paren{\frac{\mb q^*\paren{\mb w} \mb x}{\mu}} \paren{\mb w^* \bar{\mb x} }}}_{(\mc A)} - \frac{\norm{\mb w}{}^2}{q_n}\underbrace{\expect{\tanh\paren{\frac{\mb q^*\paren{\mb w} \mb x}{\mu}} x_n}}_{(\mc B)}. \label{eqn:proof_grad_1}
\end{align}

\begin{enumerate}
\item \textbf{A lower bound for $(\mc A)$.} We have 
\begin{align*}
& \expect{\tanh\paren{\frac{\mb q^*\paren{\mb w} \mb x}{\mu}} \paren{\mb w^* \overline{\mb x} }} \\
=\; & \theta \bb E_{\mc J} \brac{ \bb E_{\mb v} \brac{ \tanh\paren{\frac{\mb w_{\mc J}^* \overline{\mb v} + q_n\paren{\mb w} v_n}{\mu} } \paren{\mb w_{\mc J}^* \overline{\mb v} }}}  + (1-\theta) \bb E_{\mc J} \brac{ \bb E_{\mb v} \brac{ \tanh\paren{\frac{\mb w_{\mc J}^* \overline{\mb v}}{\mu}} \paren{\mb w_{\mc J}^* \overline{\mb v} }}} \\
=\; &\theta \bb E_{\mc J}\brac{\bb E_{X, Y}\brac{ \tanh\paren{\frac{X+Y}{\mu}}Y } }\;+\; (1-\theta) \bb E_{\mc J}\brac{\bb E_{Y}\brac{\tanh\paren{\frac{Y}{\mu}}Y} }, 
\end{align*}
where $X \doteq q_n\paren{\mb w} v_n \sim \mc N\paren{0, q_n^2\paren{\mb w}}$ and $Y \doteq \mb w^*_{\mc J} \overline{\mb v} \sim \mc N\paren{0, \norm{\mb w_{\mc J}}{}^2}$. Now by Lemma \ref{lemma:harris_ineq} we obtain 
\begin{align*}
\bb E \brac{ \tanh\paren{\frac{X+Y}{\mu} } Y} \ge \bb E\brac{\tanh\paren{\frac{X+Y}{\mu}}} \bb E\brac{Y} = 0, 
\end{align*} 
as $\tanh\paren{\frac{X+Y}{\mu} } $ and $X$ are both coordinatewise nondecreasing function of $X$ and $Y$. Using $\tanh\paren{z} \geq \paren{1-\exp\paren{-2z}}/2$ for all $z \ge 0$ and integral results in Lemma~\ref{lem:aux_asymp_proof_a}, we obtain
\begin{align*}
\bb E\brac{\tanh\paren{\frac{Y}{\mu}}Y} 
& = 2\bb E\brac{\tanh\paren{\frac{Y}{\mu}}Y \indicator{Y > 0}} \\
& \ge \bb E\brac{\paren{1-\exp\paren{-\frac{2Y}{\mu}} }Y\indicator{Y>0} } \\
& = \frac{2\sigma_Y^2}{\mu}\exp\paren{\frac{2\sigma_Y^2}{\mu^2}} \Phi^c\paren{\frac{2\sigma_Y}{\mu}}  \\
& \ge \frac{2\sigma_Y^2}{\mu \sqrt{2\pi}} \paren{\sqrt{1+\frac{\sigma_Y^2}{\mu^2}} - \frac{\sigma_Y}{\mu}} \quad (\text{Type III lower bound for $\Phi^c(\cdot)$, Lemma~\ref{lem:gaussian_tail_est}})\\
& \ge \frac{2\sigma_Y^2}{\mu\sqrt{2\pi}} \paren{ \sqrt{1+\frac{\norm{\mb w}{}^2}{\mu^2}} - \frac{\norm{\mb w}{}}{\mu}}. \quad (\text{$t \mapsto \sqrt{1+t^2}-t$ decreasing over $t > 0$})
\end{align*}
Collecting the above estimates, we have 
\begin{align}
\expect{\tanh\paren{\frac{\mb q^*\paren{\mb w} \mb x}{\mu}} \paren{\mb w^* \overline{\mb x} }} 
& \ge \paren{1- \theta} \bb E_{\mc J}\brac{\frac{2\norm{\mb w_{\mc J}}{}^2}{\mu\sqrt{2\pi}} \paren{ \sqrt{1+\frac{\norm{\mb w}{2}^2}{\mu^2}} - \frac{\norm{\mb w}{}}{\mu}}} \nonumber \\
& \ge  \paren{1- \theta} \bb E_{\mc J}\brac{\frac{2\norm{\mb w_{\mc J}}{}^2}{\mu\sqrt{2\pi}} \frac{\mu}{10\norm{\mb w}{}}} \nonumber \\
& \ge \frac{\theta\paren{1-\theta}\norm{\mb w}{}}{5\sqrt{2\pi}}, \label{eqn:proof_grad_2}
\end{align}
where at the second line we have used the assumption that $\norm{\mb w}{} \ge \mu/(4\sqrt{2})$ and also the fact that $\sqrt{1+x^2} \ge x + \frac{1}{10x}$ for $x \ge 1/(4\sqrt{2})$. 

\item \textbf{An upper bound for $(\mc B)$. } We have
\begin{align}
\bb E\brac{\tanh\paren{\frac{\mb q^*\paren{\mb w} \mb x}{\mu}} x_n} 
\le \theta \bb E\brac{\abs{\tanh\paren{\frac{\mb q^*\paren{\mb w} \mb x}{\mu}}} \abs{v_n}} 
\le \theta \sqrt{\frac{2}{\pi}}, 
\label{eqn:proof_grad_3}
\end{align}
as $\tanh\paren{\cdot}$ is bounded by one in magnitude
\end{enumerate}

 Plugging the results of~\eqref{eqn:proof_grad_2} and~\eqref{eqn:proof_grad_3} into~\eqref{eqn:proof_grad_1} and noticing that $q_n\paren{\mb w}^2 + \norm{\mb w}{}^2 = 1$ we obtain 
\begin{align*}
\mb w^* \nabla \expect{h_\mu (\mb q^*\paren{\mb w} \mb x)} \;\ge\; 
\frac{\theta\norm{\mb w}{}}{\sqrt{2\pi}}\brac{ \frac{1-\theta}{5} -\frac{2\norm{\mb w}{}}{\sqrt{1-\norm{\mb w}{}^2}}} \ge \frac{\theta\paren{1-\theta}\norm{\mb w}{}}{10\sqrt{2\pi}}, 
\end{align*}
where we have used $\frac{2\norm{\mb w}{}}{\sqrt{1-\norm{\mb w}{}^2}} \le \frac{1}{10}\paren{1-\theta}$ when $\norm{\mb w}{} \le 1/(20\sqrt{5})$ and $\theta \le 1/2$, completing the proof. 
\end{proof}

\subsubsection{Proof of Proposition~\ref{prop:geometry_asymp_strong_convexity}} \label{sec:proof_geo_asym_strcvx}

\begin{proof}
By consideration similar to proof of Proposition~\ref{prop:geometry_asymp_curvature}, we can exchange the Hessian and expectation, i.e., 
\begin{align*}
\nabla^2_{\mb w} \expect{h_{\mu}\paren{\mb q^*\paren{\mb w} \mb x}} = \expect{\nabla^2_{\mb w} h_{\mu}\paren{\mb q^*\paren{\mb w} \mb x}}. 
\end{align*}
We are interested in the expected Hessian matrix
\begin{align*}
\nabla^2_{\mb w} \bb E\brac{h_{\mu}\paren{\mb q^*\paren{\mb w}\mb x}} 
&=\frac{1}{\mu}\bb E\brac{\paren{1-\tanh^2\paren{\frac{\mb q^*\paren{\mb w} x}{\mu}} } \paren{\overline{\mb x} - \frac{x_n}{q_n\paren{\mb w}}\mb w }\paren{\overline{\mb x} - \frac{x_n}{q_n\paren{\mb w}}\mb w }^*  } \nonumber \\
&- \bb E\brac{\tanh\paren{\frac{\mb q^*\paren{\mb w}\mb x}{\mu}} \paren{\frac{x_n}{q_n\paren{\mb w}} \mb I +\frac{x_n}{q_n^3\paren{\mb w}} \mb w\mb w^*}}
\end{align*}
in the region that $0 \le \norm{\mb w}{}\le \mu/(4\sqrt{2})$. 

When $\mb w=\mb 0$, by Lemma \ref{lem:aux_asymp_proof_a}, we have 
\begin{align*}
& \left. \bb E\brac{\nabla^2_{\mb w} h_{\mu}\paren{\mb q^*\paren{\mb w}\mb x} }\right|_{\mb w=0}\\ 
=\; & \frac{1}{\mu}\bb E\brac{\paren{1-\tanh^2\paren{\frac{x_n}{\mu}}}\overline{\mb x}\; \overline{\mb x}^* } - \bb E\brac{\tanh\paren{\frac{x_n}{\mu}}x_n}\mb I \\
=\; & \frac{\theta(1-\theta)}{\mu}\mb I + \frac{\theta^2}{\mu} \bb E_{v_n}\brac{1- \tanh^2\paren{\frac{v_n}{\mu}}}\mb I - \frac{\theta}{\mu} \bb E_{v_n}\brac{1-\tanh^2\paren{\frac{v_n}{\mu}}}\mb I  \\
=\; & \frac{\theta(1-\theta)}{\mu} \bb E_{v_n}\brac{\tanh^2\paren{\frac{q_n\paren{\mb w}v_n}{\mu}}} \mb I. 
\end{align*}
Simple calculation based on Lemma~\ref{lem:aux_asymp_proof_a} shows 
\begin{align*}
\bb E_{v_n}\brac{\tanh^2\paren{\frac{v_n}{\mu}}} \ge 2\paren{1-4\exp\paren{\frac{2}{\mu^2}} \Phi^c\paren{\frac{2}{\mu}}} \ge 2\paren{1-\frac{2}{\sqrt{2\pi}} \mu}. 
\end{align*}
Invoking the assumptions $\mu \le 1/(20 \sqrt{n}) \le 1/20$ and $\theta < 1/2$, we obtain 
\begin{align*}
\left. \bb E\brac{\nabla^2_{\mb w} h_{\mu}\paren{\mb q^*\paren{\mb w}\mb x} }\right|_{\mb w=0} \succeq \frac{\theta\paren{1-\theta}}{\mu} \paren{2 - \frac{4}{\sqrt{2\pi}}\mu} \mb I 
\succeq \frac{\theta}{\mu} \paren{1-\frac{1}{10\sqrt{2\pi}}} \mb I.  
\end{align*}

When $0 <\norm{\mb w}{}\le \mu/(4\sqrt{2})$, we aim to derive a semidefinite lower bound for 
\begin{align}
& \bb E\brac{\nabla^2_{\mb w} h_{\mu}\paren{\mb q^*\paren{\mb w}\mb x}} \nonumber \\
=\; &  \underbrace{\frac{1}{\mu}\bb E\brac{\paren{1-\tanh^2\paren{\frac{\mb q^*\paren{\mb w}\mb x}{\mu}}}\overline{\mb x}\; \overline{\mb x}^*}}_{(\mc A)} - \underbrace{\frac{1}{q_n^2\paren{\mb w}}\bb E\brac{\tanh\paren{\frac{\mb q^*\paren{\mb w}\mb x}{\mu}}q_n\paren{\mb w}x_n}\mb I}_{(\mc B)}\nonumber \\
&-\underbrace{\frac{1}{\mu q_n^2\paren{\mb w}}\bb E\brac{\paren{1-\tanh^2\paren{\frac{\mb q^*\paren{\mb w}\mb x}{\mu}}}q_n\paren{\mb w}x_n\paren{\mb w \overline{\mb x}^*+\overline{\mb x}\mb w^*}}}_{(\mc C)} \nonumber \\
& + \underbrace{\frac{1}{q_n^4\paren{\mb w}}\Brac{ \frac{1}{\mu}\bb E\brac{\paren{1-\tanh^2\paren{\frac{\mb q^*\paren{\mb w}\mb x}{\mu}}}(q_n\paren{\mb w}x_n)^2 } -\bb E\brac{\tanh\paren{\frac{\mb q^*\paren{\mb w}\mb x}{\mu}}q_n\paren{\mb w}x_n}  }\mb w\mb w^*}_{(\mc D)}  \label{eqn:thm-str-cvx-1}. 
\end{align}
We will first provide bounds for $(\mc C)$ and $(\mc D)$, which are relatively simple. Then we will bound $(\mc A)$ and $(\mc B)$, which are slightly more tricky.  

\begin{enumerate}
\item \textbf{An upper bound for $(\mc C)$.} We have 
\begin{align*}
(\mc C) 
& \le \frac{1}{\mu q_n^2\paren{\mb w}}\norm{ \bb E\brac{\paren{1-\tanh^2\paren{\frac{\mb q^*\paren{\mb w}\mb x}{\mu}}}q_n\paren{\mb w}x_n\paren{\mb w \overline{\mb x}^*+\overline{\mb x}\mb w^*}} }{} \nonumber \\
& \le  \frac{2}{\mu q_n^2\paren{\mb w}}\norm{ \bb E\brac{\paren{1-\tanh^2\paren{\frac{\mb q^*\paren{\mb w}\mb x}{\mu}}}q_n\paren{\mb w}x_n \bar{\mb x} } \mb w^*}{} \nonumber \\
& \le  \frac{2}{\mu q_n^2\paren{\mb w}}\norm{ \bb E\brac{\paren{1-\tanh^2\paren{\frac{\mb q^*\paren{\mb w}\mb x}{\mu}}}q_n\paren{\mb w}x_n \overline{\mb x} } }{} \norm{\mb w}{} \nonumber \\
& \le \frac{2}{\mu q_n^2\paren{\mb w}} \bb E \norm{ \paren{1-\tanh^2\paren{\frac{\mb q^*\paren{\mb w}\mb x}{\mu}}}q_n\paren{\mb w}x_n \overline{\mb x} }{} \norm{\mb w}{} \quad \text{(Jensen's inequality)}\\
& \le  \frac{2}{\mu q_n\paren{\mb w}}\theta^2 \bb E\brac{\abs{v_n}} \expect{\norm{\overline{\mb v}}{}} \norm{\mb w}{}  \quad \text{($1-\tanh^2(\cdot) \le 1$, $x_n$ and $\ol{\mb x}$ independent)}\nonumber \\
& \le  \frac{4\theta^2}{\pi\mu q_n\paren{\mb w}} \sqrt{n}\norm{\mb w}{} \le \frac{\theta}{\mu} \frac{4\theta \sqrt{n}\norm{\mb w}{}}{\pi \sqrt{1-\norm{\mb w}{}^2}} \le \frac{\theta}{\mu}\frac{1}{40\pi}, 
\end{align*}
where to obtain the final bound we have invoked the assumptions: $\norm{\mb w}{} \le \mu/(4\sqrt{2})$, $\mu \le 1/(20\sqrt{n})$, and $\theta \le 1/2$.

\item \textbf{A lower bound for $(\mc D)$.} \js{We directly drop the first expectation which is positive, and derive an upper for the second expectation term as:
\begin{align*}
\bb E\brac{ \tanh\paren{\frac{\mb q^*\paren{\mb w}\mb x}{\mu}}q_nx_n} \le q_n(\mb w) \theta \expect{\abs{v_n}} = \sqrt{\frac{2}{\pi}} \theta q_n(\mb w). 
\end{align*}
Thus, 
\begin{align*}
(\mc D) \succeq - \frac{1}{q_n^4(\mb w)} \sqrt{\frac{2}{\pi}} \theta q_n(\mb w) \mb w \mb w^* \succeq - \frac{\theta}{q_n^3(\mb w)} \sqrt{\frac{2}{\pi}} \norm{\mb w}{}^2 \mb I \succeq -\frac{\theta}{32000\mu} \sqrt{\frac{2}{\pi}} \mb I,  
\end{align*}
}where we have again used $\norm{\mb w}{} \le \mu/(4\sqrt{2})$, $\mu \le 1/(20\sqrt{n})$, and $q_n\paren{\mb w} \ge 1/(2\sqrt{n})$ to obtain the final bound.

\item \textbf{An upper bound for $(\mc B)$.} \js{Similar to the way we bound $(\mc D)$, 
\begin{align*}
\frac{1}{q_n^2(\mb w)}\bb E\brac{ \tanh\paren{\frac{\mb q^*\paren{\mb w}\mb x}{\mu}}q_nx_n} \le \frac{1}{q_n(\mb w)}\sqrt{\frac{2}{\pi}} \theta \le \frac{\theta}{10\mu}\sqrt{\frac{2}{\pi}}. 
\end{align*}
} 

\item \textbf{A lower bound for $(\mc A)$.} First note that 
\begin{align*}
(\mc A) \succeq \frac{1-\theta}{\mu} \bb E_{\overline{\mb x}}\brac{\paren{1-\tanh^2\paren{\frac{\mb w^* \overline{\mb x}}{\mu}}}\overline{\mb x}\; \overline{\mb x}^*}. 
\end{align*}
Thus, we set out to lower bound the expectation as 
\begin{align*}
\bb E_{\overline{\mb x}}\brac{\paren{1-\tanh^2\paren{\frac{\mb w^* \overline{\mb x}}{\mu}}}\overline{\mb x}\; \overline{\mb x}^*} \succeq \theta \beta \mb I
\end{align*}
for some scalar $\beta \in (0, 1)$, as $\bb E_{\ol{\mb x}} \left[\ol{\mb x} \ol{\mb x}^*\right] = \theta \mb I$. Suppose $\mb w$ has $k \in [n-1]$ nonzeros, w.l.o.g., further assume the first $k$ elements of $\mb w$ are these nonzeros. It is easy to see the expectation above has a block diagonal structure $\diag\paren{\bm \Sigma; \alpha \theta \mb I_{n-1-k}}$, where 
\begin{align*}
\alpha \doteq \bb E_{\overline{\mb x}}\brac{\paren{1-\tanh^2\paren{\frac{\mb w^* \overline{\mb x}}{\mu}}}}. 
\end{align*}
So in order to derive the $\theta \beta \mb I$ lower bound as desired, it is sufficient to show $\bm \Sigma \succeq \theta \beta \mb I$ and $\beta \le \alpha$, i.e., letting $\widetilde{\mb w} \in \R^k$ be the subvector of nonzero elements, 
\begin{align*}
\bb E_{\widetilde{\mb x} \sim_{i.i.d.} \mathrm{BG}\paren{\theta}}\brac{\paren{1-\tanh^2\paren{\frac{\widetilde{\mb w}^* \widetilde{\mb x}}{\mu}}}\widetilde{\mb x}\; \widetilde{\mb x}^*} \succeq \theta \beta \mb I, 
\end{align*} 
which is equivalent to that for all $\mb z \in \R^k$ such that $\norm{\mb z}{} = 1$, 
\begin{align*}
\bb E_{\widetilde{\mb x} \sim_{i.i.d.} \mathrm{BG}\paren{\theta}}\brac{\paren{1-\tanh^2\paren{\frac{\widetilde{\mb w}^* \widetilde{\mb x}}{\mu}}} \paren{\widetilde{\mb x}^* \mb z}^2} \ge \theta \beta. 
\end{align*}
It is then sufficient to show that for any nontrivial support set $\mc S \subset [k]$ and any vector $\mb z \in \R^k$ such that $\supp\paren{\mb z} = \mc S$ with $\norm{\mb z}{} = 1$, 
\begin{align*}
\bb E_{\widetilde{\mb v} \sim_{i.i.d.} \mc N\paren{0, 1}}\brac{\paren{1-\tanh^2\paren{\frac{\widetilde{\mb w}^*_{\mc S} \widetilde{\mb v}}{\mu}}} \paren{\widetilde{\mb v}^* \mb z}^2} \ge \beta. 
\end{align*}
To see the above implication, suppose the latter claimed holds, then for any $\mb z$ with unit norm, 
\begin{align*}
& \bb E_{\widetilde{\mb x} \sim_{i.i.d.} \mathrm{BG}\paren{\theta}}\brac{\paren{1-\tanh^2\paren{\frac{\widetilde{\mb w}^* \widetilde{\mb x}}{\mu}}} \paren{\widetilde{\mb x}^* \mb z}^2}  \\
=\; & \sum_{s = 1}^k \theta^s \paren{1-\theta}^{k-s} \sum_{\mc S \in \binom{[k]}{s}} \bb E_{\widetilde{\mb v} \sim_{i.i.d.} \mc N\paren{0, 1}}\brac{\paren{1-\tanh^2\paren{\frac{\widetilde{\mb w}^*_{\mc S} \widetilde{\mb v}}{\mu}}} \paren{\widetilde{\mb v}^* \mb z_{\mc S}}^2} \\
\ge\; & \sum_{s = 1}^k \theta^s \paren{1-\theta}^{k-s} \sum_{\mc S \in \binom{[k]}{s}} \beta \norm{\mb z_{\mc S}}{}^2 = \beta \bb E_{\mc S}\brac{\norm{\mb z_{\mc S}}{}^2} = \theta \beta. 
\end{align*}
Now for any fixed support set $\mc S \subset [k]$, $\mb z = \mc P_{\widetilde{\mb w}_{\mc S}} \mb z + \paren{\mb I - \mc P_{\widetilde{\mb w}_{\mc S}}} \mb z$. So we have
\begin{align*}
& \bb E_{\widetilde{\mb v} \sim_{i.i.d.} \mc N\paren{0, 1}}\brac{\paren{1-\tanh^2\paren{\frac{\widetilde{\mb w}^*_{\mc S} \widetilde{\mb v}}{\mu}}} \paren{\widetilde{\mb v}^* \mb z}^2} \\
=\; & \bb E_{\widetilde{\mb v}}\brac{\paren{1-\tanh^2\paren{\frac{\widetilde{\mb w}^*_{\mc S} \widetilde{\mb v}}{\mu}}} \paren{\widetilde{\mb v}^* \mc P_{\widetilde{\mb w}_{\mc S}} \mb z}^2} \nonumber \\
& \quad + \bb E_{\widetilde{\mb v}}\brac{\paren{1-\tanh^2\paren{\frac{\widetilde{\mb w}^*_{\mc S} \widetilde{\mb v}}{\mu}}} \paren{\widetilde{\mb v}^* \paren{\mb I - \mc P_{\widetilde{\mb w}_{\mc S}}} \mb z}^2} \quad (\text{cross-term vanishes due to independence})\\
=\; & \frac{\paren{{\widetilde{\mb w}_{\mc S}}^* \mb z}^2}{\norm{\mb w_{\mc S}}{}^4} \bb E_{\widetilde{\mb v}}\brac{\paren{1-\tanh^2\paren{\frac{\widetilde{\mb w}^*_{\mc S} \widetilde{\mb v}}{\mu}}} \paren{\widetilde{\mb v}^* \widetilde{\mb w}_{\mc S}}^2} \nonumber \\
& \quad + \bb E_{\widetilde{\mb v}}\brac{\paren{1-\tanh^2\paren{\frac{\widetilde{\mb w}^*_{\mc S} \widetilde{\mb v}}{\mu}}}} \bb E_{\widetilde{\mb v}}\brac{\paren{\widetilde{\mb v}^* \paren{\mb I - \mc P_{\widetilde{\mb w}_{\mc S}}} \mb z}^2} \quad (\text{exp. factorizes due to independence})\\
\ge\; & 2\frac{\paren{{\widetilde{\mb w}_{\mc S}}^* \mb z}^2}{\norm{\mb w_{\mc S}}{}^4} \bb E_{\widetilde{\mb v}}\brac{\exp\paren{-\frac{2\widetilde{\mb w}^*_{\mc S} \widetilde{\mb v}}{\mu}} \paren{\widetilde{\mb v}^* \widetilde{\mb w}_{\mc S}}^2 \indicator{\widetilde{\mb v}^* \widetilde{\mb w}_{\mc S} > 0}} 
+ 2\bb E_{\widetilde{\mb v}}\brac{\exp\paren{-\frac{2\widetilde{\mb w}^*_{\mc S} \widetilde{\mb v}}{\mu}} \indicator{\widetilde{\mb w}^*_{\mc S} \widetilde{\mb v} > 0}} \norm{\paren{\mb I - \mc P_{\widetilde{\mb w}_{\mc S}}} \mb z}{}^2. 
\end{align*}
Using expectation result from Lemma~\ref{lem:aux_asymp_proof_a}, the above lower bound is further bounded as: 
\begin{align*}
& \bb E_{\widetilde{\mb v} \sim_{i.i.d.} \mc N\paren{0, 1}}\brac{\paren{1-\tanh^2\paren{\frac{\widetilde{\mb w}^*_{\mc S} \widetilde{\mb v}}{\mu}}} \paren{\widetilde{\mb v}^* \mb z}^2} \nonumber \\
\ge\; & 2\frac{\paren{{\widetilde{\mb w}_{\mc S}}^* \mb z}^2}{\norm{\mb w_{\mc S}}{}^4} \brac{\paren{\norm{\widetilde{\mb w}_{\mc S}}{}^2 + \frac{4}{\mu^2} \norm{\widetilde{\mb w}_{\mc S}}{}^4} \exp\paren{\frac{2\norm{\widetilde{\mb w}_{\mc S}}{}^2}{\mu^2}} \Phi^c\paren{\frac{2\norm{\widetilde{\mb w}_{\mc S}}{}}{\mu}} - \frac{2\norm{\widetilde{\mb w}_{\mc S}}{}^3}{\mu\sqrt{2\pi}}} \\ 
& \quad + 2 \exp\paren{\frac{2\norm{\widetilde{\mb w}_{\mc S}}{}^2}{\mu^2}} \Phi^c\paren{\frac{2\norm{\widetilde{\mb w}_{\mc S}}{}}{\mu}}\norm{\paren{\mb I - \mc P_{\widetilde{\mb w}_{\mc S}}} \mb z}{}^2 \\
\ge\; & \underbrace{\paren{2\frac{\paren{{\widetilde{\mb w}_{\mc S}}^* \mb z}^2}{\norm{\mb w_{\mc S}}{}^2} + 2 \norm{\paren{\mb I - \mc P_{\widetilde{\mb w}_{\mc S}}} \mb z}{}^2}}_{ =2\norm{\mb z}{}^2 = 2}    \exp\paren{\frac{2\norm{\widetilde{\mb w}_{\mc S}}{}^2}{\mu^2}} \Phi^c\paren{\frac{2\norm{\widetilde{\mb w}_{\mc S}}{}}{\mu}} - \frac{4\paren{{\widetilde{\mb w}_{\mc S}}^* \mb z}^2}{\mu\sqrt{2\pi}\norm{\widetilde{\mb w}_{\mc S}}{}} \\
\ge\; & \frac{1}{\sqrt{2\pi}}\paren{\sqrt{4+\frac{4\norm{\widetilde{\mb w}_{\mc S}}{}^2}{\mu^2}}- \frac{2\norm{\widetilde{\mb w}_{\mc S}}{}}{\mu}} - \frac{4\paren{{\widetilde{\mb w}_{\mc S}}^* \mb z}^2}{\mu\sqrt{2\pi}\norm{\widetilde{\mb w}_{\mc S}}{}} \quad (\text{Type III lower bound for $\Phi^c(\cdot)$})\nonumber \\
\ge\; & \frac{1}{\sqrt{2\pi}} \paren{\sqrt{4+\frac{4\norm{\mb w}{}^2}{\mu^2}}- \frac{2\norm{\mb w}{}}{\mu}} - \frac{4\norm{\mb w}{}}{\mu\sqrt{2\pi}} \quad (\text{$t \mapsto \sqrt{4 + t^2} - t$ nonincreasing and Cauchy-Schwarz}) \\
\ge\; & \frac{1}{\sqrt{2\pi}} \paren{2 - \frac{3}{4} \sqrt{2}}, 
\end{align*}
where to obtain the last line we have used $\norm{\mb w}{} \le \mu/(4\sqrt{2})$. On the other hand, we similarly obtain 
\begin{align*}
\alpha = \bb E_{\mc J} \bb E_{Z \sim \mc N(0, \|\mb w_{\mc J}\|^2)} [1-\tanh^2 (Z/\mu)] \ge \frac{2}{\sqrt{2\pi}} \frac{\sqrt{4 \|\mb w\|^2/\mu^2 + 4} - 2\|\mb w\|/\mu}{2} \ge \frac{1}{\sqrt{2\pi}} \paren{2 - \frac{1}{2} \sqrt{2}}. 
\end{align*}
So we can take $\beta = \frac{1}{\sqrt{2\pi}} \paren{2 - \frac{3}{4} \sqrt{2}} < 1$.
\end{enumerate}

Putting together the above estimates for the case $\mb w \neq \mb 0$, we obtain 
\begin{align*}
\expect{\nabla^2_{\mb w} h_{\mu}\paren{\mb q^*\paren{\mb w} \mb x}} \succeq \frac{\theta}{\mu\sqrt{2\pi}} \paren{1-\frac{3}{8}\sqrt{2} - \frac{\sqrt{2\pi}}{40\pi} - \frac{1}{16000} - \frac{1}{5}} \mb I \succeq \frac{1}{5\sqrt{2\pi}} \frac{\theta}{\mu} \mb I. 
\end{align*}
Hence for all $\mb w$, we can take the $\frac{1}{5\sqrt{2\pi}} \frac{\theta}{\mu}$ as the lower bound, completing the proof. 
\end{proof}

 \label{sec:proof_geo_exp}
\subsubsection{Proof of Pointwise Concentration Results} \label{proof:cn_point}
We first establish a useful comparison lemma between random i.i.d. Bernoulli random vectors random i.i.d. normal random vectors. 

\begin{lemma}\label{lem:U-moments-bound}
Suppose $\mb z, \mb z' \in \R^n$ are independent and obey $\mb z \sim_{i.i.d.} \mathrm{BG}\paren{\theta}$ and $\mb z' \sim_{i.i.d.} \mc N\paren{0, 1}$. Then, for any fixed vector $\mb v \in \R^n$, it holds that 
\begin{align*}
\expect{\abs{\mb v^* \mb z}^m} & \le \expect{\abs{\mb v^* \mb z'}^m } = \bb E_{Z \sim \mc N\paren{0, \norm{\mb v}{}^2}}\brac{\abs{Z}^m}, \\
\expect{\norm{\mb z}{}^m} & \le \expect{\norm{\mb z'}{}^m}, 
\end{align*}
for all integers $m \ge 1$. 
\end{lemma}
\begin{proof}
See page~\pageref{proof:ng_g_comparison} under Section~\ref{sec:app_aux_results}. 
\end{proof}

Now, we are ready to prove Proposition~\ref{prop:concentration-gradient} to Proposition~\ref{prop:concentration-hessian-zero} as follows.

\begin{proof}{\textbf{(of Proposition~\ref{prop:concentration-gradient})} } \label{proof:pt_cn_gradient}
Let 
\begin{align*}
X_k = \frac{\mb w^*}{\norm{\mb w}{2}}\nabla h_\mu\paren{\mb q(\mb w)^*(\mb x_0)_k}, 
\end{align*}
then $\mb w^*\nabla g(\mb w)/\norm{\mb w}{}= \frac{1}{p} \sum_{k=1}^p X_k$. 
For each $X_k, k \in [p]$, from \eqref{eqn:lse-gradient}, we know that 
\begin{align*}
\abs{X_k} = \abs{\tanh\paren{\frac{\mb q(\mb w)^*(\mb x_0)_k}{\mu } }\paren{\frac{\mb w^*\overline{\mb x_0}_k}{\norm{\mb w}{}} - \frac{\norm{\mb w}{2}x_{0k}\paren{n}}{q_n\paren{\mb w}}}} \le \abs{\frac{\mb w^*\overline{\mb x_0}_k}{\norm{\mb w}{}} - \frac{\norm{\mb w}{2}x_{0k}\paren{n}}{q_n\paren{\mb w}}}, 
\end{align*}
as the magnitude of $\tanh\paren{\cdot}$ is bounded by one. Because 
\begin{align*}
\frac{\mb w^*\overline{\mb x_0}_k}{\norm{\mb w}{2}} - \frac{\norm{\mb w}{}x_{0k}\paren{n}}{q_n\paren{\mb w}} = \paren{\frac{\mb w}{\norm{\mb w}{}}, - \frac{\norm{\mb w}{}}{q_n\paren{\mb w}}}^* (\mb x_0)_k\quad \text{and} \quad (\mb x_0)_k \sim_{i.i.d.} \mathrm{BG}\paren{\theta},
\end{align*} 
invoking Lemma~\ref{lem:U-moments-bound}, we obtain for every integer $m \ge 2$ that
\begin{align*}
\expect{\abs{X_k}^m} \le \bb E_{Z \sim \mc N\paren{0, 1/q_n^2\paren{\mb w}}}\brac{\abs{Z}^m} \le \frac{1}{q_n\paren{\mb w}^m} (m-1)!!\; \le\; \frac{m!}{2} \paren{4n} \paren{2\sqrt{n}}^{m-2}, 
\end{align*}
where the Gaussian moment can be looked up in Lemma~\ref{lem:guassian_moment} and we have used that $(m-1)!! \leq m!/2$ and the assumption that $q_n\paren{\mb w} \ge 1/(2\sqrt{n})$ to get the result. Thus, by taking $\sigma^2 = 4n\geq \bb E\brac{X_k^2}$ and $R = 2\sqrt{n}$, and we obtain the claimed result by invoking Lemma~\ref{lem:mc_bernstein_scalar}.
\end{proof}

\begin{proof}{\textbf{(of Proposition~\ref{prop:concentration-hessian-negative})}} \label{proof:pt_cn_curvature}
Let 
\begin{align*}
Y_k = \frac{1}{\norm{\mb w}{}^2}\mb w^*\nabla^2 h_\mu\paren{\mb q(\mb w)^*(\mb x_0)_k}\mb w, 
\end{align*}
then $\mb w^*\nabla^2 g(\mb w)\mb w/\norm{\mb w}{}^2 = \frac{1}{p} \sum_{k=1}^p Y_k$. For each $Y_k$ ($k \in [p]$), from \eqref{eqn:lse-hessian}, we know that
\begin{align*}
Y_k \;&=\; \underbrace{\frac{1}{\mu}\paren{1-\tanh^2\paren{\frac{\mb q(\mb w)^*(\mb x_0)_k}{\mu}}}
\paren{\frac{\mb w^*\overline{\mb x_0}_k}{\norm{\mb w}{}} - \frac{x_{0k}\paren{n}\norm{\mb w}{}}{q_n(\mb w)}}^2}_{\doteq W_k} \underbrace{- \tanh\paren{\frac{\mb q(\mb w)^*(\mb x_0)_k}{\mu}}\frac{x_{0k}\paren{n}}{q_n^3(\mb w)}}_{\doteq V_k}. 
\end{align*}
Then by similar argument as in proof to Proposition~\ref{prop:concentration-gradient}, we have for all integers $m \ge 2$ that 
\begin{align*}
\expect{\abs{W_k}^m} 
& \le \frac{1}{\mu^m} \bb E\brac{\abs{\frac{\mb w^*\overline{\mb x}_k}{\norm{\mb w}{}}- \frac{x_k\paren{n}\norm{\mb w}{}}{q_n(\mb w)}}^{2m}} \le  \frac{1}{\mu^m} \bb E_{Z \sim \mc N\paren{0, 1/q_n^2\paren{\mb w}}}\brac{\abs{Z}^{2m}} \\
& \le \frac{1}{\mu^m}(2m-1)!!(4n)^m \le  \frac{m!}{2} \paren{\frac{4n}{\mu}}^m,  \\
\bb E\brac{\abs{V_k}^m} 
& \le \frac{1}{q_n^{3m}(\mb w)}\bb E\brac{\abs{v_k\paren{n}}^m} \le \paren{2\sqrt{n}}^{3m} (m-1)!! \le \frac{m!}{2} \paren{8n\sqrt{n}}^m, 
\end{align*}
where we have again used the assumption that $q_n\paren{\mb w} \ge 1/(2\sqrt{n})$ to simplify the result. Taking $\sigma_W^2 = 16n^2/\mu^2 \geq \bb E\brac{W_k^2}$, $R_W = 4n/\mu$ and $\sigma_V^2 = 64n^3\geq \bb E\brac{V_k^2}$, $R_V = 8n\sqrt{n}$, and considering $S_W = \frac{1}{p}\sum_{k=1}^p W_k$ and $S_V = \frac{1}{p}\sum_{k=1}^p V_k$, then by Lemma \ref{lem:mc_bernstein_scalar}, we obtain 
\begin{align*}
\bb P\brac{\abs{S_W - \bb E\brac{S_W}}\geq \frac{t}{2} } \;&\leq\; 2\exp\paren{- \frac{p\mu^2t^2}{128n^2 + 16n\mu t}}, \\
\bb P\brac{\abs{S_V - \bb E\brac{S_V}}\geq \frac{t}{2} } \;&\leq\; 2\exp\paren{- \frac{pt^2}{512n^3+32n\sqrt{n} t}}. 
\end{align*}
Combining the above results, we obtain
\begin{align*}
\bb P\brac{\abs{\frac{1}{p} \sum_{k=1}^p Y_k - \expect{Y_k}} \geq t} \;&=\; \bb P\brac{\abs{S_W - \bb E\brac{S_W}+ S_V - \bb E\brac{S_V} }\geq t } \nonumber \\
\; &\leq \; \bb P\brac{\abs{S_W - \bb E\brac{S_W}}\geq \frac{t}{2} }+ \bb P\brac{\abs{S_V - \bb E\brac{S_V}}\geq \frac{t}{2} } \nonumber \\
\;&\leq \; 2\exp\paren{- \frac{p\mu^2t^2}{128n^2+16n\mu t}}+2\exp\paren{- \frac{pt^2}{512n^3+32n\sqrt{n} t}} \nonumber \\
\;&\leq\; 4\exp\paren{- \frac{p\mu^2t^2}{512n^2+32n\mu t}}, 
\end{align*}
provided that $\mu \le 1/\sqrt{n}$, as desired. 
\end{proof}

\begin{proof}{\textbf{(of Proposition \ref{prop:concentration-hessian-zero})}} \label{proof:pt_cn_strcvx}
Let $\mb Z_k =\nabla^2_{\mb w} h_\mu\paren{\mb q(\mb w)^*(\mb x_0)_k}$, then $\nabla^2_{\mb w} g\paren{\mb w} = \frac{1}{p} \sum_{k=1}^p \mb Z_k$. From \eqref{eqn:lse-hessian}, we know that
\begin{align*}
\mb Z_k \;=\; \mb W_k + \mb V_k
\end{align*}
where 
\begin{align*}
\mb W_k \;&=\; \frac{1}{\mu} \paren{1 - \tanh^2\paren{\frac{\mb q(\mb w)^*(\mb x_0)_k}{\mu }}}\paren{\overline{\mb x_0}_k - \frac{x_{0k}\paren{n}\mb w}{q_n(\mb w)}}\paren{\overline{\mb x_0}_k - \frac{x_{0k}\paren{n}\mb w}{q_n(\mb w)}}^* \\
\mb V_k \;&=\; -\tanh\paren{\frac{\mb q(\mb w)^*(\mb x_0)_k}{\mu}}\paren{\frac{x_{0k}\paren{n}}{q_n(\mb w)}\mb I + \frac{x_{0k}\paren{n}\mb w\mb w^*}{q_n^3(\mb w)}}. 
\end{align*}
For $\mb W_k$, we have
\begin{align*}
\mb 0 \preceq \bb E\brac{\mb W_k^m} 
&\preceq \frac{1}{\mu^m} \bb E\brac{ \norm{\overline{\mb x_0}_k - \frac{x_{0k}\paren{n}\mb w}{q_n(\mb w)}}{}^{2m-2} \paren{\overline{\mb x_0}_k - \frac{x_{0k}\paren{n}\mb w}{q_n(\mb w)}}\paren{\overline{\mb x_0}_k - \frac{x_{0k}\paren{n}\mb w}{q_n(\mb w)}}^* } \nonumber \\
& \preceq \frac{1}{\mu^m} \bb E\brac{ \norm{\overline{\mb x_0}_k - \frac{x_{0k}\paren{n}\mb w}{q_n(\mb w)}}{}^{2m}} \mb I \nonumber \\
\;&\preceq\; \frac{2^m}{\mu^m} \bb E\brac{ \paren{\norm{\overline{\mb x_0}_k}{}^2 + \frac{x^2_{0k}\paren{n}\norm{\mb w}{}^2}{q^2_n(\mb w)}}^{m} } \mb I \nonumber \\
\;&\preceq\;\frac{2^m}{\mu^m} \bb E\brac{ \norm{(\mb x_0)_k}{}^{2m} } \mb I
\;\preceq\; \frac{2^m}{\mu^m}\bb E_{Z \sim \chi^2\paren{n}}\brac{Z^m}\mb I,
\end{align*}
where we have used the fact that $\norm{\mb w}{}^2/q_n^2(\mb w) = \norm{\mb w}{}^2/(1-\norm{\mb w}{}^2)\leq 1 $ for $\norm{\mb w}{2}\leq 1/4$ and Lemma~\ref{lem:U-moments-bound} to obtain the last line. By Lemma \ref{lem:chi_sq_moment}, we obtain 
\begin{align*}
\mb 0 \preceq \bb E\brac{\mb W_k^m} \;\preceq\; \paren{\frac{2}{\mu}}^m \frac{m!}{2} \paren{2n}^m \mb I \;=\; \frac{m!}{2}\paren{\frac{4n}{\mu}}^{m} \mb I. 
\end{align*}
Taking $R_W = 4n/\mu$ and $\mb \sigma_W^2 =16n^2/\mu^2 \ge \expect{\mb W_k^2}$, and letting $\mb S_{W} \doteq \frac{1}{p} \sum_{k=1}^p \mb W_k$, by Lemma~\ref{lem:mc_bernstein_matrix}, we obtain 
\begin{align*}
\bb P\brac{\norm{\mb S_W - \bb E\brac{\mb S_W}}{}\geq \frac{t}{2}}\;&\leq\; 2n\exp\paren{-\frac{p\mu^2t^2}{128n^2+16\mu n t}}.
\end{align*}
Similarly, for $\mb V_k$, we have
\begin{align*}
\bb E\brac{\mb V_k^m} \;&\preceq\; \paren{\frac{1}{q_n(\mb w)} + \frac{\norm{\mb w}{}^2}{q_n^3(\mb w)}}^m\bb E\brac{ \abs{x_k\paren{n}}^m }\mb I \nonumber \\
\;&\preceq\;  \paren{8n\sqrt{n}}^m \paren{m-1}!! \mb I \nonumber \\
\;&\preceq\; \frac{m!}{2}\paren{8n\sqrt{n}}^{m} \mb I,  
\end{align*}
where we have used the fact $q_n\paren{\mb w} \ge 1/(2\sqrt{n})$ to simplify the result. Similar argument also shows $-\bb E\brac{\mb V_k^m} \preceq m!\paren{8n\sqrt{n}}^m \mb I /2$. Taking $R_V = 8 n\sqrt{n}$ and $\mb \sigma_V^2 = 64n^3$, and letting $\mb S_V \doteq \frac{1}{p} \sum_{k=1}^p \mb V_k$, again by Lemma \ref{lem:mc_bernstein_matrix}, we obtain 
\begin{align*}
\bb P\brac{\norm{\mb S_V - \bb E\brac{\mb S_V}}{}\geq \frac{t}{2}}\;&\leq\; 2n\exp\paren{- \frac{pt^2}{512n^3+32n\sqrt{n}t}}.
\end{align*}
Combining the above results, we obtain 
\begin{align*}
\bb P\brac{\norm{\frac{1}{p}\sum_{k=1}^p \mb Z_k - \expect{\mb Z_k} }{}\geq t} \;&=\; \bb P\brac{\norm{\mb S_W -\bb E\brac{\mb S_W}+\mb S_V -\bb E\brac{\mb S_V} }{}\geq t} \nonumber \\
\;&\leq \; \bb P\brac{\norm{\mb S_W -\bb E\brac{\mb S_W} }{}\geq \frac{t}{2}} + \bb P\brac{\norm{\mb S_V -\bb E\brac{\mb S_V} }{}\geq \frac{t}{2}} \nonumber \\
\;&\leq \; 2n\exp\paren{-\frac{p\mu^2t^2}{128n^2+16\mu n t}}+ 2n\exp\paren{- \frac{pt^2}{512n^3+32n\sqrt{n}t}} \nonumber \\
\;&\leq \; 4n\exp\paren{-\frac{p\mu^2t^2}{512n^2+32\mu n t}},
\end{align*}
where we have simplified the final result using $\mu \leq 1/\sqrt{n}$.
\end{proof}

\subsubsection{Proof of Lipschitz Results} \label{proof:cn_lips}
We need the following lemmas to prove the Lipschitz results.

\begin{lemma}\label{lem:composition} Suppose that $\varphi_1 : U \to V$ is an $L$-Lipschitz map from a normed space $U$ to a normed space $V$, and that $\varphi_2 : V \to W$ is an $L'$-Lipschitz map from $V$ to a normed space $W$. Then the composition $\varphi_2 \circ \varphi_1 : U \to W$ is $LL'$-Lipschitz. 
\end{lemma}

\begin{lemma}\label{lem:Lip-combined} Fix any $\mc D \subseteq \reals^{n-1}$. Let $g_1, g_2 : \mc D \to \reals$, and assume that $g_1$ is $L_1$-Lipschitz, and $g_2$ is $L_2$-Lipschitz, and that $g_1$ and $g_2$ are bounded over $\mc D$, i.e., $\abs{g_1(\mb x)}\le M_1$ and $\abs{g_2(\mb x)} \le M_2$ for all $x\in \mc D$ with some constants $M_1>0$ and $M_2>0$. Then the function $h(\mb x) = g_1(\mb x) g_2(\mb x)$ is $L$-Lipschitz, with 
\begin{align*}
L \;=\; M_1 L_2 + M_2 L_1. 
\end{align*} 
\end{lemma}

\begin{lemma}\label{lem:lip-h-mu}
For every $\mb w, \mb w^\prime \in \Gamma$, and every fixed $\mb x$, we have
\begin{align*}
\abs{\dot{h}_\mu\paren{\mb q(\mb w)^*\mb x} -\dot{h}_\mu\paren{\mb q(\mb w^\prime)^*\mb x} }\;&\le\; \frac{2\sqrt{n}}{\mu}\norm{\mb x}{} \norm{\mb w-\mb w^\prime}{}, \\
\abs{\ddot{h}_\mu\paren{\mb q(\mb w)^*\mb x} -\ddot{h}_\mu\paren{\mb q(\mb w^\prime)^*\mb x} }\;&\le\; \frac{4\sqrt{n}}{\mu^2}\norm{\mb x}{} \norm{\mb w-\mb w^\prime}{}.
\end{align*}
\end{lemma}
\begin{proof}
We have 
\begin{align*}
\abs{q_n\paren{\mb w} - q_n\paren{\mb w'}} 
& =  \abs{\sqrt{1-\norm{\mb w}{}^2} - \sqrt{1-\norm{\mb w'}{}^2}}
  = \frac{\norm{\mb w + \mb w'}{} \norm{\mb w - \mb w'}{}}{\sqrt{1-\norm{\mb w}{}^2} + \sqrt{1-\norm{\mb w'}{}^2}} \\
& \le \frac{\max\paren{\norm{\mb w}{}, \norm{\mb w'}{}}}{\min\paren{q_n\paren{\mb w}, q_n\paren{\mb w'}}} \norm{\mb w - \mb w'}{}. 
\end{align*}
Hence it holds that 
\begin{align*}
\norm{\mb q\paren{\mb w} - \mb q\paren{\mb w'}}{}^2 
& = \norm{\mb w - \mb w'}{}^2 + \abs{q_n\paren{\mb w} - q_n\paren{\mb w'}}^2 \le \paren{1+\frac{\max\paren{\norm{\mb w}{}^2, \norm{\mb w'}{}^2}}{\min\paren{q_n^2\paren{\mb w}, q_n^2\paren{\mb w'}}}} \norm{\mb w - \mb w'}{}^2 \\
& = \frac{1}{\min\paren{q_n^2\paren{\mb w}, q_n^2\paren{\mb w'}}} \norm{\mb w - \mb w'}{}^2 \le 4n \norm{\mb w - \mb w'}{}^2, 
\end{align*}
where we have used the fact $q_n\paren{\mb w} \ge 1/(2\sqrt{n})$ to get the final result. Hence the mapping $\mb w \mapsto \mb q(\mb w)$ is $2\sqrt{n}$-Lipschitz over $\Gamma$. Moreover it is easy to see $\mb q \mapsto \mb q^*\mb x$ is $\norm{\mb x}{2}$-Lipschitz. By Lemma \ref{lem:derivatives_basic_surrogate} and the composition rule in Lemma \ref{lem:composition}, we obtain the desired claims.
\end{proof}

\begin{lemma}\label{lem:lip-g} 
For any fixed $\mb x$, consider the function
\begin{align*}
t_{\mb x}(\mb w) \doteq \frac{\mb w^*\overline{\mb x}}{\norm{\mb w}{}} - \frac{x_n}{q_n(\mb w)}\norm{\mb w}{}
\end{align*}
defined over $\mb w \in \Gamma$. Then, for all $\mb w, \mb w'$ in $\Gamma$ such that $\norm{\mb w}{} \ge r$ and $\norm{\mb w'}{} \ge r$ for any constant $r \in \paren{0, 1}$, it holds that 
\begin{align*}
\abs{ t_{\mb x}(\mb w)- t_{\mb x}(\mb w^\prime)}\;&\le \; 2\paren{\frac{\norm{\mb x}{}}{r}+ 4n^{3/2} \norm{\mb x}{\infty}} \norm{\mb w-\mb w^\prime}{} , \\
 \abs{ t_{\mb x}(\mb w) }\; &\le \; 2\sqrt{n}\norm{\mb x}{}, \\
\abs{ t^2_{\mb x}(\mb w)- t^2_{\mb x}(\mb w^\prime)}\;&\le \;8\sqrt{n} \norm{\mb x}{} \paren{\frac{\norm{\mb x}{}}{r}+ 4n^{3/2} \norm{\mb x}{\infty}} \norm{\mb w-\mb w^\prime}{},  \\
 \abs{ t^2_{\mb x}(\mb w) }\; &\le \;  4n\norm{\mb x}{}^2. 
\end{align*}
\end{lemma}

\begin{proof}
First of all, we have 
\begin{align*}
\abs{ t_{\mb x}(\mb w) }\; =\; \brac{\frac{\mb w^*}{\norm{\mb w}{}},-\frac{\norm{\mb w}{}}{q_n(\mb w)}}\mb x \;\le\; \norm{\mb x}{}\paren{1+ \frac{\norm{\mb w}{}^2}{q_n^2(\mb w)}}^{1/2} =\frac{\norm{\mb x}{}}{\abs{q_n(\mb w)}}\le 2\sqrt{n}\norm{\mb x}{}, 
\end{align*}
where we have used the assumption that $q_n\paren{\mb w} \ge 1/(2\sqrt{n})$ to simplify the final result. The claim about $\abs{t_{\mb x}^2\paren{\mb w}}$ follows immediately. Now 
\begin{align*}
\abs{ t_{\mb x}(\mb w)- t_{\mb x}(\mb w^\prime)} \le \abs{\paren{\frac{\mb w}{\norm{\mb w}{}} -\frac{\mb w^\prime}{\norm{\mb w^\prime}{}}}^*\overline{\mb x}  } + \abs{x_n}\abs{ \frac{\norm{\mb w}{}}{q_n(\mb w)}- \frac{\norm{\mb w^\prime}{}}{q_n(\mb w^\prime)} }. 
\end{align*}
Moreover we have 
\begin{align*}
\abs{\paren{\frac{\mb w}{\norm{\mb w}{}} -\frac{\mb w^\prime}{\norm{\mb w^\prime}{}}}^*\overline{\mb x}  } 
& \le \norm{\overline{\mb x}}{}\norm{\frac{\mb w}{\norm{\mb w}{}} -\frac{\mb w^\prime}{\norm{\mb w^\prime}{}}}{} \le \norm{\mb x}{} \frac{\norm{\mb w - \mb w'}{} \norm{\mb w'}{} + \norm{\mb w'}{} \abs{\norm{\mb w}{} - \norm{\mb w'}{}}  }{\norm{\mb w}{} \norm{\mb w'}{}} \\
& \le \frac{2\norm{\mb x}{}}{r} \norm{\mb w - \mb w'}{},  
\end{align*}
where we have used the assumption that $\norm{\mb w}{} \ge r$ to simplify the result. Noticing that $t \mapsto t/\sqrt{1-t^2}$ is continuous over $\brac{a, b}$ and differentiable over $\paren{a, b}$ for any $0 < a < b < 1$, by mean value theorem, 
\begin{align*}
\abs{ \frac{\norm{\mb w}{}}{q_n(\mb w)}- \frac{\norm{\mb w^\prime}{}}{q_n(\mb w^\prime)} } \;\le\; \sup_{\mb w\; \in\; \Gamma } \frac{1}{\paren{1-\norm{\mb w}{}^2}^{3/2}} \norm{\mb w-\mb w^\prime}{} \; \le \; 8n^{3/2} \norm{\mb w-\mb w^\prime}{}, 
\end{align*}
where we have again used the assumption that $q_n\paren{\mb w} \ge 1/(2\sqrt{n})$ to simplify the last result. Collecting the above estimates, we obtain 
\begin{align*}
\abs{ t_{\mb x}(\mb w)- t_{\mb x}(\mb w^\prime)} \le \paren{2\frac{\norm{\mb x}{}}{r}+ 8n^{3/2} \norm{\mb x}{\infty}} \norm{\mb w-\mb w^\prime}{}, 
\end{align*}
as desired. For the last one, we have 
\begin{align*}
\abs{ t^2_{\mb x}(\mb w)- t^2_{\mb x}(\mb w^\prime)} \;&=\; \abs{ t_{\mb x}(\mb w)- t_{\mb x}(\mb w^\prime)} \abs{ t_{\mb x}(\mb w)+t_{\mb x}(\mb w^\prime)} \nonumber \\
\;&\le \; 2\sup_{\mb s\; \in\; \Gamma} \abs{t_{\mb x}(\mb s)}\abs{ t_{\mb x}(\mb w)- t_{\mb x}(\mb w^\prime)}, 
\end{align*}
leading to the claimed result once we substitute estimates of the involved quantities. 
\end{proof}

\begin{lemma}\label{lem:lp-Phi}
For any fixed $\mb x$, consider the function 
\begin{align*}
\mb \Phi_{\mb x}(\mb w) = \frac{x_n}{q_n(\mb w)}\mb I + \frac{x_n}{q_n^3(\mb w)}\mb w\mb w^*
\end{align*}
defined over $\mb w \in \Gamma$. Then, for all $\mb w, \mb w' \in \Gamma$ such that $\norm{\mb w}{} < r$ and $\norm{\mb w'}{} < r$ with any constant $r \in \paren{0, 1/2}$, it holds that 
\begin{align*}
 \norm{\mb \Phi_{\mb x}(\mb w) }{} \;&\le\;2\norm{\mb x}{\infty},  \\
\norm{\mb \Phi_{\mb x}(\mb w) - \mb \Phi_{\mb x}(\mb w^\prime)}{}\;&\le\;  4\norm{\mb x}{\infty}\norm{\mb w - \mb w^\prime}{}.  
\end{align*}
\end{lemma}
\begin{proof}
Simple calculation shows 
\begin{align*}
\norm{\mb \Phi_{\mb x}(\mb w) }{} \le \norm{\mb x}{\infty} \paren{ \frac{1}{q_n(\mb w)}+ \frac{\norm{\mb w}{}^2}{q_n^3(\mb w)} } = \frac{\norm{\mb x}{\infty}}{q_n^3(\mb w)}\le \frac{\norm{\mb x}{\infty}}{(1-r^2)^{3/2}} \le 2\norm{\mb x}{\infty}. 
\end{align*}
For the second one, we have 
\begin{align*}
\norm{\mb \Phi_{\mb x}(\mb w)-\mb \Phi_{\mb x}(\mb w^\prime)}{}
&\le \norm{\mb x}{\infty} \norm{ \frac{1}{q_n(\mb w)}\mb I + \frac{1}{q_n^3(\mb w)}\mb w\mb w^* -  \frac{1}{q_n(\mb w^\prime)}\mb I - \frac{1}{q_n^3(\mb w^\prime)}\mb w^\prime(\mb w^\prime)^*}{} \\
&\le \norm{\mb x}{\infty} \paren{\abs{\frac{1}{q_n\paren{\mb w}} - \frac{1}{q_n\paren{\mb w'}}} + \abs{\frac{\norm{\mb w}{}^2}{q_n^3\paren{\mb w}} - \frac{\norm{\mb w'}{}^2}{q_n^3\paren{\mb w'}}}}. 
\end{align*}
Now 
\begin{align*}
\abs{\frac{1}{q_n\paren{\mb w}} - \frac{1}{q_n\paren{\mb w'}}} = \frac{\abs{q_n\paren{\mb w} - q_n\paren{\mb w'}}}{q_n\paren{\mb w} q_n\paren{\mb w'}} \le \frac{\max\paren{\norm{\mb w}{}, \norm{\mb w'}{}}}{\min\paren{q_n^3\paren{\mb w}, q_n^3\paren{\mb w'}}} \norm{\mb w - \mb w'}{} \le \frac{4}{3\sqrt{3}} \norm{\mb w - \mb w'}{}, 
\end{align*}
where we have applied the estimate for $\abs{q_n\paren{\mb w} - q_n\paren{\mb w'}}$ as established in Lemma~\ref{lem:lip-h-mu} and also used $\norm{\mb w}{} \le 1/2$ and $\norm{\mb w'}{} \le 1/2$ to simplify the above result. Further noticing $t \mapsto t^2/\paren{1-t^2}^{3/2}$ is differentiable over $t \in \paren{0, 1}$, we apply the mean value theorem and obtain 
\begin{align*}
\abs{\frac{\norm{\mb w}{}^2}{q_n^3\paren{\mb w}} - \frac{\norm{\mb w'}{}^2}{q_n^3\paren{\mb w'}}} \le \sup_{\mb s \in \Gamma, \norm{\mb s}{} \le r < \frac{1}{2} } \frac{\norm{\mb s}{}^3 + 2\norm{\mb s}{}}{\paren{1-\norm{\mb s}{}^2}^{5/2}} \norm{\mb w - \mb w'}{} \le \frac{4}{\sqrt{3}} \norm{\mb w - \mb w'}{}. 
\end{align*}
Combining the above estimates gives the claimed result. 
\end{proof}

\begin{lemma}\label{lem:lp-zeta}
For any fixed $\mb x$, consider the function
\begin{align*}
\mb \zeta_{\mb x}(\mb w) = \overline{\mb x} - \frac{x_n}{q_n(\mb w)}\mb w
\end{align*}
defined over $\mb w \in \Gamma$. Then, for all $\mb w, \mb w' \in \Gamma$ such that $\norm{\mb w}{} \le r$ and $\norm{\mb w'}{} \le r$ for any constant $r \in \paren{0, 1/2}$, it holds that 
\begin{align*}
 \norm{\mb \zeta_{\mb x}(\mb w) \mb \zeta_{\mb x}(\mb w)^*  }{} \;&\le\;2n \norm{\mb x}{\infty}^2,  \\
\norm{\mb \zeta_{\mb x}(\mb w) \mb \zeta_{\mb x}(\mb w)^* - \mb \zeta_{\mb x}(\mb w^\prime) \mb \zeta_{\mb x}(\mb w^\prime)^* }{}\;&\le\; 8\sqrt{2}\sqrt{n}  \norm{\mb x}{\infty}^2 \norm{\mb w - \mb w^\prime}{}.
\end{align*}
\end{lemma}
\begin{proof}
We have $\norm{\mb w}{}^2/q_n^2\paren{\mb w} \le 1/3$ when $\norm{\mb w}{} \le r < 1/2$, hence it holds that 
\begin{align*}
\norm{\mb \zeta_{\mb x}(\mb w) \mb \zeta_{\mb x}(\mb w)^*  }{} \le \norm{\mb \zeta_{\mb x}(\mb w)}{}^2 \le 2\norm{\overline{\mb x}}{}^2 + 2x_n^2\frac{\norm{\mb w}{}}{q_n^2\paren{\mb w}} \le 2n \norm{\mb x}{\infty}^2. 
\end{align*}
For the second, we first estimate  
\begin{align*}
\norm{\mb \zeta(\mb w) - \mb \zeta(\mb w^\prime)}{} 
& = \norm{x_n\paren{\frac{\mb w}{q_n\paren{\mb w}}  - \frac{\mb w'}{q_n\paren{\mb w'}}}}{} \le \norm{\mb x}{\infty} \norm{\frac{\mb w}{q_n\paren{\mb w}}  - \frac{\mb w'}{q_n\paren{\mb w'}}}{} \\
&\le  \norm{\mb x}{\infty} \paren{ \frac{1}{q_n(\mb w)}\norm{\mb w - \mb w^\prime}{}+  \norm{\mb w^\prime}{} \abs{\frac{1}{q_n(\mb w)} - \frac{1}{q_n(\mb w^\prime)}} } \\
&\le  \norm{\mb x}{\infty} \paren{ \frac{1}{q_n(\mb w)}+ \frac{ \norm{\mb w^\prime}{}}{\min\Brac{q_n^3(\mb w),q_n^3(\mb w^\prime)}} }\norm{\mb w - \mb w^\prime}{}  \\
&\le  \norm{\mb x}{\infty} \paren{\frac{2}{\sqrt{3}}+ \frac{4}{3\sqrt{3}}}\norm{\mb w - \mb w^\prime}{} \le 4\norm{\mb x}{\infty}\norm{\mb w - \mb w^\prime}{}. 
\end{align*}
Thus, we have 
\begin{align*}
\norm{\mb \zeta_{\mb x}(\mb w) \mb \zeta_{\mb x}(\mb w)^* - \mb \zeta_{\mb x}(\mb w^\prime) \mb \zeta_{\mb x}(\mb w^\prime)^* }{}\;&\le\; \norm{\mb \zeta(\mb w)}{} \norm{\mb \zeta(\mb w) - \zeta(\mb w^\prime)}{}+ \norm{\mb \zeta(\mb w) - \zeta(\mb w^\prime)}{}\norm{\mb \zeta(\mb w^\prime)}{} \nonumber \\
\;&\le \; 8\sqrt{2}\sqrt{n}  \norm{\mb x}{\infty}^2 \norm{\mb w - \mb w^\prime}{}, 
\end{align*}
as desired.
\end{proof}

Now, we are ready to prove all the Lipschitz propositions.

\begin{proof}{\textbf{(of Proposition \ref{prop:lip-hessian-negative})}} \label{proof:cn_lips_curvature}
Let 
\begin{align*}
F_k(\mb w)=\ddot{h}_{\mu}\paren{\mb q(\mb w)^* (\mb x_0)_k} t^2_{(\mb x_0)_k}(\mb w) + \dot{h}_{\mu}\paren{\mb q(\mb w)^* (\mb x_0)_k} \frac{x_{0k}\paren{n}}{q_n^3(\mb w)}. 
\end{align*}
Then, $\mb w^*\nabla^2 g(\mb w; \mb X_0) \mb w / \norm{\mb w}{}^2 = \frac{1}{p} \sum_{k=1}^p F_k(\mb w)$. Noticing that $\ddot{h}_{\mu}\paren{\mb q(\mb w)^* (\mb x_0)_k}$ is bounded by $1/\mu$ and $\dot{h}_{\mu}\paren{\mb q(\mb w)^* (\mb x_0)_k}$ is bounded by $1$, both in magnitude. Applying Lemma \ref{lem:Lip-combined}, Lemma \ref{lem:lip-h-mu} and Lemma \ref{lem:lip-g}, we can see $F_k(\mb w)$ is $\Lconcave^k$-Lipschitz with
\begin{align*}
\Lconcave^k &= 4n \norm{(\mb x_0)_k}{}^2 \frac{4\sqrt{n}}{\mu^2} \norm{(\mb x_0)_k}{}+ \frac{1}{\mu}8\sqrt{n} \norm{(\mb x_0)_k}{} \paren{\frac{\norm{(\mb x_0)_k}{}}{\rconcave} + 4n^{3/2} \norm{(\mb x_0)_k}{\infty}} \nonumber \\
&+(2\sqrt{n})^3\norm{(\mb x_0)_k}{\infty} \frac{2\sqrt{n}}{\mu} \norm{(\mb x_0)_k}{}+ \sup_{\rconcave < a < \sqrt{\frac{2n-1}{2n}}} \frac{3}{\paren{1-a^2}^{5/2}} \norm{(\mb x_0)_k}{\infty} \nonumber \\
&=\frac{16n^{3/2}}{\mu^2} \norm{(\mb x_0)_k}{}^3 + \frac{8\sqrt{n}}{\mu \rconcave} \norm{(\mb x_0)_k}{}^2 + \frac{48 n^2}{\mu}\norm{(\mb x_0)_k}{} \norm{(\mb x_0)_k}{\infty} + 96 n^{5/2} \norm{(\mb x_0)_k}{\infty}. 
\end{align*}
Thus, $\frac{1}{\norm{\mb w}{2}}\mb w^*\nabla^2 g(\mb w; \mb X_0) \mb w$ is $\Lconcave$-Lipschitz with
\begin{align*}
\Lconcave \le \frac{1}{p}\sum_{k=1}^p \Lconcave^k \le \frac{16n^3}{\mu^2} \norm{\mb X_0}{\infty}^3 + \frac{8n^{3/2}}{\mu \rconcave} \norm{\mb X_0}{\infty}^2 + \frac{48 n^{5/2} }{\mu} \norm{\mb X_0}{\infty}^2 + 96 n^{5/2} \norm{\mb X_0}{\infty}, 
\end{align*}
as desired. 
\end{proof}

\begin{proof}{\textbf{(of Proposition \ref{prop:lip-gradient})}} \label{proof:cn_lips_gradient}
We have 
\begin{align*}
\norm{\frac{\mb w^*}{\norm{\mb w}{}}\nabla g(\mb w; \mb X_0) - \frac{\mb w'^*}{\norm{\mb w'}{}}\nabla g(\mb w^\prime; \mb X_0) }{} \le \frac{1}{p}\sum_{k=1}^p \norm{\dot{h}_{\mu}\paren{\mb q(\mb w)^* (\mb x_0)_k} t_{(\mb x_0)_k}\paren{\mb w}  - \dot{h}_{\mu}\paren{\mb q(\mb w^\prime)^* (\mb x_0)_k} t_{(\mb x_0)_k}\paren{\mb w'}}{}
\end{align*}
where $\dot{h}_\mu (t) = \tanh(t/\mu)$ is bounded by one in magnitude, and $t_{(\mb x_0)_k}(\mb w)$ and $t_{(\mb x_0)_k^\prime}(\mb w)$ is defined as in Lemma \ref{lem:lip-g}. By Lemma \ref{lem:Lip-combined}, Lemma \ref{lem:lip-h-mu} and Lemma \ref{lem:lip-g}, we know that $\dot{h}_{\mu}\paren{\mb q(\mb w)^* (\mb x_0)_k} t_{(\mb x_0)_k}\paren{\mb w}$ is $L_k$-Lipschitz with constant
\begin{align*}
L_k = \frac{2\norm{(\mb x_0)_k}{}}{r_g}+ 8n^{3/2} \norm{(\mb x_0)_k}{\infty} + \frac{4n}{\mu} \norm{(\mb x_0)_k}{}^2. 
\end{align*}
Therefore, we have
\begin{align*}
\norm{\frac{\mb w^*}{\norm{\mb w}{}}\nabla g(\mb w) - \frac{\mb w^*}{\norm{\mb w}{}}\nabla g(\mb w^\prime) }{} 
& \le \frac{1}{p}\sum_{k=1}^p\paren{ \frac{2\norm{(\mb x_0)_k}{}}{r_g}+ 8n^{3/2} \norm{(\mb x_0)_k}{\infty} + \frac{4n}{\mu} \norm{(\mb x_0)_k}{}^2} \norm{\mb w-\mb w^\prime}{} \nonumber \\
&\le \paren{\frac{2\sqrt{n}}{r_g}\norm{\mb X_0}{\infty}+ 8n^{3/2} \norm{\mb X_0}{\infty} + \frac{4n^2}{\mu}\norm{\mb X_0}{\infty}^2  }\norm{\mb w-\mb w^\prime}{}, 
\end{align*}
as desired. 
\end{proof}

\begin{proof}{\textbf{(of Proposition \ref{prop:lip-hessian-zero})}} \label{proof:cn_lips_strcvx}
Let 
\begin{align*}
\mb F_k(\mb w) &= \ddot{h}_\mu (\mb q(\mb w)^*(\mb x_0)_k)\mb \zeta_k(\mb w) \mb \zeta_k(\mb w)^*- \dot{h}_\mu\paren{\mb q(\mb w)^* (\mb x_0)_k}\bm \Phi_k(\mb w)
\end{align*}
with $\mb \zeta_k(\mb w) = \overline{\mb x_0}_k - \frac{x_{0k}\paren{n}}{q_n(\mb w)}\mb w$ and $\mb \Phi_k(\mb w) = \frac{x_{0k}\paren{n}}{q_n(\mb w)} \mb I + \frac{x_{0k}(n)}{q_n^3(\mb w)}\mb w\mb w^*$. Then, $\nabla^2 g(\mb w) =  \frac{1}{p}\sum_{k=1}^p \mb F_k(\mb w)$. Using Lemma \ref{lem:Lip-combined}, Lemma \ref{lem:lip-h-mu}, Lemma \ref{lem:lp-Phi} and Lemma \ref{lem:lp-zeta}, and the facts that $\ddot{h}_\mu(t)$ is bounded by $1/\mu$ and that $\ddot{h}_\mu(t)$ is bounded by $1$ in magnitude, we can see $\mb F_k(\mb w)$ is $\Lconvex^k$-Lipschitz continuous with 
\begin{align*}
\Lconvex^k 
& = \frac{1}{\mu}\times 8\sqrt{2}\sqrt{n}\norm{(\mb x_0)_k}{\infty}^2+ \frac{2\sqrt{n}}{\mu^2}\norm{(\mb x_0)_k}{} \times 2n \norm{(\mb x_0)_k}{\infty}^2 + 4\norm{(\mb x_0)_k}{\infty}+ \frac{2\sqrt{n}}{\mu}\norm{(\mb x_0)_k}{} \times 2\norm{(\mb x_0)_k}{\infty} \nonumber \\
& \le \frac{4n^{3/2}}{\mu^2} \norm{(\mb x_0)_k}{} \norm{(\mb x_0)_k}{\infty}^2+\frac{4\sqrt{n}}{\mu}\norm{(\mb x_0)_k}{}\norm{(\mb x_0)_k}{\infty} +  \frac{8\sqrt{2}\sqrt{n}}{\mu}\norm{(\mb x_0)_k}{\infty}^2 + 4\norm{(\mb x_0)_k}{\infty}. 
\end{align*}
Thus, we have
\begin{align*}
\Lconvex \le \frac{1}{p}\sum_{k=1}^p \Lconvex^k \le \frac{4n^2}{\mu^2} \norm{\mb X}{\infty}^3+\frac{4n}{\mu}\norm{\mb X}{\infty}^2 +  \frac{8\sqrt{2}\sqrt{n}}{\mu}\norm{\mb X}{\infty}^2 + 8\norm{\mb X}{\infty}, 
\end{align*}
as desired. 
\end{proof}

\subsection{Proofs of Theorem~\ref{thm:geometry_orth}} \label{sec:proof_geometry_orth}
Before proving Theorem \ref{thm:geometry_orth}, we record one useful lemma. 

\begin{lemma}\label{lem:X-infinty-tail-bound}
For any $\theta \in \paren{0, 1}$, consider the random matrix $\mb X \in \R^{n_1 \times n_2}$ with $\mb X \sim_{i.i.d.} \mathrm{BG}\paren{\theta}$. Define the event $\event_\infty \doteq \Brac{ 1 \le \norm{\mb X}{\infty}\leq 4\sqrt{\log\paren{n p}}}$. It holds that 
\begin{align*}
\prob{\event_\infty^c} \leq \theta \paren{np}^{-7} + \exp\paren{-0.3\theta np}. 
\end{align*}
\end{lemma}
\begin{proof}
See page~\pageref{proof:X_inf_tail} under Section~\ref{sec:app_aux_results}. 
\end{proof}

For convenience, we define three regions for the range of $\mb w$:
\begin{align*}
R_1 & \doteq \set{ \mb w: \norm{\mb w}{} \le \frac{\mu}{4\sqrt{2}}}, \qquad 
R_2 \doteq \set{ \mb w: \frac{\mu}{4\sqrt{2}} \le \norm{\mb w}{} \le \frac{1}{20\sqrt{5}}}, \\
R_3 & \doteq \set{ \mb w: \frac{1}{20\sqrt{5}} \le \norm{\mb w}{} \le \sqrt{\frac{4n-1}{4n}}}. 
\end{align*}

\begin{proof}{\textbf{(of Theorem \ref{thm:geometry_orth})}} 
\paragraph{Strong convexity in region $R_1$.} Proposition~\ref{prop:geometry_asymp_strong_convexity} shows that for any $\mb w \in R_1$, $\bb E \left[ \nabla^2 g(\mb w; \mb X_0) \right] \succeq \frac{c_1\theta }{\mu} \mb I$. For any $\eps \in (0, \mu/\paren{4\sqrt{2}} )$, $R_1$ has an $\eps$-net $N_1$ of size at most $(3 \mu / \paren{4\sqrt{2}\eps})^n$. On $\event_\infty$, $\nabla^2 g$ is 
\begin{equation*}
L_1 \doteq \frac{C_2 n^2}{\mu^2} \log^{3/2}(np)
\end{equation*}
Lipschitz by Proposition~\ref{prop:lip-hessian-zero}. Set $\eps = \frac{c_1 \theta}{3 \mu L_1}$, so 
\begin{equation*}
\# N_1 \le \exp\left(2 n \log \left( \frac{ C_3 n \log( n p ) }{\theta} \right) \right). 
\end{equation*}
Let $\event_1$ denote the event
\begin{equation*}
\event_1 = \set{ \max_{\mb w \in N_1} \norm{\nabla^2 g(\mb w; \mb X_0) - \bb E\left[ \nabla^2 g(\mb w; \mb X_0) \right] }{} \le \frac{c_1 \theta }{3 \mu} }.
\end{equation*}
On $\event_1 \cap \event_\infty$, 
\begin{equation*}
\sup_{\norm{\mb w}{} \le \mu/\paren{4\sqrt{2}}} \norm{\nabla^2 g(\mb w; \mb X_0) - \bb E \left[ \nabla^2 g(\mb w; \mb X_0) \right] }{} \;\le\; \frac{2c_1 \theta}{3 \mu},
\end{equation*}
and so on $\event_1 \cap \event_\infty$, \eqref{eqn:hess-zero-uni-orth} holds for any constant $c_\star \le c_1 / 3$. Setting $t = c_1 \theta / 3 \mu$ in Proposition \ref{prop:concentration-hessian-zero}, we obtain that for any fixed $\mb w$,
\begin{equation*}
\prob{ \norm{ \nabla^2 g(\mb w; \mb X_0) - \bb E\left[ \nabla^2 g(\mb w; \mb X_0) \right] }{} \ge \frac{c_1 \theta}{3 \mu} } \le 4 n \exp\left( - \frac{c_4 p \theta^2}{n^2} \right).
\end{equation*}
Taking a union bound, we obtain that 
\begin{equation*}
\prob{ \event_1^c } \;\le\; 4 n \exp\left( - \frac{c_4 p \theta^2}{n^2} + C_5 n \log(n) + C_5 n \log \log(p) \right). 
\end{equation*}

\paragraph{Large gradient in region $R_2$.} Similarly, for the gradient quantity, for $\mb w \in R_2$, Proposition~\ref{prop:geometry_asymp_gradient} shows that 
\begin{equation*}
\bb E \left [\frac{\mb w^* \nabla g(\mb w; \mb X_0) }{\norm{\mb w}{}} \right] \;\ge\; c_6 \theta. 
\end{equation*}
Moreover, on $\event_\infty$, $\mb w^* \nabla g(\mb w; \mb X_0)/\norm{\mb w}{}$
is 
\begin{equation*}
L_2 \doteq \frac{C_7 n^2}{\mu} \log(np)
\end{equation*}
Lipschitz by Proposition~\ref{prop:lip-gradient}. For any $\eps < \frac{1}{20 \sqrt{5}}$, the set $R_2$ has an $\eps$-net $N_2$ of size at most $\left(\frac{3}{20 \eps \sqrt{5}}\right)^n$. Set $\eps = \frac{c_6 \theta}{3 L_2}$, so 
\begin{equation*}
\# N_2 \;\le\; \exp\left( n \log \left( \frac{C_8 n^2 \log(np)}{\theta \mu} \right) \right). 
\end{equation*}
Let $\event_2$ denote the event 
\begin{equation*}
\event_2 = \set{ \max_{\mb w \in N_2} \magnitude{ \frac{\mb w^* \nabla g(\mb w; \mb X_0) }{\norm{\mb w}{}} - \bb E \left [\frac{\mb w^* \nabla g(\mb w; \mb X_0) }{\norm{\mb w}{}} \right]  } \;\le\; \frac{c_6 \theta}{3} }.
\end{equation*}
On $\event_2 \cap \event_\infty$, 
\begin{equation}
\sup_{\mb w \in R_2} \magnitude{ \frac{\mb w^* \nabla g(\mb w; \mb X_0) }{\norm{\mb w}{}} - \bb E \left [\frac{\mb w^* \nabla g(\mb w; \mb X_0) }{\norm{\mb w}{}} \right]  } \;\le\; \frac{2c_6 \theta}{3},
\end{equation}
and so on $\event_2 \cap \event_\infty$, \eqref{eqn:grad-uni-orth} holds for any constant $c_\star \le c_6 / 3$. Setting $t = c_6 \theta / 3$ in Proposition \ref{prop:concentration-gradient}, we obtain that for any fixed $\mb w \in R_2$, 
\begin{equation*}
\bb P \left[ \magnitude{ \frac{\mb w^* \nabla g(\mb w; \mb X_0) }{\norm{\mb w}{}} - \bb E \left [\frac{\mb w^* \nabla g(\mb w; \mb X_0) }{\norm{\mb w}{}} } \right] \right]
 \;\le\; 2 \exp\left( - \frac{c_9 p \theta^2 }{n} \right), 
\end{equation*}
and so 
\begin{equation}
\bb P\left[ \event_2^c \right] \;\le\; 2 \exp\left( - \frac{c_9 p \theta^2}{n} + n \log\left( \frac{C_8 n^2 \log(np)}{\theta \mu} \right) \right). 
\end{equation}

\paragraph{Existence of negative curvature direction in $R_3$.} Finally, for any $\mb w \in R_3$, Proposition~\ref{prop:geometry_asymp_curvature} shows that 
\begin{equation*}
\bb E\left[ \frac{ \mb w^* \nabla^2 g(\mb w; \mb X_0) \mb w }{\norm{\mb w}{}^2} \right] \;\le\; -c_9 \theta. 
\end{equation*}
On $\event_\infty$, $\mb w^* \nabla^2 g(\mb w; \mb X_0) \mb w /\norm{\mb w}{}^2$ is 
\begin{equation*}
L_3 = \frac{C_{10} n^3}{\mu^2} \log^{3/2}(np) 
\end{equation*}
Lipschitz by Proposition~\ref{prop:lip-hessian-negative}. As above, for any $\eps \le \sqrt{ \frac{4n - 1}{4n} }$, $R_3$ has an $\eps$-net $N_3$ of size at most $(3/\eps)^n$. Set $\eps = c_9 \theta / 3 L_3$. Then 
\begin{equation*}
\# N_3 \;\le\; \exp\left( n \log \left( \frac{C_{11} n^3 \log^{3/2}(np)}{\theta \mu^2 } \right) \right).
\end{equation*}
Let $\event_3$ denote the event 
\begin{equation*}
\event_3 = \set{ \max_{\mb w \in N_3} \magnitude{ \frac{ \mb w^* \nabla^2 g(\mb w; \mb X_0) \mb w }{\norm{\mb w}{}^2} - \bb E\left[ \frac{ \mb w^* \nabla^2 g(\mb w; \mb X_0) \mb w }{\norm{\mb w}{}^2} \right]  } \le \frac{c_9 \theta}{3} }
\end{equation*}
On $\event_3 \cap \event_\infty$, 
\begin{equation*}
\sup_{\mb w \in R_3} \magnitude{ \frac{ \mb w^* \nabla^2 g(\mb w; \mb X_0) \mb w }{\norm{\mb w}{}^2} - \bb E\left[ \frac{ \mb w^* \nabla^2 g(\mb w; \mb X_0) \mb w }{\norm{\mb w}{}^2} \right]  } \;\le\; \frac{ 2c_9 \theta}{3},
\end{equation*}
and \eqref{eqn:curvature-uni-orth} holds with any constant $c_\star < c_9 / 3$. Setting $t = c_9 \theta / 3$ in Proposition \ref{prop:concentration-hessian-negative} and taking a union bound, we obtain 
\begin{equation*}
\prob{\event_3^c} \;\le\; 4 \exp\left( - \frac{c_{12} p \mu^2 \theta^2}{n^2} + n \log\left( \frac{C_{11} n^3 \log^{3/2}(np)}{\theta \mu^2} \right) \right). 
\end{equation*}


\paragraph{The unique local minimizer located near $\mb 0$. } Let $\event_g$ be the event that the bounds \eqref{eqn:hess-zero-uni-orth}-\eqref{eqn:curvature-uni-orth} hold. On $\event_g$, the function $g$ is $\frac{c_\star\theta}{\mu}$-strongly convex over $R_1 = \set{\mb w: \norm{\mb w}{} \le \mu/\paren{4\sqrt{2}}}$. This implies that $f$ has at most one local minimum on $R_1$. It also implies that for any $\mb w \in R_1$,
\begin{eqnarray*}
g(\mb w; \mb X_0) \ge g(\mb 0; \mb X_0) + \innerprod{\nabla g(\mb 0; \mb X_0) }{\mb w} + \frac{c\theta}{2 \mu} \norm{\mb w}{}^2 \ge g(\mb 0; \mb X_0) - \norm{\mb w}{} \norm{\nabla g(\mb 0; \mb X_0) }{} + \frac{c_\star\theta}{2 \mu} \norm{\mb w}{}^2. 
\end{eqnarray*}
So, if $g(\mb w; \mb X_0) \le g(\mb 0; \mb X_0)$, we necessarily have 
\begin{equation*}
\norm{\mb w}{} \;\le\; \frac{2 \mu}{c_\star \theta} \norm{\nabla g(\mb 0; \mb X_0)}{}.
\end{equation*}
Suppose that 
\begin{equation}
\norm{\nabla g(\mb 0; \mb X_0)}{} \le \frac{c_\star \theta}{32}. \label{eqn:grad-zero-bound}
\end{equation}
Then $g(\mb w; \mb X_0) \le g(\mb 0; \mb X_0)$ implies that $\norm{\mb w}{} \le \mu / 16$. By Wierstrass's theorem, $g(\mb w; \mb X_0)$ has at least one minimizer $\mb w_\star$ over the compact set $S = \set{\mb w: \norm{\mb w}{} \le \mu / 10}$. By the above reasoning, $\norm{\mb w_\star}{} \le \mu / 16$, and hence $\mb w_\star$ does not lie on the boundary of $S$. This implies that $\mb w_\star$ is a local minimizer of $g$. Moreover, as above, 
\begin{equation*}
\norm{\mb w_\star}{} \;\le\; \frac{2 \mu}{c_\star \theta} \norm{\nabla g(\mb 0; \mb X_0)}{}.
\end{equation*}

We now use the vector Bernstein inequality to show that with our choice of $p$, \eqref{eqn:grad-zero-bound} is satisfied w.h.p. Notice that 
\begin{equation*}
\nabla g(\mb 0; \mb X_0) = \frac{1}{p} \sum_{k = 1}^p \dot{h}_\mu( x_{0k}(n) ) \overline{\mb x_0}_k, 
\end{equation*}
and $\dot{h}_\mu$ is bounded by one in magnitude, so for any integer $m \ge 2$,
\begin{eqnarray*}
\expect{ \norm{ \dot{h}_\mu( x_{0k}(n) ) \overline{\mb x_0}_k}{}^m }
\le \expect{ \norm{(\mb x_0)_k}{}^m } 
\le \bb E_{Z \sim \chi(n)} \left[ Z^m \right] \le m!n^{m/2}, 
\end{eqnarray*}
where we have applied the moment estimate for the $\chi\paren{n}$ distribution shown in Lemma~\ref{lem:chi_moment}. Applying the vector Bernstein inequality in Corollary \ref{cor:vector-bernstein} with $R = \sqrt{n}$ and $\sigma^2 = 2n$, we obtain 
\begin{equation*}
\prob{ \norm{ \nabla g(\mb 0; \mb X_0) }{} \;\ge\; t } \;\le\; 2 (n+1) \exp\left( - \frac{pt^2}{ 4 n + 2 \sqrt{n} t } \right) 
\end{equation*}
for all $t > 0$. Using this inequality, it is not difficult to show that there exist constants $C_{13}, C_{14} > 0$ such that when $p \ge C_{13} n \log n$, with probability at least $1- 4n p^{-10}$, 
\begin{equation}
\norm{\nabla g(\mb 0; \mb X_0)}{} \;\le\; C_3 \sqrt{\frac{n \log p}{p}}. \label{eqn:grad-zero}
\end{equation}
When $\frac{p}{\log p} \ge \frac{C_{14} n}{\theta^2}$, for appropriately large $C_{14}$, \eqref{eqn:grad-zero} implies \eqref{eqn:grad-zero-bound}. Summing up failure probabilities completes the proof. 
\end{proof}

\subsection{Proofs for Section~\ref{sec:geo_results_comp} and Theorem~\ref{thm:geometry_comp}} \label{sec:proof_geometry_comp}
\begin{proof}{\textbf{(of Lemma~\ref{lem:pert_key_mag})}} \label{proof:comp_pert_bound}
By the generative model, 
\begin{align*}
\overline{\mb Y} = \sqrt{p\theta}\paren{\mb Y \mb Y^*}^{-1/2} \mb Y = \sqrt{p\theta}\paren{\mb A_0 \mb X_0 \mb X_0^* \mb A_0^*}^{-1/2} \mb A_0 \mb X_0.
\end{align*} 
Since $\expect{\mb X_0 \mb X_0^*/\paren{p \theta}} = \mb I$, we will compare $\sqrt{p\theta}\paren{\mb A_0 \mb X_0 \mb X_0^* \mb A_0^*}^{-1/2} \mb A_0$ with $\paren{\mb A_0 \mb A_0^*}^{-1/2} \mb A_0 = \mb U \mb V^*$. By Lemma~\ref{lem:half_inverse_pert}, we have
\begin{align*}
& \norm{\sqrt{p\theta}\paren{\mb A_0 \mb X_0 \mb X_0^* \mb A_0^*}^{-1/2} \mb A_0 - \paren{\mb A_0 \mb A_0^*}^{-1/2} \mb A_0}{} \nonumber \\
\le\; & \norm{\mb A_0}{} \norm{ \sqrt{p\theta}\paren{\mb A_0 \mb X_0 \mb X_0^* \mb A_0^*}^{-1/2} - \paren{\mb A_0 \mb A_0^*}^{-1/2}}{} \nonumber \\
\le\; & \norm{\mb A_0}{} \frac{2\norm{\mb A_0}{}^3}{\sigma_{\min}^4\paren{\mb A_0}} \norm{\frac{1}{p\theta} \mb X_0 \mb X_0^* - \mb I}{}  = 2\kappa^4\paren{\mb A_0} \norm{\frac{1}{p\theta} \mb X_0 \mb X_0^* - \mb I}{}
\end{align*}
provided 
\begin{align*}
\norm{\mb A_0}{}^2\norm{\frac{1}{p\theta} \mb X_0 \mb X_0^* - \mb I}{} \le \frac{\sigma_{\min}^2\paren{\mb A_0}}{2} \Longleftrightarrow \norm{\frac{1}{p\theta} \mb X_0 \mb X_0^* - \mb I}{} \le \frac{1}{2\kappa^2\paren{\mb A_0}}.  
\end{align*}
On the other hand, by Lemma~\ref{lem:bg_identity_diff}, when $p \ge C_1 n^2 \log n$,  $\norm{\tfrac{1}{p\theta} \mb X_0 \mb X_0^* - \mb I}{} \le 10 \sqrt{\tfrac{\theta n \log p}{p}}$ with probability at least $1-p^{-8}$. Thus, when $p \ge C_2 \kappa^4\paren{\mb A_0} \theta n^2 \log (n \theta \kappa\paren{\mb A_0})$, 
\begin{align*}
\norm{\sqrt{p\theta}\paren{\mb A_0 \mb X_0 \mb X_0^* \mb A_0^*}^{-1/2} \mb A_0 - \paren{\mb A_0 \mb A_0^*}^{-1/2} \mb A_0}{}  \le 20 \kappa^4\paren{\mb A_0} \sqrt{\frac{\theta n \log p}{p}},  
\end{align*}
as desired. 
\end{proof}

\begin{proof}{\textbf{(of Lemma~\ref{lem:pert_key_grad_hess})}} \label{proof:comp_pert_bound2}
Let $\widetilde{\mb Y} \doteq \mb X_0 + \widetilde{\mb \Xi} \mb X_0$. Note the Jacobian matrix for the mapping $\mb q\paren{\mb w}$ is $\nabla_{\mb w} \mb q\paren{\mb w} = \brac{\mb I, -\mb w/\sqrt{1-\norm{\mb w}{}^2}}$. Hence for any vector $\mb z \in \R^{n}$ and all $\mb w \in \Gamma$,  
\begin{align*}
\norm{\nabla_{\mb w}\mb q\paren{\mb w} \mb z}{} \le \sqrt{n-1}\norm{\mb z}{\infty} + \frac{\norm{\mb w}{}}{\sqrt{1-\norm{\mb w}{}^2}} \norm{\mb z}{\infty} \le 3\sqrt{n} \norm{\mb z}{\infty}. 
\end{align*}
Now we have 
\begin{align*}
& \norm{\nabla_{\mb w} g\paren{\mb w; \widetilde{\mb Y}} - \nabla_{\mb w{}} g\paren{\mb w; \mb X_0} }{}\\
=\; & \norm{\frac{1}{p}\sum_{k=1}^p \dot{h}_{\mu}\paren{\mb q^*\paren{\mb w} \wt{\mb y}_k}\nabla_{\mb w}\mb q\paren{\mb w} \wt{\mb y}_k - \frac{1}{p}\sum_{k=1}^p \dot{h}_{\mu}\paren{\mb q^*\paren{\mb w} (\mb x_0)_k}\nabla_{\mb w}\mb q\paren{\mb w} (\mb x_0)_k}{} \\
\le\; & \norm{\frac{1}{p}\sum_{k=1}^p \dot{h}_{\mu}\paren{\mb q^* \wt{\mb y}_k}\nabla_{\mb w}\mb q\paren{\mb w} \paren{\wt{\mb y}_k - (\mb x_0)_k}}{} \\
& \qquad + \norm{\frac{1}{p}\sum_{k=1}^p \brac{\dot{h}_{\mu}\paren{\mb q^*\paren{\mb w} \wt{\mb y}_k} - \dot{h}_{\mu}\paren{\mb q^*\paren{\mb w} (\mb x_0)_k}}\nabla_{\mb w}\mb q\paren{\mb w} (\mb x_0)_k}{}\\
\le\;& \norm{\widetilde{\mb \Xi}}{}\paren{\max_{t} \dot{h}_{\mu}\paren{t} 3n \norm{\mb X_0}{\infty} + L_{\dot{h}_{\mu}} 3n\norm{\mb X_0}{\infty}^2},  
\end{align*}
where $L_{\dot{h}_{\mu}}$ denotes the Lipschitz constant for $\dot{h}_{\mu}\paren{\cdot}$. Similarly, suppose $\norm{\widetilde{\mb \Xi}}{} \le 1/(2n)$, and also notice that 
\begin{align*}
\norm{\frac{\mb I}{q_n\paren{\mb w}} + \frac{\mb w \mb w^*}{q_n^3\paren{\mb w}}}{} \le \frac{1}{q_n\paren{\mb w}} + \frac{\norm{\mb w}{}^2}{q_n^3\paren{\mb w}} = \frac{1}{q_n^3\paren{\mb w}} \le 2\sqrt{2} n^{3/2},
\end{align*} 
we obtain that  
\begin{align*}
& \norm{\nabla_{\mb w}^2 g\paren{\mb w; \widetilde{\mb Y}} - \nabla_{\mb w}^2 g\paren{\mb w; \mb X_0}}{} \\
\le\; & \norm{\frac{1}{p}\sum_{k=1}^p \brac{\ddot{h}\paren{\mb q^*\paren{\mb w} \widetilde{\mb y}_k} \nabla_{\mb w}\mb q\paren{\mb w} \widetilde{\mb y}_k \widetilde{\mb y}_k^* \paren{\nabla_{\mb w} \mb q\paren{\mb w}}^* - \ddot{h}\paren{\mb q^*\paren{\mb w} (\mb x_0)_k} \nabla_{\mb w}\mb q\paren{\mb w} (\mb x_0)_k (\mb x_0)_k^* \paren{\nabla_{\mb w} \mb q\paren{\mb w}}^*}}{} \\
& \qquad + \norm{\frac{1}{p}\sum_{k=1}^p \brac{\dot{h}\paren{\mb q^*\paren{\mb w} \widetilde{\mb y}_k}\paren{\frac{\mb I}{q_n\paren{\mb w}} + \frac{\mb w \mb w^*}{q_n^3}} \widetilde{\mb y}_k\paren{n} - \dot{h}\paren{\mb q^*\paren{\mb w} (\mb x_0)_k}\paren{\frac{\mb I}{q_n\paren{\mb w}} + \frac{\mb w \mb w^*}{q_n^3}} (\mb x_0)_k\paren{n}}}{}\\
\le\; & \tfrac{45}{2}L_{\ddot{h}_{\mu}} n^{3/2} \norm{\mb X}{\infty}^3 \norm{\widetilde{\mb \Xi}}{} + \max_t \ddot{h}_{\mu}\paren{t} \paren{18n^{3/2} \norm{\mb X}{\infty}^2 \norm{\widetilde{\mb \Xi}}{}  + 10 n^2 \norm{\mb X}{\infty}^2\norm{\widetilde{\mb \Xi}}{}^2} \\
& \qquad + 3\sqrt{2} L_{\dot{h}_{\mu}} n^2 \norm{\widetilde{\mb \Xi}}{} \norm{\mb X}{\infty}^2 + \max_t \dot{h}\paren{t} 2\sqrt{2} n^2 \norm{\widetilde{\mb \Xi}}{} \norm{\mb X}{\infty}, 
\end{align*}
where $L_{\ddot{h}_{\mu}}$ denotes the Lipschitz constant for $\ddot{h}_{\mu}\paren{\cdot}$. Since  
\begin{align*}
\max_{t} \dot{h}_{\mu}\paren{t}  \le 1, & \quad \max_{t} \ddot{h}_{\mu}\paren{t} \le \frac{1}{\mu}, \quad L_{h_{\mu}} \le 1, \quad L_{\dot{h}_{\mu}} \le \frac{1}{\mu}, \quad L_{\ddot{h}_{\mu}} \le \frac{2}{\mu^2},
\end{align*}
and by Lemma~\ref{lem:X-infinty-tail-bound}, $\norm{\mb X}{\infty} \le 4\sqrt{\log\paren{np}}$ with probability at least $1-\theta \paren{np}^{-7} -\exp\paren{-0.3\theta np}$, we obtain 
\begin{align*}
\norm{\nabla_{\mb w} g\paren{\mb w; \widetilde{\mb Y}} - \nabla_{\mb w{}} g\paren{\mb w; \mb X} }{} & \le C_1\frac{n}{\mu} \log\paren{np} \norm{\widetilde{\mb \Xi}}{}, \\
\norm{\nabla_{\mb w}^2 g\paren{\mb w; \widetilde{\mb Y}} - \nabla_{\mb w}^2 g\paren{\mb w; \mb X}}{} & \le C_2 \max\set{\frac{n^{3/2}}{\mu^2}, \frac{n^2}{\mu}} \log^{3/2}\paren{np} \norm{\widetilde{\mb \Xi}}{}, 
\end{align*}
completing the proof. 
\end{proof}

\begin{proof}{\textbf{(of Theorem~\ref{thm:geometry_comp})}}
Here $c_\star$ is as defined in Theorem~\ref{thm:geometry_orth}. By Lemma~\ref{lem:pert_key_mag}, when 
\begin{align*} 
p \ge \frac{C_1}{c_\star^2 \theta} \max\set{\frac{n^4}{\mu^4}, \frac{n^5}{\mu^2}} \kappa^8\paren{\mb A_0} \log^4 \paren{\frac{\kappa\paren{\mb A_0}n}{\mu \theta}},
\end{align*}
the magnitude of the perturbation is bounded as 
\begin{align*}
\norm{\widetilde{\mb \Xi}}{} \le C_2 c_\star \theta \paren{\max\set{\frac{n^{3/2}}{\mu^2}, \frac{n^2}{\mu}} \log^{3/2}\paren{np}}^{-1}, 
\end{align*}
where $C_2$ can be made arbitrarily small by making $C_1$ large. 
Combining this result with Lemma~\ref{lem:pert_key_grad_hess}, we obtain that for all $\mb w \in \Gamma$, 
\begin{align*}
\norm{\nabla_{\mb w} g\paren{\mb w; \mb X_0 + \widetilde{\mb \Xi} \mb X_0} - \nabla_{\mb w{}} g\paren{\mb w; \mb X} }{} & \le \frac{c_\star \theta}{2} \nonumber \\
\norm{\nabla_{\mb w}^2 g\paren{\mb w; \mb X_0 + \widetilde{\mb \Xi} \mb X_0} - \nabla_{\mb w}^2 g\paren{\mb w; \mb X}}{} & \le \frac{c_\star \theta}{2}, 
\end{align*}
with probability at least $1-p^{-8} - \theta\paren{np}^{-7} - \exp\paren{-0.3\theta n p}$. In view of~\eqref{eqn:curvature-uni-comp} in Theorem~\ref{thm:geometry_orth}, we have 
\begin{align*}
\frac{\mb w^* \nabla^2_{\mb w} g\paren{\mb w; \mb X_0 + \widetilde{\mb \Xi} \mb X_0} \mb w}{\norm{\mb w}{}^2} 
& = \frac{\mb w^* \nabla^2_{\mb w} g\paren{\mb w; \mb X_0} \mb w}{\norm{\mb w}{}^2} + \frac{\mb w^* \nabla^2_{\mb w} g\paren{\mb w; \mb X_0 + \widetilde{\mb \Xi} \mb X_0} \mb w}{\norm{\mb w}{}^2} - \frac{\mb w^* \nabla^2_{\mb w} g\paren{\mb w; \mb X_0} \mb w}{\norm{\mb w}{}^2} \\
& \le - c_\star \theta + \norm{\nabla_{\mb w}^2 g\paren{\mb w; \mb X_0 + \widetilde{\mb \Xi} \mb X_0} - \nabla_{\mb w}^2 g\paren{\mb w; \mb X_0}}{} \le -\frac{1}{2} c_\star \theta. 
\end{align*}
By similar arguments, we obtain~\eqref{eqn:hess-zero-uni-comp} through~\eqref{eqn:curvature-uni-comp} in Theorem~\ref{thm:geometry_comp}. 

To show the unique local minimizer over $\Gamma$ is near $\mb 0$, we note that (recall the last part of proof of Theorem~\ref{thm:geometry_orth} in Section~\ref{sec:proof_geometry_orth}) $g\paren{\mb w; \mb X_0 + \widetilde{\mb \Xi} \mb X_0}$ being $\frac{c_\star \theta}{2\mu}$ strongly convex near $\mb 0$ implies that 
\begin{align*}
\norm{\mb w_\star}{} \le \frac{4\mu}{c_\star \theta} \norm{\nabla g\paren{\mb 0; \mb X_0 + \widetilde{\mb \Xi} \mb X_0}}{}. 
\end{align*}
The above perturbation analysis implies there exists $C_3 > 0$ such that when 
\begin{align*}
p \ge \frac{C_3}{c_\star^2 \theta} \max\set{\frac{n^4}{\mu^4}, \frac{n^5}{\mu^2}} \kappa^8\paren{\mb A_0} \log^4 \paren{\frac{\kappa\paren{\mb A_0}n}{\mu \theta}}, 
\end{align*}
it holds that 
\begin{align*}
\norm{\nabla_{\mb w} g\paren{\mb 0; \mb X_0 + \widetilde{\mb \Xi} \mb X_0} - \nabla_{\mb w{}} g\paren{\mb 0; \mb X} }{} & \le \frac{c_\star \theta}{400}, 
\end{align*}
which in turn implies 
\begin{align*}
\norm{\mb w_\star}{} \le \frac{4\mu}{c_\star \theta} \norm{\nabla g\paren{\mb 0; \mb X_0}}{} + \frac{4\mu}{c_\star \theta} \frac{c_\star \theta}{400} \le \frac{\mu}{8} + \frac{\mu}{100} < \frac{\mu}{7}, 
\end{align*}
where we have recall the result that $\frac{2\mu}{c_\star \theta} \norm{\nabla g\paren{\mb 0; \mb X_0}}{} \le \mu /16$ from proof of Theorem~\ref{thm:geometry_orth}. A simple union bound with careful bookkeeping gives the success probability. 
\end{proof}

\begin{appendices}
\section{Technical Tools and Basic Facts Used in Proofs}
In this section, we summarize some basic calculations that are useful throughout, and also record major technical tools we use in proofs. 
\begin{lemma}[Derivates and Lipschitz Properties of $h_{\mu}\paren{z}$] \label{lem:derivatives_basic_surrogate}
For the sparsity surrogate 
\begin{align*}
h_{\mu}\left(z\right)= \mu \log \cosh\paren{z/\mu}, 
\end{align*}
the first two derivatives are
\begin{align*}
\dot{ h}_\mu (z) = \tanh(z/\mu),\;  \ddot{h}_\mu (z) = \brac{1-\tanh^2(z/\mu)}/\mu. 
\end{align*}
Also, for any $z \ge 0$, 
\js{we have
\begin{align*}
\max\set{1-2\exp(-2z/\mu), 1/2-\exp(-2z/\mu)/2}\;& \le \;\tanh(z/\mu) \;\le\; 1-\exp(-2z/\mu), \\
2\exp(-2z/\mu) - \exp(-4z/\mu) \;&\le \;1-\tanh^2(z/\mu) \;\le\; 4\exp(-2z/\mu) - 4\exp(-4z/\mu). 
\end{align*}
}
Moreover, for any $z,~z^\prime\in \reals$, we have
\begin{align*}
|\dot{h}_{\mu}(z) - \dot{h}_{\mu}(z^\prime) | \le |z - z^\prime|/\mu,\; |\ddot{h}_{\mu}(z) - \ddot{h}_{\mu}(z^\prime) | \le 2|z - z^\prime|/\mu^2. 
\end{align*}
\end{lemma}


\begin{lemma}[Harris' Inequality, ~\cite{harris1960lower}, see also Theorem 2.15 of~\cite{boucheron2013concentration}] \label{lemma:harris_ineq}
Let $X_1, \dots, X_n$ be independent, real-valued random variables and $f, g: \R^n \mapsto \R$ be nonincreasing (nondecreasing) w.r.t. any one variable while fixing the others. Define a random vector $\mb X = \paren{X_1,\cdots,X_n}\in \R^n$, then we have 
\begin{align*}
\expect{f\left(\mb X\right) g\left(\mb X\right)} \geq \expect{f\left(\mb X\right)} \expect{g\left(\mb X\right)}. 
\end{align*}
Similarly, if $f$ is nondecreasing (nonincreasing) and $g$ is nonincreasing (nondecreasing) coordinatewise in the above sense, we have
\begin{align*}
\expect{f\left(\mb X\right) g\left(\mb X\right)} \leq \expect{f\left(\mb X\right)} \expect{g\left(\mb X\right)}. 
\end{align*}
\end{lemma}

\begin{lemma}[Differentiation under the Integral Sign] \label{lemma:exchange_diff_int}
Consider a function $F: \R^n \times \R \mapsto \R$ such that $\frac{\partial F\left(\mb x, s\right)}{\partial s}$ is well defined and measurable over $\mc U \times \left(0, t_0\right)$ for some open subset $\mc U \subset \R^n$ and some $t_0 > 0$. For any probability measure $\rho$ on $\R^n$ and any $t \in \left(0, t_0\right)$ such that $\int_{0}^t \int_{\mc U} \abs{\frac{\partial F\left(\mb x, s\right)}{\partial s}} \; \rho\left(d \mb x\right) ds < \infty$, it holds that 
\begin{align*}
\frac{d}{dt} \int_{\mc U} F\left(\mb x, t\right) \rho\left(d\mb x\right) = \int_{\mc U} \frac{\partial F\left(\mb x, t\right)}{\partial t} \rho \left(d\mb x\right), \; \text{or} \; \frac{d}{dt} \bb E_{\mb x} \left[F\left(\mb x, t\right) \indicator{\mc U}\right] = \bb E_{\mb x} \left[\frac{\partial F\left(\mb x, t\right)}{\partial t} \indicator{\mc U}\right]. 
\end{align*}
\end{lemma}
\begin{proof}
See proof of Lemma A.4 in the technical report~\cite{sun2015complete_tr}. 
\end{proof}

\begin{lemma}[Gaussian Tail Estimates] \label{lem:gaussian_tail_est}
Let $X \sim \mc N\paren{0, 1}$ and $\Phi\paren{x}$ be CDF of $X$. For any $x \ge 0$, we have the following estimates for $\Phi^c\paren{x} \doteq 1 - \Phi\paren{x}$: 
\begin{align*}
\paren{\frac{1}{x} - \frac{1}{x^3}}\frac{\exp\paren{-x^2/2}}{\sqrt{2\pi}} & \le \Phi^c\paren{x} \le \paren{\frac{1}{x} - \frac{1}{x^3} + \frac{3}{x^5}}\frac{\exp\paren{-x^2/2}}{\sqrt{2\pi}}, \quad (\text{Type I}) \\
\frac{x}{x^2 + 1} \frac{\exp\paren{-x^2/2}}{\sqrt{2\pi}}& \le \Phi^c\paren{x} \le \frac{1}{x} \frac{\exp\paren{-x^2/2}}{\sqrt{2\pi}},  \quad (\text{Type II}) \\
\frac{\sqrt{x^2 + 4} - x}{2} \frac{\exp\paren{-x^2/2}}{\sqrt{2\pi}} & \le  \Phi^c\paren{x} \le \paren{\sqrt{2 + x^2} - x}  \frac{\exp\paren{-x^2/2}}{\sqrt{2\pi}} \quad (\text{Type III}). 
\end{align*}
\end{lemma}
\begin{proof}
See proof of Lemma A.5 in the technical report~\cite{sun2015complete_tr}. 
\end{proof}

\begin{lemma}[Moments of the Gaussian RV] \label{lem:guassian_moment}
If $X \sim \mc N\left(0, \sigma^2\right)$, then it holds for all integer $p \geq 1$ that
\begin{align*}
\expect{\abs{X}^p} \leq \sigma^p \paren{p -1}!!. 
\end{align*}
\end{lemma}

\begin{lemma}[Moments of the $\chi^2$ RV] \label{lem:chi_sq_moment}
If $X \sim \mc \chi^2\paren{n}$, then it holds for all integer $p \geq 1$ that
\begin{align*}
\expect{X^p} = 2^p \frac{\Gamma\paren{p + n/2}}{\Gamma\paren{n/2}} =  \prod_{k=1}^p (n+2k-2) \le p!\paren{2n}^p/2. 
\end{align*}
\end{lemma}

\begin{lemma}[Moments of the $\chi$ RV] \label{lem:chi_moment}
If $X \sim \mc \chi\paren{n}$, then it holds for all integer $p \geq 1$ that
\begin{align*}
\expect{X^p} = 2^{p/2} \frac{\Gamma\paren{p/2 + n/2}}{\Gamma\paren{n/2}} \le p! n^{p/2}. 
\end{align*}
\end{lemma}

\begin{lemma}[Moment-Control Bernstein's Inequality for Scalar RVs, Theorem 2.10 of~\cite{foucart2013mathematical}] \label{lem:mc_bernstein_scalar}
Let $X_1, \dots, X_p$ be i.i.d. real-valued random variables. Suppose that there exist some positive numbers $R$ and $\sigma^2$ such that
\begin{align*}
\expect{\abs{X_k}^m} \leq m! \sigma^2 R^{m-2}/2, \; \; \text{for all integers}\; m \ge 2. 
\end{align*} 
Let $S \doteq \frac{1}{p}\sum_{k=1}^p X_k$, then for all $t > 0$, it holds  that 
\begin{align*}
\prob{\abs{S - \expect{S}} \ge t} \leq 2\exp\left(-\frac{pt^2}{2\sigma^2 + 2Rt}\right).   
\end{align*}
\end{lemma}

\begin{lemma}[Moment-Control Bernstein's Inequality for Matrix RVs, Theorem 6.2 of~\cite{tropp2012user}] \label{lem:mc_bernstein_matrix}
Let $\mb X_1, \dots, \mb X_p\in \R^{d \times d}$ be i.i.d. random, symmetric matrices. Suppose there exist some positive number $R$ and $\sigma^2$ such that
\begin{align*}
\expect{\mb X_k^m} \preceq m! \sigma^2 R^{m-2}/2 \cdot \mb I \; \text{and} -\expect{\mb X_k^m} \preceq m! \sigma^2 R^{m-2}/2 \cdot \mb I\;,  \; \text{for all integers $m \ge 2$}. 
\end{align*}
Let $\mb S \doteq \frac{1}{p} \sum_{k = 1}^p \mb X_k$, then for all $t > 0$, it holds that 
\begin{align*}
\prob{\norm{\mb S - \expect{\mb S}}{} \ge t} \le 2d\exp\paren{-\frac{pt^2}{2\sigma^2 + 2Rt}}.
\end{align*}
\end{lemma}
\begin{proof}
See proof of Lemma A.10 in the technical report~\cite{sun2015complete_tr}. 
\end{proof}

\begin{corollary}[Moment-Control Bernstein's Inequality for Vector RVs] \label{cor:vector-bernstein} 
Let $\mb x_1, \dots, \mb x_p \in \reals^d$ be i.i.d. random vectors. Suppose there exist some positive number $R$ and $\sigma^2$ such that
\begin{equation*}
\bb E\left[ \norm{\mb x_k }{}^m \right] \;\le\; m! \sigma^2R^{m-2}/2, \quad \text{for all integers $m \ge 2$}. 
\end{equation*}
Let $\mb s = \frac{1}{p}\sum_{k=1}^p \mb x_k$, then for any $t > 0$, it holds that
\begin{align*}
\bb P\brac{\norm{\mb s - \bb E\brac{\mb s}}{} \geq t} \; \leq \; 2(d+1)\exp\paren{-\frac{pt^2}{2\sigma^2+2Rt}}.
\end{align*}
\end{corollary}
\begin{proof}
See proof of Lemma A.11 in the technical report~\cite{sun2015complete_tr}. 
\end{proof}

\begin{lemma}[Integral Form of Taylor's Theorem]\label{lem:Taylor-integral-form}
	Let $f(\mb x): \bb R^n \mapsto \bb R$ be a twice continuously differentiable function, then for any direction $\mb y\in \bb R^n$, we have
	\begin{align*}
		f(\mb x+t\mb y) &= f(\mb x) + t \int_0^1\innerprod{\nabla f(\mb x+st\mb y)}{\mb y} \; ds, \\
		f(\mb x+t\mb y) &= f(\mb x) + t \innerprod{\nabla f(\mb x)}{\mb y} + t^2 \int_0^1 (1-s)\innerprod{\nabla^2 f(\mb x+st\mb y) \mb y}{\mb y}\; ds.
	\end{align*}
\end{lemma}

\section{Auxillary Results for Proofs} \label{sec:app_aux_results}
\begin{lemma}\label{lem:aux_asymp_proof_a}
Let $X\sim \mc N(0,\sigma_X^2)$ and $Y\sim \mc N(0,\sigma_Y^2)$ be independent random variables and $\Phi^c\paren{t} = \frac{1}{\sqrt{2\pi}} \int_t^{\infty} \exp\paren{-x^2/2}\; dx$ be the complementary cumulative distribution function of the standard normal. For any $a>0$, we have 
\begin{align}
\expect{X\indicator{X>0}} \;&= \;\frac{\sigma_X}{\sqrt{2\pi}}, \label{eqn:lem-aux-aymp-proof-a-1} \\
\expect{\exp\paren{-aX}X\indicator{X>0}} \;&= \; \frac{\sigma_X}{\sqrt{2\pi}} - a \sigma_X^2\exp\paren{\frac{a^2\sigma_X^2}{2}}\Phi^c\paren{a\sigma_X}, \label{eqn:lem-aux-aymp-proof-a-2} \\
\bb E\brac{\exp\paren{-aX}\indicator{X>0} }\;&=\; \exp\paren{\frac{a^2\sigma_X^2}{2}} \Phi^c\paren{a\sigma_X},  \label{eqn:lem-aux-aymp-proof-a-3}\\
\bb E\brac{\exp\paren{-a(X+Y)}X^2\indicator{X+Y>0}}\; &=\;\sigma_X^2\paren{1+a^2\sigma_X^2}\exp\paren{\frac{a^2\sigma_X^2 + a^2 \sigma_Y^2}{2}}\Phi^c\paren{a\sqrt{\sigma_X^2 + \sigma_Y^2}} \nonumber \\
& \qquad - \frac{a\sigma_X^4}{\sqrt{2\pi}\sqrt{\sigma_X^2 + \sigma_Y^2}}, \label{eqn:lem-aux-aymp-proof-a-4}\\
\bb E\brac{\exp\paren{-a(X+Y)}XY\indicator{X+Y>0} }\; &=\; a^2\sigma_X^2\sigma_Y^2 \exp\paren{\frac{a^2 \sigma_X^2 + a^2 \sigma_Y^2}{2}} \Phi^c\paren{a\sqrt{\sigma_X^2 + \sigma_Y^2}} \nonumber \\
& \qquad - \frac{a\sigma_X^2\sigma_Y^2}{\sqrt{2\pi}\sqrt{\sigma_X^2 + \sigma_Y^2}},  \label{eqn:lem-aux-aymp-proof-a-5} \\
\bb E\brac{\tanh\paren{aX}X} \;&=\; a\sigma_X^2 \bb E\brac{1-\tanh^2\paren{aX}},  \label{eqn:lem-aux-aymp-proof-a-6}\\
\bb E\brac{\tanh\paren{a(X+Y)}X} \;&=\; a\sigma_X^2 \bb E\brac{1-\tanh^2\paren{a(X+Y)}}\label{eqn:lem-aux-aymp-proof-a-7}. 
\end{align}
\end{lemma}

\begin{proof}
Equalities \eqref{eqn:lem-aux-aymp-proof-a-1}, \eqref{eqn:lem-aux-aymp-proof-a-2}, \eqref{eqn:lem-aux-aymp-proof-a-3}, \eqref{eqn:lem-aux-aymp-proof-a-4} and \eqref{eqn:lem-aux-aymp-proof-a-5} can be obtained by direct integrations. Equalities \eqref{eqn:lem-aux-aymp-proof-a-6} and \eqref{eqn:lem-aux-aymp-proof-a-7} can be derived using integration by part. 
\end{proof}

\begin{proof}{\textbf{(of Lemma~\ref{lem:neg_curvature_norm_bound})}} \label{proof:poly_approx_tanh}
Indeed $\frac{1}{\paren{1+\beta t}^2} = \sum_{k=0}^\infty (-1)^k(k+1)\beta^kt^k$, as 
\begin{align*}
\sum_{k=0}^\infty (-1)^k(k+1)\beta^kt^k = \sum_{k=0}^\infty(- \beta t)^k + \sum_{k=0}^\infty k (-\beta t)^k = \frac{1}{1+\beta t} + \frac{-\beta t}{(1+\beta t)^2} = \frac{1}{(1+\beta t)^2}. 
\end{align*}
The magnitude of the coefficient vector is 
\begin{align*}
\norm{\mb b}{\ell^1} &= \sum_{k=0}^\infty \beta^k (1+k) = \sum_{k=0}^\infty \beta^k + \sum_{k=0}^\infty k \beta^k = \frac{1}{1-\beta} +\frac{\beta}{(1-\beta)^2} = \frac{1}{(1-\beta)^2} = T. 
\end{align*}
Observing that $\frac{1}{\paren{1+\beta t}^2} > \frac{1}{\paren{1+t}^2}$ for $t \in \brac{0, 1}$ when $0 < \beta < 1$, we obtain 
\begin{align*}
\norm{p -f}{L^1[0,1]} &= \int_0^1 \abs{p(t) - f(t)}dt = \int_0^1 \brac{\frac{1}{(1+\beta t)^2} -\frac{1}{(1+t)^2}} dt=\frac{1-\beta}{2(1+\beta)}  \le \frac{1}{2\sqrt{T}}. 
\end{align*}
Moreover, we have 
\begin{align*}
\norm{f-p}{L^\infty[0,1]} = \max_{t\in[0,1]} p(t) - f(t) = \max_{t\in[0,1]} \frac{t(1-\beta)\paren{2+t(1+\beta)}}{(1+t)^2(1+\beta t)^2} \le 1-\beta = \frac{1}{\sqrt{T}}.  
\end{align*}
Finally, notice that
\begin{align*}
\sum_{k=0}^\infty \frac{b_k}{(1+k)^3} = \sum_{k=0}^\infty \frac{\paren{-\beta}^k}{(1+k)^2} 
& = \sum_{i=0}^\infty\brac{\frac{\beta^{2i}}{(1+2i)^2} - \frac{\beta^{2i+1}}{(2i+2)^2} } \nonumber \\
& = \sum_{i=0}^\infty \beta^{2i} \frac{(2i+2)^2 -\beta(2i+1)^2 }{(2i+2)^2(2i+1)^2} > 0, 
\end{align*}
where at the second equality we have grouped consecutive even-odd pair of summands. In addition, we have
\begin{align*}
\sum_{k=0}^n \frac{b_k}{(1+k)^3}\le \sum_{k=0}^n \frac{\abs{b_k}}{(1+k)^3} = \sum_{k=0}^n\frac{\beta^k}{(1+k)^2} \le 1+ \sum_{k=1}^n \frac{1}{(1+k)k} = 2-\frac{1}{n+1}, 
\end{align*}
which converges to $2$ when $n \to \infty$, completing the proof. 
\end{proof}

\begin{proof}{\textbf{(of Lemma~\ref{lem:U-moments-bound})}} \label{proof:ng_g_comparison}
The first inequality is obviously true for $\mb v = \mb 0$. When $\mb v \neq \mb 0$, we have 
\begin{align*}
\expect{\abs{\mb v^* \mb z}^m } 
& =  \sum_{\ell = 0}^n \theta^\ell \paren{1-\theta}^{n - \ell} \sum_{\mc J \in \binom{[n]}{\ell}} \bb E_{Z \sim \mc N\paren{0, \norm{\mb v_{\mc J}}{}^2}}\brac{\abs{Z}^m} \\
& \le \sum_{\ell = 0}^n \theta^\ell \paren{1-\theta}^{n - \ell} \sum_{\mc J \in \binom{[n]}{\ell}} \bb E_{Z \sim \mc N\paren{0, \norm{\mb v}{}^2}}\brac{\abs{Z}^m}\\
& = \bb E_{Z \sim \mc N\paren{0, \norm{\mb v}{}^2}}\brac{\abs{Z}^m} \sum_{\ell = 0}^n \theta^\ell \paren{1-\theta}^{n - \ell} \binom{n}{\ell} \\
& = \bb E_{Z \sim \mc N\paren{0, \norm{\mb v}{}^2}}\brac{\abs{Z}^m}, 
\end{align*}
where the second line relies on the fact $\norm{\mb v_{\mc J}}{} \le \norm{\mb v}{}$ and that for a fixed order, central moment of Gaussian is monotonically increasing w.r.t. its variance. Similarly, to see the second inequality, 
\begin{align*}
\expect{\norm{\mb z}{}^m}
& = \sum_{\ell = 0}^n \theta^\ell \paren{1 -\theta}^{n-\ell} \sum_{\mc J \in \binom{[n]}{\ell}} \expect{\norm{\mb z'_{\mc J}}{}^m} \\
& \le \expect{\norm{\mb z'}{}^m}  \sum_{\ell = 0}^n \theta^\ell \paren{1 -\theta}^{n-\ell} \binom{n}{\ell} = \expect{\norm{\mb z'}{}^m}, 
\end{align*}
as desired. 
\end{proof}

\begin{proof}{\textbf{(of Lemma~\ref{lem:X-infinty-tail-bound})}} \label{proof:X_inf_tail}
Consider one component of $\mb X$, i.e., $X_{ij}=B_{ij}V_{ij}$ for $i \in [n]$ and $j \in [p]$, where $B_{ij}\sim \mathrm{Ber}\paren{\theta}$) and $V_{ij}\sim \mc N(0,1)$. We have 
\begin{align*}
\bb P\brac{\abs{X_{ij}}> 4\sqrt{\log\paren{np}}} \leq \theta \bb P\brac{\abs{V_{ij}}> 4\sqrt{\log(np)}} \leq \theta \exp\paren{-8\log(np)}= \theta (np)^{-8}. 
\end{align*}
And also 
\begin{align*}
\prob{\abs{X_{ij}} < 1} = 1 -\theta + \theta \prob{\abs{V_{ij}} < 1} \le 1-0.3\theta. 
\end{align*}
Applying a union bound as  
\begin{align*}
\prob{\norm{\mb X}{\infty} \le 1 \; \text{or} \; \norm{\mb X}{\infty} \ge 4\sqrt{\log\paren{np}}} \le \paren{1-0.3\theta}^{np} + np\theta \paren{np}^{-8} \le \exp\paren{-0.3\theta n p} + \theta \paren{np}^{-7}, 
\end{align*}
we complete the proof. 
\end{proof}

\begin{lemma} \label{lem:half_inverse_pert}
Suppose $\mb A \succ \mb 0$. Then for any symmetric perturbation matrix $\mb \Delta$ with $\norm{\mb \Delta}{} \le \tfrac{\sigma_{\min}\paren{\mb A}}{2}$, it holds that 
\begin{align}
\norm{\paren{\mb A + \mb \Delta}^{-1/2} - \mb A^{-1/2}}{} \le \frac{2\norm{\mb A}{}^{1/2} \norm{\mb \Delta}{} }{\sigma_{\min}^2\paren{\mb A}}. 
\end{align}
\end{lemma}
\begin{proof}
See proof of Lemma B.2 in the technical report~\cite{sun2015complete_tr}. 
\end{proof}

\begin{lemma} \label{lem:bg_identity_diff}
For any $\theta \in \paren{0, 1/2}$, $\mb X \in \R^{n_1 \times n_2}$ with $\mb X\sim_{i.i.d.} \mathrm{BG}\paren{\theta}$ obeys
\begin{align}
\norm{\frac{1}{n_2 \theta} \mb X \mb X^* - \mb I}{} \le 10\sqrt{\frac{\theta n_1 \log n_2}{n_2}}
\end{align}
with probability at least $1 - n_2^{-8}$, provided $n_2 > C n_1^2 \log n_1$. Here $C > 0$ is a constant. 
\end{lemma}
\begin{proof}
Observe that $\expect{\tfrac{1}{\theta}\mb x_k \mb x_k^*} = \mb I$ for any column $\mb x_k$ of $\mb X$ and so $\tfrac{1}{n_2 \theta} \mb X \mb X^*$ can be considered as a normalize sum of independent random matrices. Moreover, for any integer $m \ge 2$, 
\begin{align*}
\expect{\paren{\frac{1}{\theta} \mb x_k \mb x_k^*}^m} = \frac{1}{\theta^m} \expect{\norm{\mb x_k}{}^{2m-2} \mb x_k \mb x_k^*}. 
\end{align*}
Now $\expect{\norm{\mb x_k}{}^{2m-2} \mb x_k \mb x_k^*}$ is a diagonal matrix (as $\expect{\norm{\mb x_k}{}^2 x_k\paren{i} x_k\paren{j}} = - \expect{\norm{\mb x_k}{}^2 x_k\paren{i} x_k\paren{j}}$ for any $i \neq j$ by symmetry of the distribution) in the form $\expect{\norm{\mb x_k}{}^{2m-2} \mb x_k \mb x_k^*} = \expect{\norm{\mb x}{}^{2m-2} x(1)^2}\mb I$ for $\mb x \sim_{i.i.d.} \mathrm{BG}\paren{\theta}$ with $\mb x \in \R^{n_1}$. Let $t^2\paren{\mb x} = \norm{\mb x}{}^2 - x(1)^2$. Then if $m = 2$, 
\begin{align*}
\expect{\norm{\mb x}{}^2 x(1)^2} 
& = \expect{x(1)^4} + \expect{t^2\paren{\mb x}} \expect{x(1)^2} \\
& = \expect{x(1)^4} + \paren{n_1-1} \paren{\expect{x(1)^2}}^2 = 3\theta + \paren{n_1-1} \theta^2 \le 3n_1 \theta, 
\end{align*}
where for the last simplification we use the assumption $\theta \le 1/2$. For $m \ge 3$, 

\begin{align*}
\expect{\norm{\mb x}{}^{2m-2} x(1)^2} 
& = \sum_{k=0}^{m-1} \binom{m-1}{k} \expect{t^{2k}\paren{\mb x} x(1)^{2m-2k}} =  \sum_{k=0}^{m-1} \binom{m-1}{k}  \expect{t^{2k}\paren{\mb x}} \expect{x(1)^{2m-2k}} \\
& \le \sum_{k=0}^{m-1} \binom{m-1}{k} \bb E_{Z \sim \chi^2\paren{n_1 -1}}\brac{Z^k} \theta \bb E_{W \sim \mc N\paren{0, 1}}\brac{W^{2m-2k}} \\
& \le \theta \sum_{k=0}^{m-1} \binom{m-1}{k} \frac{k!}{2} \paren{2n_1 - 2}^k \paren{2m-2k}!! \\
& \le \theta 2^m \frac{m!}{2} \sum_{k=0}^{m-1} \binom{m-1}{k} \paren{n_1-1}^k \\
& \le \frac{m!}{2} n_1^{m-1} 2^{m-1}, 
\end{align*}
where we have used the moment estimates for Gaussian and $\chi^2$ random variables from Lemma~\ref{lem:guassian_moment} and Lemma~\ref{lem:chi_sq_moment}, and also $\theta \le 1/2$. Taking $\sigma^2 = 3n_1 \theta$ and $R = 2n_1$, and invoking the matrix Bernstein in Lemma~\ref{lem:mc_bernstein_matrix}, we obtain 
\begin{align}
\expect{\norm{\frac{1}{p\theta} \sum_{k=1}^p \mb x_k \mb x_k^* - \mb I}{} > t} \le \exp\paren{-\frac{n_2 t^2}{6n_1 \theta + 4n_1 t} + 2\log n_1}
\end{align}
for any $t \ge 0$. Taking $t = 10\sqrt{\theta n_1 \log\paren{n_2}/n_2}$ gives the claimed result. 
\end{proof}

\end{appendices}

\section*{Acknowledgment}
We thank Dr. Boaz Barak for pointing out an inaccurate comment made on overcomplete dictionary learning using SOS. We thank Cun Mu and Henry Kuo of Columbia University for discussions related to this project. We also thank the anonymous reviewers for their careful reading of the paper, and for comments which have helped us to substantially improve the presentation. JS thanks the Wei Family Private Foundation for their generous support. This work was partially supported by grants ONR N00014-13-1-0492, NSF 1343282, NSF CCF 1527809, NSF IIS 1546411, and funding from the Moore and Sloan Foundations.

\bibliographystyle{IEEEtran}
\bibliography{IEEEabrv,../dl_tit,../ncvx}








\end{document}